\documentclass[11pt]{article}
\usepackage{geometry}
\usepackage{amsmath,mathtools}
\usepackage{amsfonts}
\usepackage{bm}
\usepackage{amsthm}
\usepackage{enumerate}
\usepackage{esint}
\usepackage{color}
\usepackage [english]{babel}
\usepackage [autostyle, english = american]{csquotes}
\usepackage[T1]{fontenc}
\usepackage[utf8]{inputenc}
\usepackage{authblk}
\usepackage{graphicx}
\usepackage{subcaption}
\numberwithin{equation}{section}
\usepackage{titlesec}
\usepackage{tocloft}

\titleformat*{\section}{\centering \large\bfseries}
\titleformat*{\subsection}{\centering \large\itshape}

\MakeOuterQuote{"}
\DeclareMathOperator{\Tr}{Tr}
\DeclareMathOperator{\adj}{adj}

\DeclareMathOperator{\rank}{rank}

\DeclareMathOperator{\dist}{dist}

\DeclareMathOperator{\cof}{cof}

\DeclareMathOperator{\II}{II}

\DeclareMathOperator{\spt}{spt}

\DeclareMathOperator{\spn}{span}
\DeclareMathOperator{\co}{co}

\allowdisplaybreaks

\newcommand{\be}{\begin}
\newcommand{\e}{\end}
\newcommand{\beq}{\begin{equation}}
\newcommand{\eeq}{\end{equation}}
\newcommand{\beqs}{\begin{equation*}}
\newcommand{\eeqs}{\end{equation*}}
\newcommand{\bal}{\begin{align}}
\newcommand{\eal}{\end{align}}
\newcommand{\bals}{\begin{align*}}
\newcommand{\eals}{\end{align*}}

\renewcommand{\l}{\left}
\renewcommand{\r}{\right}

\newcommand{\set}[1]{\mathbb{#1}}

\newcommand{\R}{\set{R}}

\renewcommand{\S}{\set{S}}

\newcommand{\om}{\omega}
\newcommand{\Om}{\Omega}

\newcommand{\de}{\delta}

\newtheorem{theorem}{Theorem}[section]
\newtheorem{corollary}[theorem]{Corollary}
\newtheorem{lemma}[theorem]{Lemma}
\newtheorem{proposition}[theorem]{Proposition}

\theoremstyle{definition}

\newtheorem{definition}[theorem]{Definition}
\newtheorem{remark}[theorem]{Remark}

\title{Actuation of thin nematic elastomer sheets with controlled heterogeneity}
\author{Paul Plucinsky} 
\author{Marius Lemm}
\author{Kaushik Bhattacharya}
\affil{California Institute of Technology, Pasadena, California 91125, USA}

\date{\today}
\begin{document}
\maketitle

\begin{abstract}
Nematic elastomers and glasses deform spontaneously when subjected to temperature changes. This property can be exploited in the design of heterogeneously patterned thin sheets that deform into a non-trivial shape when heated or cooled. In this paper, we start from a variational formulation for the entropic elastic energy of liquid crystal elastomers and we derive an effective two-dimensional metric constraint, which links the deformation and the heterogeneous director field. Our main results show that satisfying the metric constraint is both necessary and sufficient for the deformation to be an approximate minimizer of the energy. We include several examples which show that the class of deformations satisfying the metric constraint is quite rich.
\end{abstract}
\addtocontents{toc}{\protect\setcounter{tocdepth}{2}}
\tableofcontents

\section{Introduction and main results}

\subsection{Actuation in heterogeneously patterned nematic elastomers}
Nematic elastomers are rubbery solids made of cross-linked polymer chains which have liquid crystals (rod-like molecules) either incorporated into the main chain or pendent from them.  Their structure enables a coupling between the entropic (mechanical) elasticity of the polymer network and the ordering of liquid crystals.  This underlies the  dramatic shape changing response to temperature change in these elastomers \cite{btw_pre_93,detal_12_ACIE, mbw_prsa_11, mw_phystoday_16,tt_01_LCE,wetal_science_15,wt_lceboox_03}.  At low temperatures, the rod-like liquid crystals within the solid tend to align themselves, giving rise to a local {\it nematic} orientational order described by a director (a unit vector on $\mathbb{R}^3$).  As the temperature is increased, thermal fluctuations thwart the attempt to order, driving a nematic to isotropic transition in the solid whereby the liquid crystals become randomly oriented. Due to the intrinsic coupling of the liquid crystals to the soft polymer network, the solid distorts to accommodate this temperature driven transition\textemdash typically by a large spontaneous contraction along the director and expansion transverse to it.

Bladon et al.\ \cite{btw_pre_93} proposed a free energy based on entropic elasticity of the chains in the presence of nematic order to describe the elasticity of nematic elastomers.  This has been used by Tajbakhsh and Terentjev \cite{tt_01_LCE} to explain the shape changing response to thermal actuation in mono-domain sheets, i.e., sheets with spatially uniform director. (Also, many other features inherent to these elastomers have been explained with this theory; see Warner and Terentjev \cite{wt_lceboox_03} for a comprehensive introduction and review.)    More recently, it has been recognized by Modes et al.\ \cite{mbw_prsa_11} and others \cite{ask_prl_14,m_pre_15,metal_16_pra,PLB} that   heterogeneously programing the sheet, so that the director varies spatially in the plane of the sheet, could result in complex three dimensional shape upon thermal actuation.  That is, since sheets are characteristically thin compared to their lateral dimensions, the non-uniform shape changing response of patterned sheets to thermal actuation could induce out-of-plane {\it buckling}.  Indeed, based on a membrane idealization of the free energy, Modes et al.\ \cite{mbw_prsa_11} predicted that a sheet with programmed azimuthal or radial heterogeneity would actuate into a conical or saddle-like three-dimensional shape upon temperature change.  Such heterogeneity was later realized experimentally\textemdash first by de Haan et al.\ \cite{detal_12_ACIE} for nematic glass sheets, and then by Ware et al.\ \cite{wetal_science_15} in nematic elastomer sheets\textemdash and the actuation of these sheets agreed with the prediction.   Since then, a range of Gaussian curvature has been explored theoretically and achieved experimentally \cite{m_pre_15,metal_16_pra}, and following the formalism of non-Euclidean plate theory (i.e., Efrati et al.\ \cite{esk_sm_13}), a metric constraint was proposed by Aharoni et al.\ \cite{ask_prl_14} to govern shape changing actuation in these sheets.  All these results suggest an intimate connection between the microscopic physics of nematic elastomers and the geometry of a thin sheet.  However, to our knowledge, this has not yet been illuminated with mathematical precision and rigor.  Our work here addresses this point.

We start from a variational formulation for the entropic elastic energy of nematic elastomers and we derive the effective two dimensional metric constraint, which links the deformation and the heterogeneous director field.  This constraint (equation (\ref{eq:2DMetric}) below) arises in the context of energy minimization due the interplay of stretching, bending and heterogeneity in these sheets.  It is also a generalization of the constraint proposed by Aharoni et al.\ \cite{ask_prl_14} in two directions in that (i) it extends the constraint to three dimensional programming of the director field (where the director can tilt out of the plane of the sheet) and (ii) it relaxes the smoothness requirement asserted there.  These generalizations admit a rich class of examples under the metric constraint. 

This metric constraint first appeared in our earlier short paper \cite{PLB} with a view towards applications. 

\subsection{The model and the metric constraint}

We consider a thin sheet of nematic elastomer of thickness $h\ll 1$. Initially, the sheet occupies a flat region in space, 
$$
\Om_h:=\om\times (-h/2,h/2),\qquad \om\subset\R^2,
$$
where $\om$ is an open, connected and bounded Lipschitz domain which we call the midplane of the sheet.
We envision that the elastomer sheet is patterned heterogeneously by a director field $N_0^h \colon \Omega_h \rightarrow \mathbb{S}^2$ at the initial temperature $T_0$.  Upon changing the temperature from $T_0$ to the final temperature $T_f$, the sheet will spontaneously deform by a deformation $Y^h:\Om_h\to \R^3$ which we assume minimizes the entropic elastic energy 
\begin{align}\label{eq:StrainEnergy}
\mathcal{I}^h_{N^h_0}(Y^h) := \int_{\Omega_h} W^e\l(\nabla Y^h, \frac{(\nabla Y^h) N^h_0}{|(\nabla Y^h) N^h_0|},N^h_0\r) dx.
\end{align}
Following Bladon et al.\ \cite{btw_pre_93} (see also Warner and Terentjev \cite{wt_lceboox_03}), we take the \emph{entropic elastic energy density} $W^e:\mathbb{R}^{3 \times 3} \times \mathbb{R}^3 \times \mathbb{R}^3 \rightarrow \mathbb{R} \cup \{ +\infty\}$ as
\begin{align}\label{eq:We} 
W^e(F,\nu,\nu_0) :=& 
\frac{\mu}{2}
\begin{dcases}
\Tr\Big(F^T (\ell_{\nu}^f)^{-1} F (\ell_{\nu_0}^0)\Big) - 3 &\text{ if } \det F = 1, \text{ and } \nu,\nu_0 \in \mathbb{S}^2 \\
+\infty &\text{ otherwise}.
\end{dcases}
\end{align}
Here, $\mu>0$ is the shear modulus, $F$ is the deformation gradient, and $F^T$ denotes the transpose matrix of $F$. Moreover, $\ell_{\nu_0}^{0}, \ell_{\nu}^f \in \mathbb{R}^{3 \times 3}_{sym}$ are the step length tensors at the initial temperature $T_0$ and final temperature $T_f$ respectively. They are defined by
\begin{equation}
\begin{aligned}\label{eq:StepLength}
\ell_{\nu_0}^0 :=& r_0^{-1/3} \left( I_{3\times3} + (r_0 - 1) \nu_0 \otimes \nu_0\right),\\
\ell_\nu^f :=& r_f^{-1/3} \left( I_{3\times3} + (r_f - 1) \nu \otimes \nu \right).
\end{aligned}
\end{equation}
The parameters $r_0,r_f\geq 1$ quantify the degree of anisotropy at the initial and final temperature respectively. They describe the extent to which the material tends to spontaneously deform in the directions $\nu_0$ and $\nu$ respectively.

\be{remark}\label{RemarkModel}
\begin{enumerate}[(i)]
\item The energy density (\ref{eq:We}) is a purely entropic Helmholtz free energy density which captures the configurational entropy of polymer chains in the presence of nematic order \cite{btw_pre_93,wt_lceboox_03}.  It has been used to explain many novel features inherent  to nematic elastomers including soft elasticity, material microstructure, and the dramatic shape changing response to temperature change \cite{cdd_pre_02,cdd_jmps_02,dd_arma_02,mbw_prsa_11,tt_01_LCE}.  The constant $-3$ in this energy density is chosen so that $\min W^e=0$.
\item  The elastic energy $\mathcal{I}^h_{N^h_0}$ is defined without any displacement or traction boundary conditions as we are dealing with actuation only.  
\item  We envision that $r_0,r_f$ arise from evaluating some underlying monotone decreasing function $\bar r(T) \geq 1$ (which is equal to 1 above a critical temperature) at the temperatures $T=T_0$ and $T=T_f$.  Note that in setting $r_0=r_f=1$ in the formula above, one recovers the standard incompressible neo-Hookean energy for isotropic materials.
\item
In the definition \eqref{eq:StrainEnergy} of $\mathcal{I}^h_{N^h_0}(Y^h)$, we imposed the kinematic constraint $N^h=\frac{(\nabla Y^h) N^h_0}{|(\nabla Y^h) N^h_0|}$.   The constraint is similar to one that was imposed by Modes et al.\ \cite{mbw_prsa_11} in their prediction for conical and saddle like actuation in nematic glass sheets with radial and azimuthal heterogeneity (in fact, both constraints are equivalent for zero energy/stress free states; see Proposition \ref{EquivalenceProp}). 

There are nematic elastomers which do not satisfy this kinematic constraint (i.e., where the director $N^h$ is allowed to vary more freely). Those materials can show macroscopic deformations which arise from the fine-scale microstructure produced by oscillations of $N^h$ \cite{cpk_arma_15,cdd_pre_02,cdd_jmps_02,dd_arma_02} (see also the experiments by Kundler and Finkelmann \cite{kf_mrc_95}). 

In the present paper, we are interested in actuating complex, yet predictable, shape by programming an initial heterogeneous anisotropy $N_0^h$ in the nematic elastomer. It would be difficult to control actuation for a material that is capable of freely forming microstructure which competes with the shape change driven by the programmed anisotropy, even at low energy. For simplicity, we have chosen the hard kinematic constraint $N^h=\frac{(\nabla y^h) N^h_0}{|(\nabla Y^h) N^h_0|}$ here in order to exclude the free formation of microstructure.  The results that we prove for this energy (i.e.,  $\mathcal{I}_{N_0^h}^h$ with this kinematic constraint) can also be proven for a more realistic energy in which the sharp constraint is replaced by a non-ideal energy contribution penalizing deviations from the constraint.  In fact, we use this more realistic model when deriving the metric constraint as a necessary condition; this is discussed in Section \ref{ssec:metricNecessary}.

\item We have neglected Frank elasticity (an elasticity thought to play a critical role in the behavior of liquid crystal {\it fluids}, see for example de Gennes and Prost \cite{dgp_book_95}) and related effects in our model, as these are expected to be small in comparison to the entropic elasticity (see discussion in Chapter 3 in Warner and Tarentjev \cite{wt_lceboox_03}). However, to derive the key metric constraint (introduced below) as a necessary feature of low energy deformations, we add a small contribution from Frank elasticity for technical reasons. This is also discussed in Section \ref{ssec:metricNecessary}.
\end{enumerate}
\e{remark}

Our goal is to characterize designable actuation in nematic elastomer sheets.  By this, we mean a classification of the director fields $N_0^h$ and corresponding deformations $Y^h$ which yield small elastic energy $\mathcal{I}_{N_0^h}^h(Y^h)$ under the assumption of a desired director field design $n_0 \colon \omega \rightarrow \mathbb{S}^2$ (i.e., varying only in $(x_1,x_2)$). To be precise:

\be{assumption}\label{ass:N}We assume
\begin{equation}
\begin{aligned}\label{eq:n0tomidn0}
&N^h_0(x_1, x_2, x_3) = n_0(x_1, x_2) + O(h) ,\quad \text{for a.e. } (x_1, x_2 ,x_3)  \in \Omega_h, \\
&\text{i.e., } \quad \|N_0^h - n_0\|_{L^{\infty}(\Omega_h)} \leq \tau h \text{ for some $\tau>0$}.
\end{aligned}
\end{equation}
\end{assumption}
The $O(h)$ term accounts for the following two possible deviations from the desired design. 
For definiteness, we have fixed the maximum tolerance $\tau >0$ for these non-idealities.
\be{enumerate}
\item[(a)] The assumption accounts for deviations of the director field through the thickness which are of the same order as the thickness. Note that this excludes 
twisted or splay-bend nematic sheets \cite{fwbrvwj_sm_2015,wmc_prsa_10}, for which one prescribes the director field on the top surface of the sheet and then differently then on the bottom surface, so that the director field has to vary by an $O(1)$ amount through the thickness.

\item[(b)] The assumption also accounts for the possibility of planar deviations.  In the synthesis techniques employed by Ware et al.\ \cite{wetal_science_15}, the director field is prescribed in voxels or cubes whose characteristic length is similar to the thickness and we expect the experimental error to be of this order.
\e{enumerate}

\noindent Under Assumption \ref{ass:N}, the characterization of designable actuation comes in the form of a two-dimensional effective metric constraint (\ref{eq:2DMetric}).  The intuition is expressed in Figure \ref{fig:SketchArg}.

\begin{figure}[t!]
\centering
\includegraphics[width=12 cm]{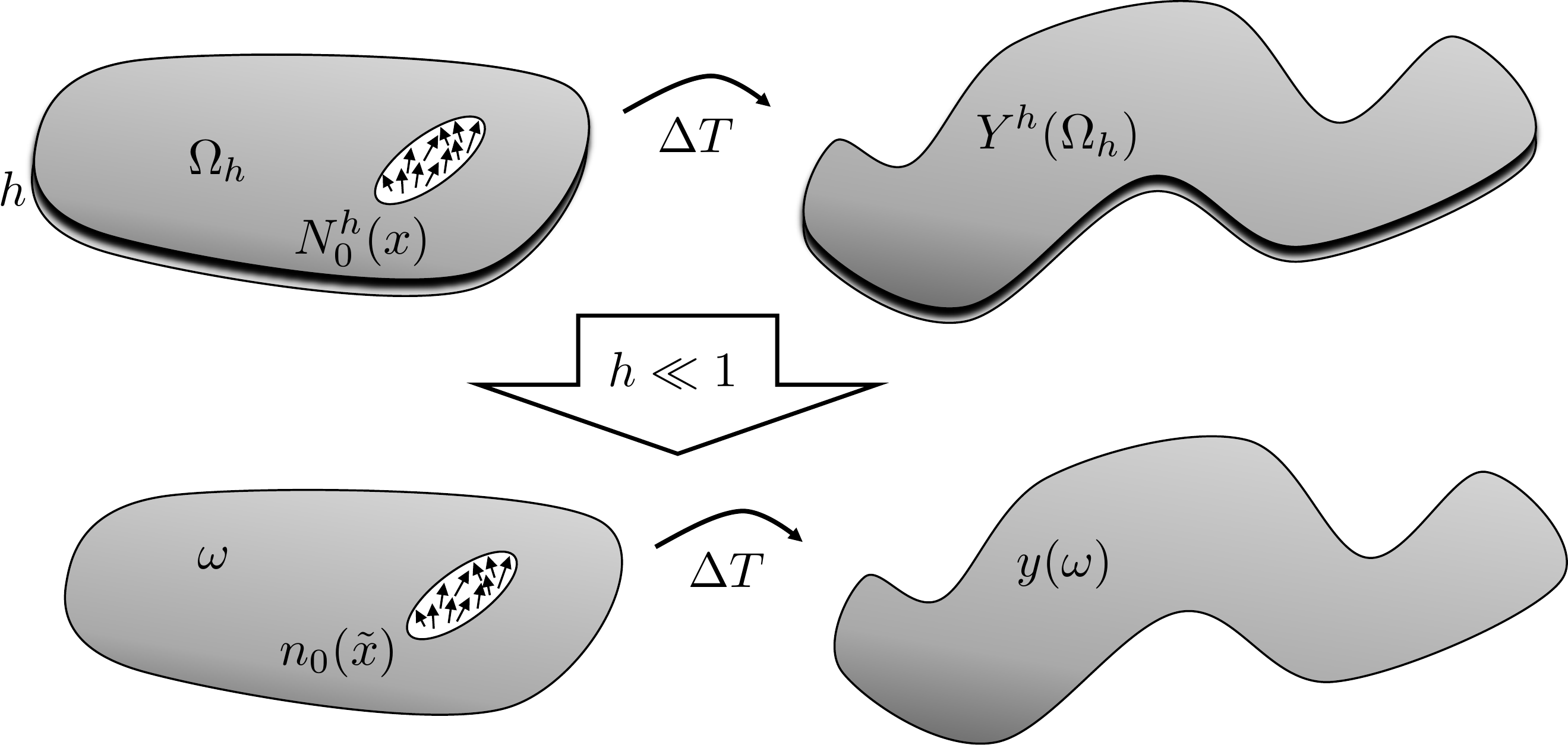} 
\caption{Actuation for thin sheets is characterized by the midplane fields.}
\label{fig:SketchArg}
\end{figure}  

To see how the metric constraint arises, we first consider a naive approach by requiring $\mathcal{I}^h_{N_0^h}(Y^h) = 0$ (recall that $\min W^e = 0$).  By Proposition \ref{EquivalenceProp}, $\mathcal{I}^h_{N^h_0}(Y^h)=0$ is equivalent to 
\begin{align}\label{eq:3DMetric}
(\nabla Y^h)^T \nabla Y^h = r^{-1/3}(I_{3\times3} + (r-1) N^h_0 \otimes N^h_0 ) = : \ell_{N^h_0} \quad \text{ a.e. on } \Omega_h,
\end{align}
where $r = r_f/r_0$ so that $r \in (0,1)$ for heating and $r> 1$ for cooling. However, \eqref{eq:3DMetric} is too strong of a condition to be useful, meaning that there are only few choices of $N_0^h$ for which a $Y^h$ satisfying \eqref{eq:3DMetric} exists. 
\begin{remark}\label{rmk:3D}
Assuming that $N_0^h$ is sufficiently smooth, there exists a $Y^h$ satisfying (\ref{eq:3DMetric}) if and only if the components of the Riemann curvature tensor of $\ell_{N^h_0}$ vanish. This condition is well-known in the physics literature (e.g., Efrati et al.\ \cite{esk_sm_13}), and in the language of continuum mechanics, it gives compatibility of the Right Cauchy-Green deformation tensor (e.g., Blume \cite{b_jelast_89}). As a consequence, $N_0^h$ has to satisfy a certain nonlinear partial differential equation, and so it must come from a very restricted set of functions.  The non-smooth case is treated in Lewicka and Pakzad \cite{lp_esaim_11}, and it is similar.
\e{remark}

Given that \eqref{eq:3DMetric} is too restrictive, we relax the problem and study \emph{approximate minimizers of the elastic energy} $\mathcal{I}^h_{N^h_0}(Y^h)$. The key observation is that by making use of the thinness of the sheet $\Om_h$ and the assumption that $N_0^h$ does not vary too much in $x_3$, we show that approximate minimizers are characterized (in a sense to be made precise) by the following \emph{effective metric constraint} \eqref{eq:2DMetric}. It is a two-dimensional reduction of the three-dimensional constraint \eqref{eq:3DMetric} and reads 
\begin{align}\label{eq:2DMetric}
(\tilde{\nabla} y)^T \tilde{\nabla} y = r^{-1/3}( I_{2\times 2} + (r-1)\tilde{n}_0 \otimes \tilde{n}_0) =: \tilde{\ell}_{n_0} \quad \text{ a.e. on } \omega.
\end{align}
 
\paragraph{Notation.}
Here and throughout, we denote vector fields which are mappings $\Om_h\to \R^3$ by capital letters (e.g., $N_0^h, Y^h$) and vector fields defined on the midplane $\om\subset\R^2$ by lowercase letters (e.g., $n_0,y:\om \to \R^3$). Moreover, we use $\tilde{(\cdot)}$ to distinguish two dimensional quantities from  three-dimensional quantities. For instance, 
$$
x:=(x_1,x_2,x_3),\qquad \tilde x:=(x_1,x_2),\qquad \nabla = (\partial_1,\partial_2, \partial_3), \qquad \tilde \nabla = (\partial_{1},\partial_{2}),
$$
and $\tilde{n}_0 \in B_1(0) \subset \mathbb{R}^2$ is the projection of $n_0$ onto $\omega$.

\begin{remark}
\begin{enumerate}[(i)]
\item
If there exists a deformation $y$ which satisfies (\ref{eq:2DMetric}) for a given $n_0$, then there may be, in general, multiple such deformations (e.g., the sheet can actuate upward or downwards in different places). We imagine that one can distinguish between these by appropriately breaking additional symmetries, but we do not investigate this further.
\item
The constraint \eqref{eq:2DMetric} generalizes a metric constraint that has been proposed by Aharoni et al.\ \cite{ask_prl_14} for actuation of nematic sheets. Indeed, \eqref{eq:2DMetric} is more general in that (a) it need only hold almost everywhere, allowing for piecewise constant director designs and (b) the director can be programmed out of plane. At the same time, it is easy to see that \eqref{eq:2DMetric} reduces to the constraint \cite{ask_prl_14}  for smooth planar director fields. (With $n_0 \equiv \tilde{n}_0$, we can write $\tilde{n}_0 \cdot \tilde{e}_1 = \cos(\theta)$ and $\tilde{n}_0 \cdot \tilde{e}_2 = \sin(\theta)$ for a Cartesian basis $\{\tilde{e}_1, \tilde{e}_2\} \subset \mathbb{R}^2$ on the plane.  It follows that $(\tilde{\nabla} y)^T \tilde{\nabla} y = \tilde{\ell}_{n_0} = \tilde{R}(\theta) \text{diag }(r^{2/3}, r^{-1/3}) \tilde{R}(\theta)^T $ for $\tilde{R}(\theta) \in SO(2)$ a rotation of $\theta$ about the normal to the initially flat sheet as required by \cite{ask_prl_14}.)
\end{enumerate}
\end{remark}

We justify the use of the metric constraint as a characterization of approximate minimizers of the strain energy through a series of main results summarized as follows.  We consider two classes of designs: (a) Nonisometric origami and (b) smooth designs.  For the former, we show that if the metric constraint holds, then the energy of actuation is  $O(h^2)$ and this $h^2$ scaling is optimal.  For the latter, we show that the metric constraint is both a necessary and sufficient condition for the energy of actuation to be $O(h^3)$.

\subsection{Nonisometric origami constructions under the metric constraint}\label{ssec:introNonisometricOrigami}
For our first main result, we consider \textit{nonisometric origami} under the metric constraint, and show that their strain energy scales at most like $h^2$.

\begin{definition}[Nonisometric origami]\label{def:nonIsoDef}
These are characterized by the following assumptions on the {\it design} and {\it deformation} respectively:
\begin{enumerate}[(i)] 
\item (The design). $\omega \subset \mathbb{R}^2$ is the union of a finite number of polygonal regions $\omega_{\alpha}$ which each have constant director field, i.e., 
\begin{equation}
\begin{aligned}
&\omega = \bigcup_{\alpha = \{1,\ldots,N\}} \omega_{\alpha}, \qquad \omega_{\alpha} \text{ mutually disjoint and polygonal,} \\
&n_0 \colon \omega \rightarrow \mathbb{S}^2 \quad \text{satisfies } \quad n_0(\tilde{x}) \equiv n_{0\alpha}, \qquad (\tilde{x} \in \omega_{\alpha}, \forall\alpha \in\{1,\ldots,N\}), \label{eq:piecewiseConst}  \\
&\tilde{n}_{0\alpha} \neq \pm \tilde{n}_{0\beta} \quad \text{when there is an interface between $\omega_{\alpha}$ and $\omega_{\beta}, \alpha \neq \beta$. }
\end{aligned}
\end{equation}
\item (The deformation). $y \in W^{1,\infty}(\omega,\mathbb{R}^3)$ is a piecewise affine and continuous midplane deformation that satisfies the metric constraint (\ref{eq:2DMetric}), i.e., 
\begin{align}\label{eq:ybarAffine}
y(\tilde{x}) = \tilde{F}_{\alpha} \tilde{x} + c_{\alpha} \quad \text{and} \quad  (\tilde{F}_{\alpha})^T \tilde{F}_{\alpha} = \tilde{\ell}_{n_{0\alpha}}
\end{align}
for all $\tilde{x} \in \omega_{\alpha}$ and all $\alpha = \{1,\ldots, N\}$.  
\end{enumerate}
\end{definition}
\noindent Note, the last condition in (\ref{eq:piecewiseConst}) is only there to ensure that each interface corresponds to a non-trivial change of the director (otherwise that interface would be superfluous). 

For a nonisometric origami design (i.e., $\omega, n_0$ as in (i)) and deformation (i.e., $y$ as in (ii)), we show that we can construct a map $Y^{h} \colon \Omega_h \rightarrow \mathbb{R}^3$ which approximately extends $y$ to $\Omega_h$ and has strain energy $\mathcal{I}_{N_0^h}^{h}(Y^h) = O(h^2)$.  In order to do so, we first smooth $y$. This relies on a technical hypothesis that $y$ have a $\delta$-smoothing:
\be{definition}\label{defn:delta}
We say that $y: \om\to\R^3$ has a $\de$-smoothing if for any $\delta > 0$ sufficiently small, there exists a map $y^\de\in C^3(\bar{\omega},\mathbb{R}^3)$ and a subset $\om_\de\subset\om$ of area less than $C \de$ such that
\beq
\begin{aligned}\label{eq:KeySmoothingHyp}
 \quad y^\delta =y \quad  \text{ on } \omega \setminus \omega_\delta, \qquad |\partial_1y^\delta \times \partial_2 y^\delta| \geq c > 0 \quad \text{ on } \omega, \\
   \|\tilde{\nabla} y^\delta\|_{L^{\infty}} \leq C, \qquad
   \|\tilde{\nabla} \tilde{\nabla} y^\delta\|_{L^{\infty}} \leq C\delta^{-1},
    \qquad \| \tilde{\nabla}^{(3)} y^\delta\|_{L^{\infty}} \leq C\delta^{-2}
\end{aligned}
\eeq
for some constants $C,c>0$ which can depend on $y$ and $\omega$ but not on $\de$.
\e{definition}
\noindent We have the following theorem:
\begin{theorem}\label{NonIsoTheorem}
Let $\omega$ and $n_0$ be as in Definition \ref{def:nonIsoDef}(i), let $y$ as in \ref{def:nonIsoDef}(ii), and let $N^h_0:\Om_h\to \S^2$ be any vector field that is close to $n_0$ in the sense of  \eqref{eq:n0tomidn0}. Suppose further that for all small enough $\de>0$, $y$ has a $\de$-smoothing $y^\de$ in the sense of Definition \ref{defn:delta} above.

Then, there exists an $m>0$ such that if we set $\delta_h=mh$, then for all small enough $h>0$ there exists a map $Y^h:\Om_h\to\R^3$ with
\begin{equation}
\begin{aligned}\label{eq:Oh2Properties}
Y^h(\tilde{x},0) =& y^{\delta_h} (\tilde{x}),\quad \tilde{x} \in \omega,\\
\mathcal{I}^h_{N_0^h}(Y^h) \leq& O(h^2). 
\end{aligned}
\end{equation}
Moreover, $Y^h$ is an approximate extension of $y$ in the sense that $\| y^{\delta_h} - y\|_{W^{1,2}(\omega,\mathbb{R}^3)} = O(h)$.
\end{theorem} 

\noindent Theorem \ref{NonIsoTheorem} is proved in Section \ref{ssec:nonIsoConst}.

The existence of such a $\de$-smoothing (of the Lipschitz continuous/origami midplane deformation $y$) is an important technical tool. It is needed because the global deformation $Y^h$ has to satisfy the \emph{incompressibility constraint} $\det  \nabla Y^h=1$. (Essentially, the non-degeneracy of the derivatives of $y^\de$ allows one to employ the inverse function theorem to derive a sufficiently well-behaved ordinary differential equation as described in section 2.)

 This technical issue has appeared in previous works on incompressibility in thin sheets (also, a $\det F >0$ constraint). It was first appreciated by Belgacem \cite{b_french_97} and later addressed in some generality by Trabelsi \cite{t_aa_06} and Conti and Dolzmann \cite{cd_incomp_06}. However, their methods are very geometrical in nature (they are largely based on Whitney's ideas on the singularities of functions $\R^n\to \R^{2n-1}$) and it is not obvious how to extract from them the $\delta-$dependent control of the higher derivatives which we need in the present context.
 
Importantly though, we prove that several examples of nonisometric origami (detailed below) \emph{indeed have a $\de$-smoothing}, in the sense of Definition \ref{defn:delta}. We do this by first showing that the existence of a $\delta$-smoothing can be reduced to a linear algebra constraint on the sets of deformation gradients associated to the origami deformation (Theorem \ref{DSmoothTheorem}), and then by explicitly verifying that this constraint holds for all nonisometric origami considered.  All this is developed in Section \ref{sec:SmoothingSection}.  

Finally, we discuss examples of nonisometric origami in  Section \ref{sec:Applications} on applications. The examples are depicted in Figure \ref{fig:ExNonisometricOrigami} and include a construction which will fold into a box, originally due to \cite{mw_phystoday_16}, as well as further examples which previously appeared in a short companion paper to this one \cite{PLB}.   We also discuss in some detail an equivalent formulation of the metric constraint \eqref{eq:ybarAffine} for nonisometric origami in terms of \emph{compatibility conditions}. These are akin to the rank-one condition studied in the context of fine-scale twinning during the austenite martensite phase transition (also actuation of active martensitic sheets) \cite{balljames_analcontmech_89,bhatta_03,bhattajames_jmps_99} and to the recently studied compatibility conditions for the actuation for nematic elastomer and glass sheets using planar programming of the director \cite{mw_PRE_11,mw_phystoday_16}. 

\subsection{On the optimality of nonisometric origami}

From Theorem \ref{NonIsoTheorem}, we can construct approximations to nonisometric origami (under the hypothesis (\ref{eq:KeySmoothingHyp})) with energy $O(h^2)$.  Thus, it is natural to ask whether these constructions are energetically optimal for a prescribed director field. 

For our second main result, we prove that this is the case (not for $\mathcal{I}_{N_0^h}^h$, but) for a two-dimensional analogue of the three-dimensional entropic strain energy,
\begin{align}\label{eq:2DEnergy}
\tilde{\mathcal{I}}_{n_0}^h(y) = h \int_{\omega} \left( | (\tilde{\nabla} y)^T \tilde{\nabla} y - \tilde{\ell}_{n_0}|^2 + h^2 |\tilde{\nabla} \tilde{\nabla} y |^2 \right) d \tilde{x}.
\end{align}
The first term here represents membrane stretching part and is minimized exactly when the metric constraint (\ref{eq:2DMetric}) is satisified.  The second term approximates bending. Such a two-dimensional energy is a widely used proxy to describe the elasticity of non-Euclidean plates (e.g., Efrati et al.\ \cite{esk_sm_13} and Bella and Kohn \cite{bk_jns_14}).  In a broader context, these proxies often agree in $h$-dependent optimal energy scaling with that of the three dimensional elastic energy, and deformations which achieve this scaling in this two dimensional setting tend to form the midplane deformations for optimal three dimensional constructions (e.g., Bella and Kohn \cite{bk_cpam_14} and the single fold approximation of Conti and Maggi \cite{cm_arma_08}). 

\begin{theorem}\label{LBTheoremTrue}
Let $r>0$ and $\neq 1$ and  let $\omega$ and $n_0$ as in Definition \ref{def:nonIsoDef}(i).  For $h> 0$ sufficiently small 
\begin{align*}
\inf \left\{\tilde{\mathcal{I}}_{n_0}^h(y) \colon y \in W^{2,2}(\omega,\mathbb{R}^3) \right\} \geq c_L h^2.
\end{align*}
Here, $c_L = c_L(n_0,r,\omega) > 0$ is independent of $h$.  
\end{theorem}

\noindent Theorem \ref{LBTheoremTrue} is proved in Section \ref{sec:LBOrigami}.

\begin{remark}
\be{enumerate}[(i)]
\item 
Theorem 1.8 shows that the best possible energy scale for the modified energy $\tilde{\mathcal{I}}_{n_0}$ is $h^2$. Conversely, we may observe that the modified energy of nonisometric origami lives on this optimal scale, at least after smoothing out the interfaces. Indeed, if we assume that, in addition to the assumptions of the theorem, there exists a $y$ as in Definition \ref{def:nonIsoDef}(ii), then for $h > 0$ sufficiently small, there exists a $y^{h} \in C^{3}(\bar{\omega},\mathbb{R}^3)$ such that
\begin{align}\label{eq:upBound2D}
\|y^h - y\|_{W^{1,2}} \leq  O(h) \quad \text{ and } \quad \tilde{\mathcal{I}}_{n_0}^h(y^h) = O(h^2).
\end{align}
It is precisely in this sense that we have \textit{optimality} of nonisometric origami.

The proof of (\ref{eq:upBound2D}) is straightforward.  Indeed, the estimates (\ref{eq:KeySmoothingHyp}), with exception to the full-rank condition, can be obtained by standard mollification (for more details, see Section \ref{sec:SmoothingSection}).  Setting $\delta = h$ for these estimates yields a $y^h$ satisfying (\ref{eq:upBound2D}). 

\item Let us discuss some of the heuristics behind the lower bound in Theorem \ref{LBTheoremTrue}. At an interface separating two regions of distinct constant director, an energetic penalty associated with membrane stretching at $O(h)$ drives the deformation to be piecewise affine with a fold precisely at the interface connecting the two regions, whereas an energetic penalty associated with bending at $O(h^3)$ cannot accommodate sharp folds, and thus a smoothing is necessitated.   This interplay gives rise to an intermediate energetic scaling between $O(h)$ and $O(h^3)$.  For isometric origami, folds can be smoothed to mostly preserve the isometric condition, leading to approximate constructions and (under suitable hypothesis) lower bounds which scale as $O(h^{8/3})$ (see, for instance, Conti and Maggi \cite{cm_arma_08}).  For nonisometric origami, the preferred metric jumps across a possible fold and this leads to a larger membrane stretching term.

\item  For the proof, we show that it is possible to reduce this estimate to a canonical problem localized at a single interface. Further, we show that a lower bound for this canonical problem is described by a one-dimensional Modica-Mortola type functional.  In their result, Modica and Mortola \cite{mm_it_77} (see also Modica \cite{m_arma_87}) prove that such functionals (under suitable hypothesis) $\Gamma$-converge to functionals which are proportional to the number of jumps of their argument.  In our setting, these jumps correspond to the jump in the preferred metric over the interface.  That these "jumps" have finite energy in the $\Gamma$-convergence setting implies the estimate in the theorem.

\end{enumerate}
\end{remark}

\subsection{Examples of pure bending actuation under the metric constraint}\label{ssec:smoothDesign}
We turn now to the case of smooth or sufficiently smooth surfaces and programs satisfying the metric constraint (\ref{eq:2DMetric}).  For these configurations, we show the actuation is pure bending, i.e., $O(h^3)$ in the entropic strain energy after actuation. 

\begin{theorem}[Smooth Surfaces]\label{SmoothSurfTheorem}
Let $r \in (0,1)$ or $ r> 1$.  Let $n_0$  and $N_0^h$ satisfy (\ref{eq:n0tomidn0}). If $y \in C^3(\bar{\omega},\mathbb{R}^3)$ and $n_0 \in C^2(\bar{\omega},\mathbb{S}^2)$ such that $(\tilde{\nabla} y)^T \tilde{\nabla} y = \tilde{\ell}_{n_0}$ everywhere on $\omega$, then for $h > 0$ sufficiently small, there exists a $Y^h \in C^{1}(\overline{\Omega}_h,\mathbb{R}^3)$ such that 
\begin{align*}
Y^h(\tilde{x},0) = y(\tilde{x}), \;\; \tilde{x} \in \omega \quad \quad \mathcal{I}_{N^h_0}^h(Y^h) = O(h^3).  
\end{align*}
\end{theorem}

Notice that for this theorem we assume $y$ and $n_0$ are $C^3$ and $C^2$ respectively.  Such smoothness is not always necessary. To highlight this, we introduce a large class of $y,n$ which \emph{automatically satisfy the two-dimensional metric constraint} \eqref{eq:2DMetric}. These surfaces are given as the graph of a function,  combined with an appropriate contraction (here we consider cooling, so $r>1$). We call these ``lifted surfaces''. They are defined by
\begin{align}\label{eq:ansatz}
y(\tilde{x}) = r^{-1/6} (x_1 e_1 + x_2 e_2) + \varphi(r^{-1/6}\tilde{x}) e_3,
\end{align}
where the function $\varphi$ is from the following set
\begin{align}\label{eq:graphVarphi}
\big\{ \phi \in W^{2,\infty}(r^{-1/6}\omega, \mathbb{R})  \colon  \|\tilde{\nabla} \phi\|_{L^{\infty}} < \lambda_{r} := r - 1, \quad \mathrm{supp} \phi \subset r^{-1/6} \omega_m \big\}.
\end{align}  
Here, we set $\omega_m := \{ \tilde{x} \in \omega \colon \dist(\tilde{x}, \partial \omega) > m > 0\}$ (recall that $\omega \subset \mathbb{R}^2$ is the midplane of the sheet, a bounded Lipschitz domain). The corresponding director field of a lifted surface is
\begin{align}\label{eq:designRef}
n_0(\tilde{x}) = \frac{1}{\lambda_{r}^{1/2}}  \left(\begin{array}{c} \partial_1 \varphi(r^{-1/6} \tilde{x} )  \\
\partial_2 \varphi(r^{-1/6}\tilde{x}) \\
(\lambda_{r} - |\tilde{\nabla} \varphi(r^{-1/6}\tilde{x})|^2)^{1/2} \end{array}\right).
\end{align}

We emphasize that any such choice of $y,n_0$ satisfies \eqref{eq:2DMetric}. This fact can be proved by rewriting \eqref{eq:2DMetric} in an equivalent form, which is in fact more practical from the perspective of design and we discuss this in Section \ref{sec:Applications}, which has a focus towards applications. These lifted surfaces have entropic energy of $O(h^3)$ (and therefore they are good candidates for designable actuation).

\begin{corollary}[Lifted Surfaces]\label{SmoothSurfCor}
Let $r > 1$ and $m >0$.  Given a midplane deformation $y$ as in (\ref{eq:ansatz}) with $\varphi$ taken from the set (\ref{eq:graphVarphi}), define the director field $n_0$ as in (\ref{eq:designRef}).  Let $N^h_0$ be close to $n_0$ in the sense of (\ref{eq:n0tomidn0}).

Then, for every $h > 0$ sufficiently small, there exists a $y^h \in C^{3}(\bar{\omega},\mathbb{R}^3)$ and an extension $Y^h \in C^1(\overline{\Omega}_h,\mathbb{R}^3)$ such that 
\begin{align*}
 Y^h(\tilde{x},0) =  y^h(\tilde{x}), \;\; \tilde{x} \in \omega, \quad \quad  \|y^h - y\|_{W^{1,\infty}(\omega)} = O(h), \quad \quad \mathcal{I}_{N^h_0}^h(Y^h) = O(h^3).
\end{align*}
\end{corollary}

The key reason why the lifted surface configurations satisfy the $O(h^3)$ scaling is that they satisfy the metric constraint, they are sufficiently smooth and (for our proof technique) they can be  approximated by even smoother configurations which satisfy the metric constraint (see Remark \ref{remBendConst}(ii)). Thus, we can generalize the proof of Theorem \ref{SmoothSurfTheorem} to obtain this result. 

The results stated here are proved in Section \ref{ssec:liftSurf}.

\be{remark}\label{remBendConst}
\be{enumerate}[(i)]
\item
The surfaces of revolution in Aharoni et al.\ \cite{ask_prl_14} and the designs exploring Gaussian curvature in Mostajeran \cite{m_pre_15} satisfy the conditions of Theorem \ref{SmoothSurfTheorem}.  Thus, these designs and their predicted actuation are pure bending configurations in that they have entropic energy of $O(h^3)$ (which justifies that they are good candidates to be realized in actuation).  
\item
To arrive at the results presented in this section (detailed in Section \ref{sec:Constructions}), we employ techniques of Conti and Dolzmann \cite{cd_incomp_06,cd_cvpd_09} to construct incompressible three dimensional deformations $Y^h \in C^{1}(\overline{\Omega}_h,\mathbb{R}^3)$.  These techniques rely on the ability to approximate Sobolev functions by sufficiently smooth functions (see Section \ref{ssec:incompSec}-\ref{ssec:liftSurf}).  In this direction, an important feature of lifted surfaces is that given any $y$ as in (\ref{eq:ansatz}) with $\varphi$ as in (\ref{eq:graphVarphi}), there exists a smooth $y^h$ approximating $y$ in the  $W^{2,2}(\omega, \mathbb{R}^3)$ norm which additionally satisfies $\tilde{\nabla} y^h \in \mathcal{D}_{r}$ on $\omega$ (see Theorem \ref{prop:equivalent} for the definition of $\mathcal{D}_r$).   The space $\mathcal{D}_{r}$ can be thought of as the appropriate generalization to nematic anisotropy of the space of matrices representing isometries.  Specifically, in the isotropic case $r = 1$, $\mathcal{D}_{r}$ reduces to $\mathcal{D}_{1} = \{\tilde{F} \in \mathbb{R}^{3\times2} \colon \tilde{F}^T \tilde{F} = I_{2\times2} \}$.  The corresponding function space $$W^{2,2}_{iso}(\omega, \mathbb{R}^3) := \{ y \in W^{2,2}(\omega,\mathbb{R}^3) \colon (\tilde{\nabla} y)^T \tilde{\nabla} y  = I_{2\times2}\;\; \text{ a.e.} \}$$ has been studied extensively in the literature as this is the space of all bending deformations for isotropic sheets (as detailed rigorously by Friesecke et al.\ \cite{fjm_cpam_02}).  For instance, Pakzad \cite{p_jdg_04} showed that smooth isometric immersions are dense in $W^{2,2}_{iso}$ as long as the initially flat sheet $\omega$ is a convex regular domain.  This was later generalized by Hornung \cite{h_arma_11} for flat sheets which belong to a much larger class of bounded and Lipschitz domains.  For nematic elastomers, an appealing analogue to these results would be a similar density result for the space 
$$
W^{2,2}_{r}(\omega, \mathbb{R}^3) := \{ y \in W^{2,2}(\omega,\mathbb{R}^3) \colon \tilde{\nabla} y \in \mathcal{D}_{r} \;\; \text{ a.e.}\}.
$$  
For instance, this space arises in compactness at the bending scale for the combined entropic, non-ideal and Frank energy studied in section \ref{ssec:metricNecessary}.  It does not appear that a result of this type has been considered so far.  Our result for non-smooth midplane deformations satisfying  $\tilde{\nabla} y \in \mathcal{D}_{r}$ a.e. is only stated for lifted surfaces, as these are the examples we can explicitly construct and approximate.  
\end{enumerate}
\end{remark}

\subsection{The metric constraint as a necessary condition for bending}\label{ssec:metricNecessary}
We come to our last main result. So far, we exhibited constructions (nonisometric origami and smooth surfaces) which satisfy the metric constraint (\ref{eq:2DMetric}) and this guarantees that they have small entropic strain energy ($O(h^2)$ and $O(h^3)$ respectively).  Now, we assume that the strain energy of a sequence of $Y^h$ is of order $h^3$ (i.e., is small) and we prove a suitable rescaling of $Y^h$ converges to a map $y:\om\to \R^3$ satisfying the metric constraint.  For this, we augment the entropic elastic energy from before. \

We no longer require the deformed director $N^h$ to be constrained as $N^h = (\nabla Y^h)N^h_0/|(\nabla Y^h N^h_0|$ (see the discussion in Remark \ref{RemarkModel}(iii)).  Instead, we introduce the \textit{non-ideal} elastic energy associated to nematic elastomers.  Following Biggins et al.\ \cite{bwb_prl_09,bwb_jmps_12} and others \cite{cdd_pre_02,ns_arxiv_16,vwt_french_96,vw_macro_97}, we take $W^{ni} \colon \mathbb{R}^{3 \times 3} \times \mathbb{R}^3 \times \mathbb{R}^3 \rightarrow \mathbb{R} \cup\{+\infty\}$ to be 
\begin{align}\label{eq:Wni}
W^{ni}(F,\nu,\nu_0) = \frac{\mu \alpha}{2} \begin{cases}
|(I_{3 \times 3}- \nu_0 \otimes \nu_0)F^T \nu|^2 & \text{ if } \nu, \nu_0 \in \mathbb{S}^2 \\
 +\infty & \text{ otherwise}
 \end{cases}
\end{align}  
(see Remark \ref{remarkNonIdeal} below). Moreover, we set $\widehat{W} := (\mu/2)^{-1}(W^e + W^{ni})$, and study the combined energy 
\begin{align}\label{eq:bigEnergy}
\mathcal{I}_{N_0^h}^{h, \varepsilon}(Y^h,N^h) := \int_{\Omega_h} \Big( \widehat{W}(\nabla Y^h, N^h,N_0^h) +\varepsilon^2|\nabla N^h|^2 \Big) dx.
\end{align}
Here, we also introduce a Frank elastic term (see Remark \ref{FrankRemark} below).

For the compactness result, we rescale the $x_3$ variable via a change of coordinates $z(x) = (\tilde{x},h^{-1} x_3)$. This allows us to consider sequences on the fixed domain $\Omega = \omega \times (-1/2,1/2)$, i.e., 
\begin{align}\label{eq:chgVar2}
V^h(z(x)) = Y^h(x),\qquad \quad M_0^h(z(x)) = N_0^h(x),\qquad h^{-3} \mathcal{I}_{N_0^h}^{h,\varepsilon} (Y^h,N^h) \equiv \mathcal{J}_{M_0^h}^{h,\varepsilon}(V^h, M^h),
\end{align}
where the rescaled energy $\mathcal{J}_{M_0^h}^{h, \varepsilon} \colon W^{1,2}(\Omega,\mathbb{R}^3) \times W^{1,2}(\Omega,\mathbb{S}^2) \rightarrow \mathbb{R} \cup \{ +\infty\}$ is given by
\begin{align}
\mathcal{J}_{M_0^h}^{h,\varepsilon}(V^h,M^h) :=  \int_{\Omega}\left( \frac{1}{h^2}\widehat{W} (\nabla_h V^h,M^h, M_0^h) + \frac{\varepsilon^2}{h^2}| \nabla_h M^h|^2 \right) dz. 
\end{align}
Here, for $f \colon \Omega \rightarrow \mathbb{R}^3$, we denote $\nabla_h f$ as $(\tilde{\nabla} f |\frac{1}{h}\partial_3 f)$, which reflects the rescaling of $x_3$ by $1/h$.
 
Given these rescalings, we have:

\begin{theorem}[Compactness]\label{CompactnessTheorem}
Let $r > 0$. Let $n_0 \in W^{1,2}(\omega,\mathbb{S}^2)$ and let
\begin{align}\label{eq:varepsilonh}
c_l h \leq \varepsilon\equiv \varepsilon_h \leq c_u h, 
\end{align}
for some constants $c_u \geq  c_l > 0$. Moreover, let $M_0^h$ satisfy 
 \begin{align}\label{eq:M0h}
M_0^h(z) = n_0(\tilde{z}) + O(h), \quad \text{ for a.e. } z \in \Omega, \quad i.e., \;\; \|M_0^h- n_0\|_{L^{\infty}(\Omega)} \leq \tau h.
\end{align} For every sequence $\{V^h, M^h \} \subset W^{1,2}(\Omega,\mathbb{R}^3) \times W^{1,2}(\Omega,\mathbb{S}^2)$ with $\mathcal{J}_{M_0^h}^{h,\varepsilon_h}(V^h, M^h) \leq C$ as  $h \rightarrow 0$, there exists a subsequence (not relabeled) and a $y \in W^{2,2}(\Omega,\mathbb{R}^3)$ independent of $z_3$ such that as  $h \rightarrow 0$
\begin{align}\label{eq:compactnessFrank}
\left(V^h - \frac{1}{|\Omega|} \int_{\Omega} V^h dz\right) \rightarrow y \quad \text{ in } W^{1,2}(\Omega,\mathbb{R}^3) \quad \text{ with } \quad (\tilde{\nabla} y)^T \tilde{\nabla} y = \tilde{\ell}_{n_0} \text{ a.e. on } \omega.
\end{align}
Moreover as $h \rightarrow 0$, $\big\| M^h - \sigma \frac{(\nabla V^h)M_0^h}{|(\nabla V^h)M_0^h|} \big \|_{L^2(\Omega)} \rightarrow 0$ for some $\sigma$ a fixed constant from the set $\{1,-1\}$.
\end{theorem}
In the energy (\ref{eq:bigEnergy}) above, we introduced two new terms compared to the strain energy (\ref{eq:StrainEnergy}): the non-ideal term (\ref{eq:Wni}) (which replaces the hard kinematic constraint) and the Frank elastic term $|\nabla N^h|^2$. We now discuss the physical background behind these energetic contributions.


\begin{remark}[The non-ideal energy density]\label{remarkNonIdeal}
\begin{enumerate}[(i)]
\item The energy density (\ref{eq:Wni}) for this contribution is well-established in the physics literature \cite{bwb_prl_09,bwb_jmps_12,ns_arxiv_16} (though, in these works it is written out in a different but nevertheless completely equivalent form).   It has microscopic origins as detailed by Verwey and Warner \cite{vw_macro_97}, and a slight variant of this energy has been used to explain the semi-soft behavior of clamped-stretched nematic elastomer sheets \cite{cdd_pre_02,vwt_french_96}.

\item  The non-ideal term prevents the material from freely forming microstructure at low energy.  As we discussed in Remark \ref{RemarkModel}(iii), some control on microstructure is necessary for predictable shape actuation.   
Nematic elastomers heterogeneously patterned for actuation are typically cross-linked in the nematic phase (e.g., the samples of Ware et al.\ \cite{wetal_science_15}), and thus encode some \textit{memory} of their patterned director $n_0$.  The non-ideal term (\ref{eq:Wni}) is modeling this memory.  (This is in contrast to nematic elastomers which are cross-linked in the high temperature isotropic phase as in the samples of Kundler and Finkelmann \cite{kf_mrc_95}, and which do readily form microstructure.)

\item During thermal actuation, the entropic energy density $W^e$ is minimized (and equal to zero) when $F = (\ell_{\nu}^f)^{1/2}  R (\ell_{\nu_0}^0)^{-1/2}$ for any $R \in SO(3)$ and any $\nu \in \mathbb{S}^2$. That is, there is a degenerate set of shape changing \textit{soft} deformations since $n$ is unconstrained by the deformation.  Introducing the non-ideal term breaks this degeneracy.    Specifically, if $W^e$ and $W^{ni}$ are both minimized (and equal to zero), then $\nu = \sigma R\nu_0 = \sigma F\nu_0/|F\nu_0|$ for $\sigma \in \{-1,1\}$ in addition to the identity above  (we make this precise in Proposition \ref{WnonIdealProp} in the appendix).  That is, $\nu$ is no longer unconstrained, but instead the initial director $\nu_0$ gets convected by the deformation to $\nu$ (or $-\nu_0$ gets convected to $\nu$ since the energies are invariant under a change of sign of the director).   This observation underlies the fact that the director is \textit{approximately} convected by the deformation at low enough energies (and therefore, we recover the sharp kinematic constraint in (\ref{eq:StrainEnergy}) up to a trivial change in the sign in the limit $h \rightarrow 0$).  As a result, the metric constraint emerges rigorously at the bending scale.
\end{enumerate}
\end{remark}

\begin{remark}[Frank elasticity]
\begin{enumerate}[(i)]\label{FrankRemark}
\item Following de Gennes and Prost \cite{dgp_book_95}, Frank elasticity is a phenomenological continuum model for an energy penalizing distortions in the alignment of the current director $N$,
\begin{align*}
W_{Fr} = \frac{\kappa_1}{2}( \text{div } N)^2 + \frac{\kappa_2}{2}  (N \cdot \text{curl } N)^2 + \frac{\kappa_3}{2}  | N \times \text{curl } N|^2.
\end{align*}
Here, the three terms physically represent splay, twist and bend of the director field with respective moduli $\kappa_1,\kappa_2,\kappa_3 > 0$.   If the moduli are equal, i.e., $\kappa_i = \kappa$ for $i = 1,2,3$, then $W_{Fr}$ reduces to 
\begin{align}\label{eq:FrankApprox}
W_{Fr} \equiv \frac{\kappa}{2} |\text{grad } N|^2.  
\end{align}
More generally, since the moduli $K_i$ are positive, we have the estimate 
\begin{align}\label{eq:ulFrank}
\frac{1}{2}\kappa_l | \text{grad } N|^2 \leq W_{Fr} \leq  \frac{1}{2} \kappa_u |\text{grad } N|^2
\end{align}
for $\kappa_{l} = \min \{ \kappa_1,\kappa_2,\kappa_3\}$ and $\kappa_u = \max\{\kappa_1, \kappa_2,\kappa_3\}$.  

We are interested foremost in how Frank energy may compete with the entropic energy at the bending scale.  Thus, we consider only the simplified model (\ref{eq:FrankApprox}) since the detailed model is sandwiched energetically by models of this type (\ref{eq:ulFrank}).   We make a further assumption regarding how distortions in nematic alignment are accounted in the energetic framework.  To elaborate, a model for Frank elasticity should ideally penalize spatial distortions in the alignment of the director field, i.e., the $\text{div}, \text{curl}$ and $\text{grad}$ operators should be with respect to the current frame.  Unfortunately, this seems quite technical to capture in a variational setting, as notions of invertibility of Sobolev maps must be carefully considered.  It is, however, an active topic of mathematical research.  For instance, we refer the interested reader to the works of Barchiesi and DeSimone \cite{bd_esaim_15} and Barchiesi et al.\ \cite{bhmc_pp_2015} for Frank elasticity and nematic elastomers in this context.  Nevertheless, for our purpose in understanding whether the metric constraint (\ref{eq:2DMetric}) is necessitated by a smallness in the energy, we find it sufficiently interesting to consider the simplified model 
\begin{align}
\label{eq:FrankApprox2}
W_{Fr} \approx \frac{\kappa}{2}|\nabla N^h|^2, \quad N^h \colon \Omega_h \rightarrow \mathbb{S}^2
\end{align}
where $N^h$ refers to the current director field as a mapping from the initially flat sheet $\Omega_h$ and $\nabla$ is the gradient with respect  to this reference state. We normalize this energy by $\mu/2$ and set $\varepsilon^2 = \kappa/ \mu$ to obtain (\ref{eq:bigEnergy}).

\item  The presence of this Frank elastic term allows us to employ the geometric rigidity result of Friesecke, James and M\"{u}ller \cite{fjm_cpam_02}. Geometric rigidity is the central technical ingredient for deriving a compactness result for bending theories from three dimensional elastic energies (compare \cite{bls_arma_15,fjm_arma_06,lmp_prsa_11,lp_esaim_11}).  The choice $\varepsilon \propto h$ is dictated by the desire to have Frank elasticity be comparable to the entropic elasticity (and thus to get a non-trivial limit) at the bending scale.  We discuss this further below.

\item 
The parameter $\varepsilon = \sqrt{\kappa/\mu}$ is likely quite small in nematic elastomers.  Specifically, in liquid crystal fluids, the moduli $\kappa_i$ (which bound $\kappa$) have been measured in detail, and these moduli are likely similar for nematic elastomers (see, for instance, the discussion in Chapter 3 \cite{wt_lceboox_03}).  Further, the shear modulus $\mu$ of the rubbery network, which is distinct to elastomers, is much larger.  Substituting the typical values for these parameters, we find $\varepsilon \sim 10-100 nm$.  Thus, entropic elasticity will often dominate Frank elasticity in these elastomers.  However, a typical thin sheet  will have a thickness $h \sim 10-100 \mu m$.  So there are two small lengthscales to consider in this problem.   For mechanical boundary conditions which induce stretch and stress in these sheets, the entropic energy does appear to dominate the Frank term.  For instance, stripe domains of oscillating nematic orientation would be suppressed by a large Frank energy, and yet these have been observed by Kundler and Finkelmann \cite{kf_mrc_95} in the clamped stretch experiments on thin sheets.  Mathematically, this dominance under stretch is made precise, for instance, by Cesana et al.\ \cite{cpk_arma_15} in studying an energy which includes Frank and entropic elastic contributions. The resulting membrane theory does not depend on Frank elasticity.  These results notwithstanding, actuation of nematic sheets with controlled heterogeneity occurs at a much lower energy state.  Therefore, it is possible that the actuated configuration emerges from a non-trivial competition between entropic and Frank elasticity at these small energy scales.  Hence, we study this competition in an asymptotic sense by taking $h$ and $\varepsilon \rightarrow 0$.   
\end{enumerate}
\end{remark}

\addtocontents{toc}{\protect\setcounter{tocdepth}{1}}
\section{Low energy deformations}\label{sec:Constructions}

We  now prove Theorems \ref{NonIsoTheorem} and \ref{SmoothSurfTheorem}.

In Section \ref{ssec:incompSec}, we construct three dimensional {\it incompressible} deformations starting from sufficiently smooth two dimensional deformations.  These constructions cover all cases of idealized actuation considered in this work.  In Section \ref{ssec:liftSurf}, we use the construction presented in Section \ref{ssec:incompSec} to prove the $O(h^3)$ energy statement for smooth surfaces and lifted surfaces (Theorem \ref{SmoothSurfTheorem} and Corollary \ref{SmoothSurfCor}).  As part of the proof, we develop two dimensional approximations to the lifted surface ansatz as needed.  In Section \ref{ssec:nonIsoConst}, we follow analogous steps to prove the $O(h^2)$ energy statement for nonisometric origami (Theorem \ref{NonIsoTheorem}).
\subsection{Incompressible extensions}\label{ssec:incompSec}
We begin with extensions of the deformations of a planar domain to three dimensional incompressible deformations of a thin domain based on the techniques of Conti and Dolzmann \cite{cd_incomp_06,cd_cvpd_09}.  
\begin{lemma}\label{IncompressibleLemma}
Let $\alpha \in \{-1,0,1\}$.  Suppose for any $\delta > 0$ sufficiently small we have $y_{\alpha}^{\delta}  \in C^{3}(\bar{\omega},\mathbb{R}^3)$ and $b_{\alpha}^{\delta} \in C^2(\bar{\omega},\mathbb{R}^3)$ satisfying
\begin{equation}\label{eq:hypIncomp}
\begin{aligned}
&\det(\tilde{\nabla}y_{\alpha}^{\delta} | b_{\alpha}^{\delta}) = 1 \;\;\; \text{ on } \omega,  \quad \quad \quad \quad \quad \quad \quad \quad  \;  \|\tilde{\nabla} y_{\alpha}^{\delta}\|_{L^{\infty}} + \|b_{\alpha}^{\delta} \|_{L^{\infty}(\omega)} \leq M,\\
&\|\tilde{\nabla} \tilde{\nabla} y_{\alpha}^{\delta}\|_{L^{\infty}} + \|\tilde{\nabla} b_{\alpha}^{\delta} \|_{L^{\infty}} \leq M\delta^{\min\{-\alpha,0\}}, \quad \quad \|\tilde{\nabla}^{(3)} y_{\alpha}^{\delta}\|_{L^{\infty}} + \|\tilde{\nabla}\tilde{\nabla} b_{\alpha}^{\delta} \|_{L^{\infty}} \leq M\delta^{-\alpha - 1}
\end{aligned}
\end{equation}
for some uniform constant $M >0$.  Then there exists an $m \equiv m(M, \alpha) \geq 1$ such that for any $h > 0$ sufficiently small, there exists a unique $\xi^{h}_{\alpha} \in C^1(\overline{\Omega}_{h},\mathbb{R})$ and an extension $Y_{\alpha}^h \in C^{1}(\overline{\Omega}_{h},\mathbb{R}^3)$ satisfying 
\begin{align}\label{eq:yhDef}
Y_{\alpha}^{h} =y^{\delta_h}_{\alpha} + \xi^h_{\alpha} b^{\delta_h}_{\alpha} ,\quad \text{ with } \quad \delta_h = mh \quad \text{ and } \quad  \det \nabla Y_{\alpha}^h = 1 \;\; \text{ on } \Omega_h.
\end{align}
In addition, $\xi_{\alpha}^{h}$ satisfies the pointwise estimates
\begin{align}\label{eq:xihEst}
|\xi^{h}_{\alpha} - x_3| \leq Ch^{\min\{-\alpha,0\}}|x_3|^{2}, \;\; |\partial_3\xi_{\alpha}^{h} - 1| \leq Ch^{\min\{-\alpha,0\}}|x_3|, \;\; |\tilde{\nabla} \xi_{\alpha}^{h}| \leq Ch^{-\alpha-1} |x_3|^2
\end{align}
everywhere on $\Omega_{h}$.  Here, each $C \equiv C(M)$ and does not depend on $h$.    
\end{lemma}
\noindent We prove this lemma in Appendix \ref{sec:IncompAppendix}.
\begin{remark}\label{incompRemark}
\begin{enumerate}[(i)]
\item The $\alpha \in \{ -1,0,1\}$ dependent hypotheses (\ref{eq:hypIncomp}) is related to the (sufficiently) smooth approximations to midplane fields $y$ and $b$ which satisfy $(\tilde{\nabla} y|b)^T (\tilde{\nabla} y |b) = \ell_{n_0}$ a.e. on $\omega$ for idealized actuation.  These approximations depend on the regularity of the midplane field $y,b$ and $n_0$.  If the fields are smooth enough, then no approximation is required, and this is reflected in the hypotheses with $\alpha = -1$.   Lifted surfaces need not be smooth (i.e., we can have $y \in W^{2,\infty}(\omega,\mathbb{R}^3) \setminus C^{3}(\bar{\omega},\mathbb{R}^3)$).  Consequently, approximations in this case correspond to $\alpha = 0$.  Finally, nonisometric origami actuations are strictly Lipschitz continuous, and as such, the approximations correspond to $\alpha = 1$.  
\item We will show below that the three dimensional extensions $Y^h$ defined in (\ref{eq:yhDef}) have low energy for appropriate choices of $y_{\alpha}^{\delta}$ and $b_{\alpha}^{\delta}$.  Moreover, the estimates (\ref{eq:xihEst}) precisely quantify the approximation $\xi_{\alpha}^h \approx x_3$, and these are crucial for the energy argument.\
\item We can choose $m =1$ for $\alpha = \{-1,0\}$.  For $\alpha = 1$, we generally have to choose $m$ such that $m \geq \max\{C(M),1\}$ where $C(M)$ is a constant that depends on $M$ but is independent of $h$.
\end{enumerate}
\end{remark}

\subsection{Upper bound for  sufficiently smooth surfaces}\label{ssec:liftSurf}
We begin with the case of sufficiently smooth surfaces  and programs which satisfy the metric constraint.  In this case, we do not have to approximate the midplane fields associated to idealized actuation, and so the approach is straightforward.

We find it useful for the proofs (here and later on) to introduce the notation
\begin{align}\label{eq:WeWnH}
W^e(F,\nu,\nu_0) =\begin{cases}
 W_{nH}((\ell_{\nu}^f)^{-1/2} F (\ell_{\nu_0}^{0})^{1/2}) & \text{ if } \nu,\nu_0 \in \mathbb{S}^2 \\
 +\infty  &\text{ otherwise},
 \end{cases}
\end{align}
for $W_{nH} \colon \mathbb{R}^{3\times3} \rightarrow \mathbb{R} \cup \{ +\infty\}$ which denotes the standard incompressible neo-Hookean model
\begin{align}\label{eq:WnH}
W_{nH}(F) := \frac{\mu}{2}\begin{dcases}
|F|^2 - 3 &\text{ if } \det F =1\\
+\infty &\text{ otherwise}
\end{dcases}
\end{align} 
(as implied, the representation of the entropic energy density above $-$ (\ref{eq:WeWnH}) combined with (\ref{eq:WnH}) $-$ is equivalent to the representation (\ref{eq:We}) since $\det( \ell_{n}^f)$ and $\det ( \ell^0_{n_0})$ are both equal to $1$ for $n,n_0 \in \mathbb{S}^2$).   We turn now to the proof.

Let $r > 0$.  We suppose that $n_0 \in C^2(\bar{\omega},\mathbb{S}^2)$ and $y \in C^{3}(\bar{\omega},\mathbb{R}^3)$ such that $(\tilde{\nabla} y)^T\tilde{\nabla} y = \tilde{\ell}_{n_0}$ on $\omega$. 

\begin{proof}[Proof of Theorem \ref{SmoothSurfTheorem}.]
Following Proposition \ref{bPropDef}, there exists a $b \in C^2(\bar{\omega},\mathbb{R}^3)$ such that $(\tilde{\nabla} y|b)^T(\tilde{\nabla} y |b) = \ell_{n_0}$ and $\det (\tilde{\nabla} y|b) = 1$.  The smoothness is due to the regularity of $n_0$ and $y$ by explicit differentiation of the parameterization in (\ref{eq:rewritb}).  Now $y$ and $b$ satisfy the hypotheses of Lemma \ref{IncompressibleLemma} with $\alpha = -1$ since these fields are $\delta$-independent.  Hence, for $h > 0$ sufficiently small there exists a $\xi^h \in C^{1}(\overline{\Omega}_h,\mathbb{R})$ and an extension $Y^h \in C^{1}(\overline{\Omega}_h,\mathbb{R}^3)$ with the properties: 
\begin{align}\label{eq:PropSmooth}
\begin{cases}Y^h := y + \xi^{h} b, 
\quad \det \nabla Y^h = 1 \;\; \text{ on } \Omega_h, \\
|\xi^{h} - x_3| \leq C|x_3|^{2}, \;\; |\partial_3\xi^{h} - 1| \leq C|x_3|, \;\; |\tilde{\nabla} \xi^{h}| \leq C |x_3|^2 \;\; \text{ on } \Omega_h
\end{cases}
\end{align}
for $C = C(\tilde{\nabla} y) > 0$ independent of $h$. 

Now note that $Y^h(\tilde{x},0) = y(\tilde{x})$ for $\tilde{x} \in \omega$ since $\xi^h(\tilde{x},0) = 0$ by the first estimate for $\xi^h$ in (\ref{eq:PropSmooth}). So it remains to prove only the $O(h^3)$ scaling of the energy $\mathcal{I}_{N_0^h}(Y^h)$.  

We compute explicitly
\begin{align*}
\nabla Y^h &= (\tilde{\nabla} y|b) + x_3 (\tilde{\nabla} b|0) + (\xi^h - x_3)(\tilde{\nabla} b|0) \\
&\;\;\;\;\;\;\; + (\partial_3 \xi^h -1) b \otimes e_3 + b \otimes \tilde{\nabla} \xi^h. 
\end{align*}
Hence, by the estimates on $\xi^h$ in (\ref{eq:PropSmooth}), we conclude 
\begin{align}\label{eq:YhLiftScale}
\nabla Y^h = (\tilde{\nabla}y|b) + O(x_3). 
\end{align}

By hypothesis, $(\tilde{\nabla} y |b)^T (\tilde{\nabla} y|b) = \ell_{n_0}$, and so we find  that 
\begin{align}\label{eq:vanishEnergy1}
W^e((\tilde{\nabla} y|b),n,n_0) = W_{nH}((\ell_{n}^{f})^{-1/2} (\tilde{\nabla} y|b)(\ell^0_{n_0})^{1/2}) = 0 \quad \text{ on } \omega
\end{align}  
following Proposition \ref{EquivalenceProp} and the identity (\ref{eq:WeWnH}) where
\begin{align*}
n := \frac{(\tilde{\nabla} y |b) n_0}{|(\tilde{\nabla} y |b) n_0|}\;\; \text{ on } \omega.
\end{align*}
Since the energy density (\ref{eq:vanishEnergy1}) vanishes, we deduce from Proposition \ref{LBProp} that 
\begin{align}\label{eq:rotationEquality1}
(\ell_{n}^{f})^{-1/2} (\tilde{\nabla} y|b)(\ell^0_{n_0})^{1/2} =: R \in SO(3) \quad \text{ on } \omega.
\end{align}
Now, we let $N^h := (\nabla Y^h)N_0^h/|(\nabla Y^h)N_0^h|$ on $\Omega_h$, and observe that 
\begin{align}\label{eq:stepLengthN0h}
(\ell_{N_0^h}^0)^{1/2} &= (\ell_{n_0}^0)^{1/2} + O(h),
\end{align}
where the equality follows from the scaling of the non-ideal terms in (\ref{eq:n0tomidn0}). Additionally given (\ref{eq:YhLiftScale}), we conclude 
\begin{align}\label{eq:stepLengthNh}
(\ell_{N^h}^f)^{-1/2} = (\ell_{n}^f)^{-1/2} + O(x_3) +  O(h).
\end{align}
Hence, combining the estimates (\ref{eq:YhLiftScale}), (\ref{eq:stepLengthN0h}) and (\ref{eq:stepLengthNh}), we find
\begin{align*}
W^e(\nabla Y^h, N^h, N_0^h) &= W_{nH}((\ell_{n}^f)^{-1/2} (\tilde{\nabla}y|b)(\ell_{n_0}^0)^{1/2} + O(x_3) + O(h)) \\
&=W_{nH}(R^T(\ell_{n}^f)^{-1/2} (\tilde{\nabla}y|b)(\ell_{n_0}^0)^{1/2} + O(x_3) + O(h)) \\
&= W_{nH}(I_{3\times3} + O(x_3) + O(h)) = O(h^2) \quad \text{ on } \omega.
\end{align*}
For the last equality, we used the definition of $R$ in (\ref{eq:rotationEquality1}) and for the inequality, we used the estimate in Proposition \ref{DumbProp}. Since this inequality holds on all of $\omega$, 
\begin{align*}
\mathcal{I}_{N_0^h}^h(Y^h) = \int_{-h/2}^{h/2} \int_{\omega} W^e(\nabla Y^h,N^h,N_0^h) dx = O(h^3).  
\end{align*}
This completes the proof.  
\end{proof}

We now apply Lemma \ref{IncompressibleLemma} to the case of lifted surfaces.    

 Let $r > 1$.  We suppose $\{\varphi,y,n_0 \}$ are as in the lifted surface ansatz (i.e., $y$ satsifying (\ref{eq:ansatz}) and $n_0$ satisfying (\ref{eq:designRef}) for  $\varphi$ as in (\ref{eq:graphVarphi}) for some $m > 0$) and $N_0^h$ is as in (\ref{eq:n0tomidn0}). 

\begin{proof}[Proof of Corollary \ref{SmoothSurfCor}.]
For $\delta >0$ sufficiently small, there exist $\delta-$dependent functions $\{\varphi_\delta, y^\delta, n_0^\delta, b^\delta\}$ approximating this ansatz as detailed in Propositions \ref{VarphihLiftProp}-\ref{bhLiftProp} at the end of this section. The approximations satisfy $(\tilde{\nabla} y^{\delta} |b^{\delta})^T(\tilde{\nabla} y^{\delta}|b^{\delta}) \approx \ell_{n_0}$, and they are sufficiently smooth so that we can apply Lemma \ref{IncompressibleLemma} with $\alpha = 0$ when we set $\delta = h$ (see Remark \ref{incompRemark}(iii)).   Thus for $h >0$ sufficiently small, there exists a $\xi^h \in C^{1}(\overline{\Omega}_h,\mathbb{R})$ and an extension $Y^h \in C^{1}(\overline{\Omega}_h,\mathbb{R}^3)$ with the properties: 
\begin{align}\label{eq:PropLift}
\begin{cases}Y^h := y^{h} + \xi^{h} b^{h}, 
\quad \det \nabla Y^h = 1 \;\; \text{ on } \Omega_h, \\
|\xi^{h} - x_3| \leq C|x_3|^{2}, \;\; |\partial_3\xi^{h} - 1| \leq C|x_3|, \;\; |\tilde{\nabla} \xi^{h}| \leq Ch^{-1} |x_3|^2 \;\; \text{ on } \Omega_h
\end{cases}
\end{align}
for $C = C(\tilde{\nabla} y) > 0$ independent of $h$.  

Now note that $Y^h(\tilde{x},0) = y^h(\tilde{x})$ for $\tilde{x} \in \omega$ since $\xi^h(\tilde{x},0) = 0$ by the first estimate for $\xi^h$ in (\ref{eq:PropLift}).  Moreover, $\|y^h - y\|_{W^{1,\infty}} = O(h)$ is shown in Proposition \ref{SmoothLiftSurfProp}.  So it remains to prove only the $O(h^3)$ scaling of the energy $\mathcal{I}_{N_0^h}(Y^h)$.  

For this, we note that given the estimates and properties established in Proposition \ref{VarphihLiftProp}-\ref{bhLiftProp}, the fact that we are smoothing on a lengthscale $\delta = h$ and the estimates for $\xi^h$ in (\ref{eq:PropLift}), the proof here follows exactly the same line of arguments as in the theorem above by replacing $\{y,n_0,n,b,R\}$ with $\{y^h,n_0^h,n^h,b^h,R^h\}$.  
\end{proof}

It remains to construct the $\delta$-dependent smoothings $\{ \varphi_{\delta}, n_0^{\delta}, y^{\delta}, b^{\delta}\}$ asserted in the definition of three dimensional deformations for lifted surfaces.  

{\it Construction of $\varphi_{\delta}$}. Consider any $\varphi$ as in (\ref{eq:graphVarphi}) for $m >0$.  We extend $\varphi$ to all of $\mathbb{R}^2$ yielding $\varphi \in W^{2,\infty}(\mathbb{R}^2,\mathbb{R}^3)$ (the extension is not relabeled), and we set 
\begin{align}\label{eq:varphihLift}
\varphi_\delta := \eta_\delta \ast \varphi \quad \text{ on } r^{-1/6} \omega
\end{align}
for a standard mollifier $\eta_\delta$ supported on a ball of radius $\delta/2$.  For this mollification, we have:
\begin{proposition}\label{VarphihLiftProp}
For $\delta > 0$ sufficiently small, $\varphi_\delta$ in (\ref{eq:varphihLift}) belongs to  $C^{\infty}_0(r^{-1/6}\omega,\mathbb{R})$ and satisfies the estimates
\begin{align*}
&\|\varphi - \varphi_\delta\|_{W^{1,\infty}} = O(\delta), \quad \|\tilde{\nabla} \varphi_\delta\|_{L^{\infty}} < \lambda_{r}, \\ 
&\|\tilde{\nabla}^{(n)} \varphi_\delta \|_{L^{\infty}} = O(\delta^{2-n}), \quad \text{ for any integer } n \geq 2.   
\end{align*}
\end{proposition}
\begin{proof}
$\varphi_\delta$ is smooth by mollification.  It vanishes on the boundary of $r^{-1/6} \omega$ for $\delta > 0$ sufficiently small since by (\ref{eq:graphVarphi}), $\spt \varphi \subset r^{-1/6}\omega_m := r^{-1/6}\{ \tilde{x} \in \omega \colon \dist(\tilde{x},\omega) > m\}$ and since $\eta_\delta$ is supported on a ball of radius $\delta/2$.  From standard manipulation of the mollification (\ref{eq:varphihLift}), the estimate on the $W^{1,\infty}$ norm follows from the Lipschitz continuity of $\varphi$ and $\tilde{\nabla} \varphi$,  the estimate on $\tilde{\nabla} \varphi_{\delta}$ follows from that fact that $\|\tilde{\nabla} \varphi\|_{L^{\infty}} < \lambda_{r}$ and the estimates on the higher derivatives follow from the fact that $\tilde{\nabla} \tilde{\nabla} \varphi \in L^{\infty}$.  
\end{proof}

{\it Construction of $n_0^{\delta}$ and $y^{\delta}$}.  We replace $\varphi$ in the lifted surface ansatz (\ref{eq:ansatz}) and (\ref{eq:designRef}) with $\varphi_{\delta}$ from the proposition above and define $n_0^\delta$ as in (\ref{eq:designRef}) and $y^\delta$ as in (\ref{eq:ansatz}) with this replacement.  We make the following observations:
\begin{proposition}\label{SmoothLiftSurfProp}
Let $\delta > 0$ sufficiently small.  Let $n_0^\delta$ and $y^\delta$ as defined above for $\varphi_\delta$ as in (\ref{eq:varphihLift}), $\varphi$ as in (\ref{eq:graphVarphi}), $n_0$ as in (\ref{eq:designRef}) and $y$ as in (\ref{eq:ansatz}).  Then $n_0^\delta \in C^{\infty}(\bar{\omega},\mathbb{S}^2)$ and $y^\delta \in C^{\infty}(\bar{\omega},\mathbb{R}^3)$ and they satisfy
\begin{align*}
&(\tilde{\nabla} y^\delta)^T (\tilde{\nabla} y^\delta) = \tilde{\ell}_{n_0^\delta} \;\; \text{ on } \omega, \quad \|n_0^\delta - n_0\|_{L^{\infty}}   = O(\delta), \quad \|y^\delta - y\|_{W^{1,\infty}} = O(\delta), \\
&\|\tilde{\nabla} y^\delta  \|_{L^{\infty}} + \|\tilde{\nabla} \tilde{\nabla} y^\delta\|_{L^{\infty}} + \|\tilde{\nabla} n_0^\delta\|_{L^{\infty}} \leq C , \quad \| \tilde{\nabla}^{(3)} y^\delta\|_{L^{\infty}} + \|\tilde{\nabla} \tilde{\nabla} n_0^\delta\|_{L^{\infty}} \leq C\delta^{-1}
\end{align*}
for $C$ independent of $\delta$.  
\end{proposition}
\begin{proof}
These properties are a consequence of the properties on $\varphi_\delta$ established in Proposition \ref{VarphihLiftProp}.  In particular, smoothness follows since $\varphi_\delta$ is a mollification; the metric constraint holds by the equivalence (\ref{eq:equivalence2DMet}) since $\|\tilde{\nabla} \varphi_\delta\|_{L^{\infty}} < \lambda_r$; the estimates on the approximations $n_0^\delta- n_0$ and $y^\delta- y$ follow from the $W^{1,\infty}$ estimate of $\varphi_\delta - \varphi$ using the explicit definition of each field; and the $\delta$-dependent derivative estimates follow from the $\delta$-dependent derivative estimates of $\varphi_\delta$ again using the explicit definition of each field.
\end{proof}

{\it Construction of $b^{\delta}$.} We construct the out-of-plane vector $b^{\delta} \colon \omega \rightarrow \mathbb{R}^3$ to ensure the metric constraint is satisfied at the midplane: 
\begin{proposition}\label{bhLiftProp}
Let $\delta > 0$ sufficiently small. Let $n_0^\delta$ and $y^\delta$ as in Proposition \ref{SmoothLiftSurfProp}.  There exists a $b^\delta \in C^{\infty}(\bar{\omega},\mathbb{R}^3)$ such that 
\begin{align}
&(\tilde{\nabla} y^\delta |b^\delta)^T (\tilde{\nabla} y^\delta |b^\delta) = \ell_{n_0^\delta}, \quad \det (\tilde{\nabla} y^\delta |b^\delta) = 1, \label{eq:bhIdents}\\ 
&\|b^\delta\|_{L^{\infty}} + \|\tilde{\nabla} b^\delta\|_{L^{\infty}} \leq C, \quad \| \tilde{\nabla} \tilde{\nabla} b^\delta\|_{L^{\infty}} \leq C\delta^{-1} \nonumber
\end{align}
for $C$ independent of $\delta$.  
\end{proposition}
\begin{proof}
Since by Proposition \ref{SmoothLiftSurfProp}, we have $(\tilde{\nabla} y^\delta)^T\tilde{\nabla} y^\delta = \tilde{\ell}_{n_0^\delta}$ everywhere on $\omega$, we apply Proposition \ref{bPropDef} pointwise everywhere on $\omega$.  Thus, we define the vector $b^\delta \colon \omega \rightarrow \mathbb{R}^3$ as in (\ref{eq:rewritb}) with $\tilde{\nabla} y^\delta$ replacing $\tilde{F}$ and $n_0^\delta$ replacing $n_0$ in these relations.  Hence, (\ref{eq:bhIdents}) holds on $\omega$.  Smoothness follows since $n_0^\delta$, $y^\delta$ and the parameterization (\ref{eq:rewritb}) are each themselves smooth.  The estimates on the derivatives of $b^\delta$ follow from the estimates on the derivative of $y^\delta$ and $n_0^\delta$ in Proposition \ref{SmoothLiftSurfProp} by explicit differentiation of the parameterization in (\ref{eq:rewritb}).  
\end{proof}

\subsection{Upper bound for nonisometric origami}\label{ssec:nonIsoConst}
We now apply Lemma \ref{IncompressibleLemma} to the case of nonisometric origami, and we use it to prove Theorem \ref{NonIsoTheorem}. 

Let $r > 0$.  We suppose $\omega$ and $n_0$ satisfy Definition \ref{def:nonIsoDef}(i), $y$ satisfies Definition \ref{def:nonIsoDef}(ii) and $N^h_0$ satisfies (\ref{eq:n0tomidn0}).  In addition, we assume there exists a $\delta$-smoothing $y^{\delta} \in C^{3}(\bar{\omega},\mathbb{R}^3)$ as in Definition \ref{defn:delta}.  
 
\begin{proof}[Proof of Theorem \ref{NonIsoTheorem}.]
In Proposition \ref{bhProp} below, we prove that the existence of a $\delta$-smoothing $y^{\delta}$ also guarantees the existence of a vector field $b^{\delta}$ that complements $y^{\delta}$.  (By this, we mean that it satisfies $(\tilde{\nabla} y^{\delta} |b^{\delta})^T(\tilde{\nabla} y^{\delta}|b^{\delta}) \approx \ell_{n_0}$, and it is sufficiently smooth so that we can apply Lemma \ref{IncompressibleLemma} with $\alpha = 1$.) Thus by Lemma \ref{IncompressibleLemma}, there exists a $m = m(\tilde{\nabla} y)\geq 1$ such that for $h >0$ sufficiently small there exists a $\xi^h \in C^{1}(\overline{\Omega}_h,\mathbb{R})$ and an extension $Y^h \in C^{1}(\overline{\Omega}_h,\mathbb{R}^3)$ with the properties: 
\begin{align}\label{eq:PropOrigami}
\begin{cases}Y^h := y^{\delta_h} + \xi^{h} b^{\delta_h} \;\; \text{ with } \;\;  \delta_h = mh,
\quad \det \nabla Y^h = 1 \;\; \text{ on } \Omega_h, \\
|\xi^{h} - x_3| \leq Ch^{-1}|x_3|^{2}, \;\; |\partial_3\xi^{h} - 1| \leq Ch^{-1}|x_3|, \;\; |\tilde{\nabla} \xi^{h}| \leq Ch^{-2} |x_3|^2 \;\; \text{ on } \Omega_h
\end{cases}
\end{align}
for $C = C(\tilde{\nabla} y) > 0$ independent of $h$. 

Now note that $Y^h(\tilde{x},0) = y^{\delta_h}(\tilde{x})$ for every $\tilde{x} \in \omega$ since  $\xi^h(\tilde{x},0) = 0$ following the first  estimate for $\xi^h$ in (\ref{eq:PropOrigami}).  Further since $y^{\delta_h}$ is a $\delta_h$-smoothing of $y$ (recall definition \ref{defn:delta}), we find $\|y^{\delta_h} - y\|_{W^{1,2}} = O(h)$. Thus, it remains only to show that the energy scales as $O(h^2)$ for this deformation. 

To this end, we first compute $\nabla Y^h$ explicitly.  We find that 
\begin{align*}
\nabla Y^h = (\tilde{\nabla} y^{\delta_h}|0) + \xi^h (\tilde{\nabla} b^{\delta_h}|0) +  (\partial_1 \xi^h b^{\delta_h}|\partial_2 \xi^h b^{\delta_h} |\partial_3 \xi^h b^{\delta_h}),
\end{align*}
and note that from Proposition \ref{bhProp}, $\tilde{\nabla} b^{\delta_h} = 0$ on the set $\omega \setminus \tilde{\omega}_{\delta_h}$ where $|\tilde{\omega}_{\delta_h}| = O(\delta_h) = O(h)$.  It follows that $\xi^h = x_3$ on this set.  Indeed, since $\det \nabla Y^h = 1$, we find that on $\omega \setminus \tilde{\omega}_{\delta_h}$,
\begin{align*}
1 = \det( (\tilde{\nabla} y^{\delta_h}|0) + (\partial_1 \xi^h b^{\delta_h}|\partial_2 \xi^h b^{\delta_h} |\partial_3 \xi^h b^{\delta_h})) = \partial_3 \xi^h \det(\tilde{\nabla} y^{\delta_h}|b^{\delta_h}).
\end{align*}
Also from Proposition \ref{bhProp}, $\det (\tilde{\nabla} y^{\delta_h}|b^{\delta_h}) = 1$.  Thus, $\partial_3 \xi^h = 1$ on $\omega \setminus \tilde{\omega}_{\delta_h}$.  Consequently, $\xi^h = x_3$ on this set since we have the condition $\xi^h(\tilde{x},0)=0$.  Thus,
\begin{align}\label{eq:stFreeGradyh}
\nabla Y^h = (\tilde{\nabla} y^{\delta_h} |b^{\delta_h}) \quad \text{ on } \omega \setminus \tilde{\omega}_{\delta_h}.
\end{align}
On the exceptional set $\tilde{\omega}_{\delta_h}$, we find that 
\begin{align}\label{eq:gradyhOriEst}
|\nabla Y^h| &= |(\tilde{\nabla} y^{\delta_h}|b^{\delta_h}) + (\partial_3 \xi^h - 1)b^{\delta_h} \otimes e_3 + x_3 (\tilde{\nabla} b^{\delta_h}|0) + (\xi^h - x_3)(\tilde{\nabla} b^{\delta_h}|0) + b^{\delta_h} \otimes \tilde{\nabla} \xi^h| \nonumber \\
&\leq |(\tilde{\nabla} y^{\delta_h}|b^{\delta_h})| + |\partial_3 \xi^h -1||b^{\delta_h}| + |x_3||\tilde{\nabla} b^{\delta_h}| + |\xi^h - x_3||\tilde{\nabla} b^{\delta_h}| + |b^{\delta_h}| |\tilde{\nabla} \xi^h|  \nonumber \\
&\leq C\left(1+ h^{-1}|x_3| + h^{-2}|x_3|^2 \right) \leq C
\end{align}
where each $C = C(\tilde{\nabla} y, m(\tilde{\nabla} y)) >0$ is independent of $h$.  These estimates follow from the estimates (\ref{eq:KeySmoothingHyp}), (\ref{eq:bh3}) in Proposition \ref{bhProp}, and (\ref{eq:PropOrigami}).  

Now, from Proposition \ref{bhProp}  $(\tilde{\nabla} y^{\delta_h}|b^{\delta_h})^T(\tilde{\nabla} y^{\delta_h}|b^{\delta_h}) = \ell_{n_0}$ on $\omega \setminus \tilde{\omega}_{\delta_h}$.  Thus, 
\begin{align}\label{eq:vanishEnergy}
W^e((\tilde{\nabla} y^{\delta_h}|b^{\delta_h}),n^h,n_0) = W_{nH}((\ell_{n^h}^{f})^{-1/2} (\tilde{\nabla} y^{\delta_h}|b^{\delta_h})(\ell^0_{n_0})^{1/2}) = 0 \quad \text{ on } \omega \setminus \tilde{\omega}_{\delta_h}
\end{align}  
following Proposition \ref{EquivalenceProp} and the identity (\ref{eq:WeWnH}) where 
\begin{align*}
n^h := \frac{(\tilde{\nabla} y^{\delta_h} |b^{\delta_h}) n_0}{|(\tilde{\nabla} y^{\delta_h} |b^{\delta_h}) n_0|}\;\; \text{ on } \omega.
\end{align*}
Since the energy density (\ref{eq:vanishEnergy}) vanishes, we deduce from Proposition \ref{LBProp} that 
\begin{align}\label{eq:rotationEquality}
(\ell_{n^h}^{f})^{-1/2} (\tilde{\nabla} y^{\delta_h}|b^{\delta_h})(\ell^0_{n_0})^{1/2} =: R^h \in SO(3) \quad \text{ on } \omega \setminus \tilde{\omega}_{\delta_h}.
\end{align}

We have yet to account for the non-ideal terms on this set as $N_0^h$ in (\ref{eq:n0tomidn0}) is the appropriate argument for the energy density, not $n_0$.  To do this, we exploit the observation in (\ref{eq:rotationEquality}).  Indeed, we set 
\begin{align*}
N^h := \frac{(\nabla Y^h) N^h_0}{|(\nabla Y^h)N^h_0|} \;\; \text{ on } \Omega_h,
\end{align*}
and observe 
\begin{align}\label{eq:estStepLength}
(\ell_{N^h_0}^0)^{1/2} = (\ell_{n_0}^0)^{1/2} + O(h), \quad (\ell_{N^h}^f)^{-1/2} = (\ell_{n^h}^f)^{-1/2} + O(h) \quad \text{ on } \omega \setminus \tilde{\omega}_{\delta_h}
\end{align}
following (\ref{eq:stFreeGradyh}) and the scaling of the non-ideal term in (\ref{eq:n0tomidn0}).  Hence on $\omega \setminus \tilde{\omega}_{\delta_h}$, we find
\begin{align}\label{eq:nonIsoEst1}
W^e(\nabla Y^h, N^h,N^h_0) &=W_{nH}((\ell_{n^h}^{f})^{-1/2} (\tilde{\nabla} y^{\delta_h}|b^{\delta_h})(\ell^0_{n_0})^{1/2} + O(h)) \nonumber \\
&=W_{nH}((R^h)^T((\ell_{n^h}^{f})^{-1/2} (\tilde{\nabla} y^{\delta_h}|b^{\delta_h})(\ell^0_{n_0})^{1/2}  + O(h))) \nonumber  \\
&= W_{nH}(I_{3\times3} + O(h)) = O(h^2).
\end{align} 
For the equalities, we used (\ref{eq:stFreeGradyh}), (\ref{eq:estStepLength}), the frame invariance of $W_{nH}$, and (\ref{eq:rotationEquality}).  For the inequality, we used the estimate in Proposition \ref{DumbProp}.  

Now, on the exceptional set $\tilde{\omega}_{\delta_h}$, we have 
\begin{align}\label{eq:nonIsoEst2}
W^e(\nabla Y^h ,N^h,N^h_0) \leq c(|\nabla Y^h|^2 +1) \leq C 
\end{align}
given the estimate in Proposition \ref{UpLowProp} and (\ref{eq:gradyhOriEst}).  Thus, on the set $\tilde{\omega}_{\delta_h}$, the energy is $|O(1)|$ compared to $h$ but this set is small for nonisometric origami, i.e., $|\tilde{\omega}_{\delta_h}| = O(\delta_h) = O(h)$ given $\delta_h = mh$ in (\ref{eq:PropOrigami}).  Hence, combining the estimates (\ref{eq:nonIsoEst1}) and (\ref{eq:nonIsoEst2}), we conclude
\begin{align*}
\mathcal{I}_{N^h_0}^h(Y^h) &= \int_{-h/2}^{h/2} \int_{\tilde{\omega}_{\delta_h}} W^e(\nabla Y^h, N^h,N^h_0) dx + \int_{-h/2}^{h/2} \int_{\omega \setminus \tilde{\omega}_{\delta_h}} W^e(\nabla Y^h,N^h, N^h_0) dx \\
&\leq Ch|\tilde{\omega}_{\delta_h}| + O(h^3) = O(h^2).
\end{align*}
This completes the proof. 
\end{proof}

\begin{proposition}\label{bhProp}
Let $r > 0$. Let $\omega$ and $n_0$ satisfy Definition \ref{def:nonIsoDef}(i) and $y$ satisfy \ref{def:nonIsoDef}(ii).  If there exists a $\delta$-smoothing $y^\delta$ of $y$ as in definition \ref{defn:delta}, then for $\delta > 0$ sufficiently small, there exists a $b^\delta \in C^2(\bar{\omega},\mathbb{R}^3)$ such that 
\begin{equation}
\begin{aligned}\label{eq:bh1}
&(\tilde{\nabla} y^\delta |b^\delta)^T (\tilde{\nabla} y^\delta |b^\delta) = \ell_{n_0} \quad \text{ and }\quad \tilde{\nabla} b^\delta = 0 \quad \text{ on } \omega \setminus \tilde{\omega}_\delta \quad \text{ with } |\tilde{\omega}_\delta| = O(\delta), \\
&\det (\tilde{\nabla} y^\delta|b^\delta) = 1 \quad \text{ everywhere on } \omega. 
\end{aligned}
\end{equation}
Moreover, $b^\delta$ satisfies
\begin{align}\label{eq:bh3}
\|b^\delta\|_{L^{\infty}} \leq C, \quad \|\tilde{\nabla} b^\delta\|_{L^{\infty}} \leq C\delta^{-1}, \quad \|\tilde{\nabla} \tilde{\nabla} b^\delta\|_{L^{\infty}} \leq C\delta^{-2}
\end{align}
everywhere on $\omega$ for some $C> 0$ which can depend on $y$ and $n_0$ but is independent of $\delta$.  
\end{proposition}
\begin{proof}
From Proposition \ref{bPropDef}, if $\tilde{F} \in \mathbb{R}^{3\times2}$ and $n_0 \in \mathbb{S}^2$ such that $\tilde{F}^T \tilde{F} = \tilde{\ell}_{n_0}$, then there exists a $b \equiv b(\tilde{F},n_0) \in \mathbb{R}^3$ such that $(\tilde{F}|b)^T (\tilde{F}|b) = \ell_{n_0}$ and $\det (\tilde{F}|b) = 1$.  The parameterization is explicit, i.e., (\ref{eq:rewritb}).  Hence, we set 
\begin{align}\label{eq:bParam}
b^{\delta} := b(\tilde{\nabla} y^{\delta}, n_0^{\delta}) \quad \text{ on } \omega
\end{align}
for the $\delta$-smoothing $y^{\delta}$ and the director $n_0^{\delta} \in C^{\infty}(\bar{\omega},\mathbb{S}^2)$ given below in Proposition \ref{ndProp}.  The parameterization $b(\tilde{F},n_0)$ is smooth in its arguments when $|\tilde{F}\tilde{e_1} \times \tilde{F} \tilde{e}_2|$ is bounded away from zero.  Consequently, (\ref{eq:bh3}) holds by the chain rule given the properties of the $\delta$-smoothing $y^{\delta}$ and that $n_0^{\delta}$ satisfies (\ref{eq:ndProps}).  Further $\det(\tilde{\nabla} y^{\delta} |b^{\delta}) = 1$ everywhere on $\omega$ as the parameterization ensures this (even when the metric constraint is not satisfied). 

It remains to verify the first two properties in (\ref{eq:bh1}).  To this end, note for $\delta$ sufficiently small we have that $y^{\delta} = y$ except on a set of measure $O(\delta)$ (by hypothesis of a $\delta$-smoothing) and that $n_0^{\delta} = n_0$ except on (perhaps a different) set of measure $O(\delta)$ (Proposition \ref{ndProp} below).   Therefore, we conclude that there is a set $\tilde{\omega}_{\delta}$ of measure $O(\delta)$ such that $y^{\delta} = y$ and $n_0^{\delta} = n_0$ on $\omega \setminus \tilde{\omega}_{\delta}$.  Moreover, $\tilde{\nabla} y = const.$, $n_0 = const$. and $(\tilde{\nabla} y)^T \tilde{\nabla} y = \tilde{\ell}_{n_0}$ in any connected region in $\omega \setminus \omega_{\delta}$.  Hence, we conclude the first two properties in (\ref{eq:bh1}) given (\ref{eq:bParam}) for $b$ as in Proposition \ref{bPropDef}.
\end{proof}

To construct $b^{\delta}$, we utilized a smoothing approximation for the piecewise constant direction design $n_0 \colon \omega \rightarrow \mathbb{S}^2$ akin to a construction of DeSimone (Assertion 1 \cite{de_arma_93}).  Precisely:
\begin{proposition}\label{ndProp}
Let $r > 0$. Let $\omega$ and $n_0$ satisfy Definition \ref{def:nonIsoDef}(i).  For any $\delta > 0$ sufficiently small, there exists an $n_0^{\delta} \in C^{\infty}(\bar{\omega},\mathbb{S}^2)$ which satisfies 
\begin{equation}
\begin{aligned}\label{eq:ndProps}
&n_0^{\delta} = n_0 \quad \text{ on } \omega \setminus \omega_{\delta} \quad \text{ with } \quad |\omega_{\delta}| = O(\delta), \\
&\|\tilde{\nabla} n_{0}^{\delta}\|_{L^{\infty}} \leq C \delta^{-1} \quad \text{ and } \quad \|\tilde{\nabla} \tilde{\nabla} n_0^{\delta} \|_{L^{\infty}} \leq C \delta^{-2}.
\end{aligned}
\end{equation}
Here $C \equiv C(n_0) >0 $ is independent of $\delta$.  
\end{proposition}
\begin{proof}
Given that $\omega = \cup_{\alpha = 1,\ldots, N} \omega_{\alpha}$ for connected polygonal regions $\omega_{\alpha}$ and $n_0 \colon \omega \rightarrow \mathbb{S}^2$ satisfies  $n_0 = n_{0\alpha}$ on each $\omega_{\alpha}$, there exists a $\nu \in \mathbb{S}^2$ such that $B_{\epsilon}(\nu) \cap  \text{range}\{ n_0\} = \emptyset$ for some $\epsilon > 0$.  We let $\Pi_{\nu} \colon \mathbb{S}^2\setminus\{\nu\} \rightarrow \mathbb{R}^2$ denote the stereographic projection with projection point $\nu$.  This map is bijective (i.e., there exists a $\Pi_{\nu}^{-1} \colon \mathbb{R}^2 \rightarrow \mathbb{S}^2 \setminus \{\nu\}$).  Thus, we extend $n_0$ to all of $\mathbb{R}^2$ by setting $n_0 = n_{01}$ for $\mathbb{R}^2 \setminus \omega$ (we do not relabel) and we define 
\begin{align}\label{eq:n0dDef}
n_0^{\delta}(\tilde{x}) = (\Pi_{\nu}^{-1} \circ \;(\eta_{\delta} \ast (\Pi_{\nu} \circ \; n_0)))(\tilde{x}), \quad \tilde{x} \in \omega.
\end{align} 
Here $\eta_{\delta} \in C^{\infty}(\mathbb{R}^2,\mathbb{R})$ is the standard mollifier on $\mathbb{R}^2$ supported on a ball of radius $\delta/2$.

We claim that this map has all the properties stated in the proposition.  Indeed, $\Pi_{\nu} \circ n_0$ maps to a compact subset of $\mathbb{R}^2$ given that $\nu$ is at least $\epsilon$ away from any $n_{0\alpha}$.  Thus, $\|\Pi_{\nu} \circ n_0\|_{L^{\infty}} \leq C$ for $C \equiv C(n_0)>0$.  Consequently, $\eta_{\delta} \ast (\Pi_{\nu} \circ \; n_0) \in C^{\infty}(\mathbb{R}^2,\mathbb{R}^2)$ with 
\begin{align*}
\| \tilde{\nabla} (\eta_{\delta} \ast (\Pi_{\nu} \circ \; n_0))\|_{L^{\infty}}  \leq C \delta^{-1} \quad \text{ and }  \quad \| \tilde{\nabla} \tilde{\nabla} (\eta_{\delta} \ast (\Pi_{\nu} \circ \; n_0))\|_{L^{\infty}}  \leq C \delta^{-2}
\end{align*}
given that $\eta_{\delta}$ is the mollifier as above.  Here $C \equiv C(n_0) > 0$ is independent of $\delta> 0$. Now $\Pi_{\nu}^{-1}$ is smooth.  Thus, $n_0^{\delta} \in \mathbb{C}^{\infty}(\bar{\omega},\mathbb{S}^2)$ and by the chain rule, we deduce the estimates in (\ref{eq:ndProps}).

For the equality condition in (\ref{eq:ndProps}), we set $\omega_{\delta} = \{ \tilde{x} \in \omega \colon \dist(\tilde{x} , \partial \omega_{\alpha}) \leq \delta/2 \;\; \text{ for some } \alpha \in \{ 1,\ldots,N\} ) \}$.  Clearly this set has measure $O(\delta)$ for $\delta >0$ sufficiently small.  Moreover, we observe that 
\begin{align*}
\eta_{\delta} \ast (\Pi_{\nu} \circ \; n_0) = \II_{\nu} \circ \; n_0 \quad \text{ on } \omega \setminus \omega_{\delta}
\end{align*}
since $n_0 = const.$ on $B_{\delta/2}(\tilde{x})$ for any $\tilde{x} \in \omega \setminus \omega_{\delta}$.  Given this and the definition of $n_0^{\delta}$ in (\ref{eq:n0dDef}), we deduce the equality in (\ref{eq:ndProps}).  This completes the proof. 
\end{proof}

\section{Optimal energy scaling of nonisometric origami}\label{sec:LBOrigami}
In this section, we prove Theorem \ref{LBTheoremTrue}.  Specifically, we show that for the two-dimensional analog to the entropic energy given by $\tilde{\mathcal{I}}^h_{n_0}$ in (\ref{eq:2DEnergy}), a piecewise constant director design in the sense of Definition \ref{def:nonIsoDef}(i) necessarily implies an energy of at least $O(h^2)$ upon actuation.  In section \ref{ssec:unitCell}, we show that this estimate can be reduced to a canonical problem localized at a single interface.  Further, we show that a lower bound for this canonical problem is described by a one-dimensional Modica-Mortola type functional \cite{m_arma_87,mm_it_77}.  In Section \ref{ssec:mmSec}, we present a self-contained argument which shows that the minimum of our Modica-Mortola type functional is necessarily bounded away from zero for $h >0$ sufficiently small.  This is the key result we use to prove Theorem \ref{LBTheoremTrue}.  

\subsection{The canonical problem}\label{ssec:unitCell}
\begin{figure}
\centering
\begin{subfigure}{.5\textwidth}
  \centering
  \includegraphics[height = 2.3in]{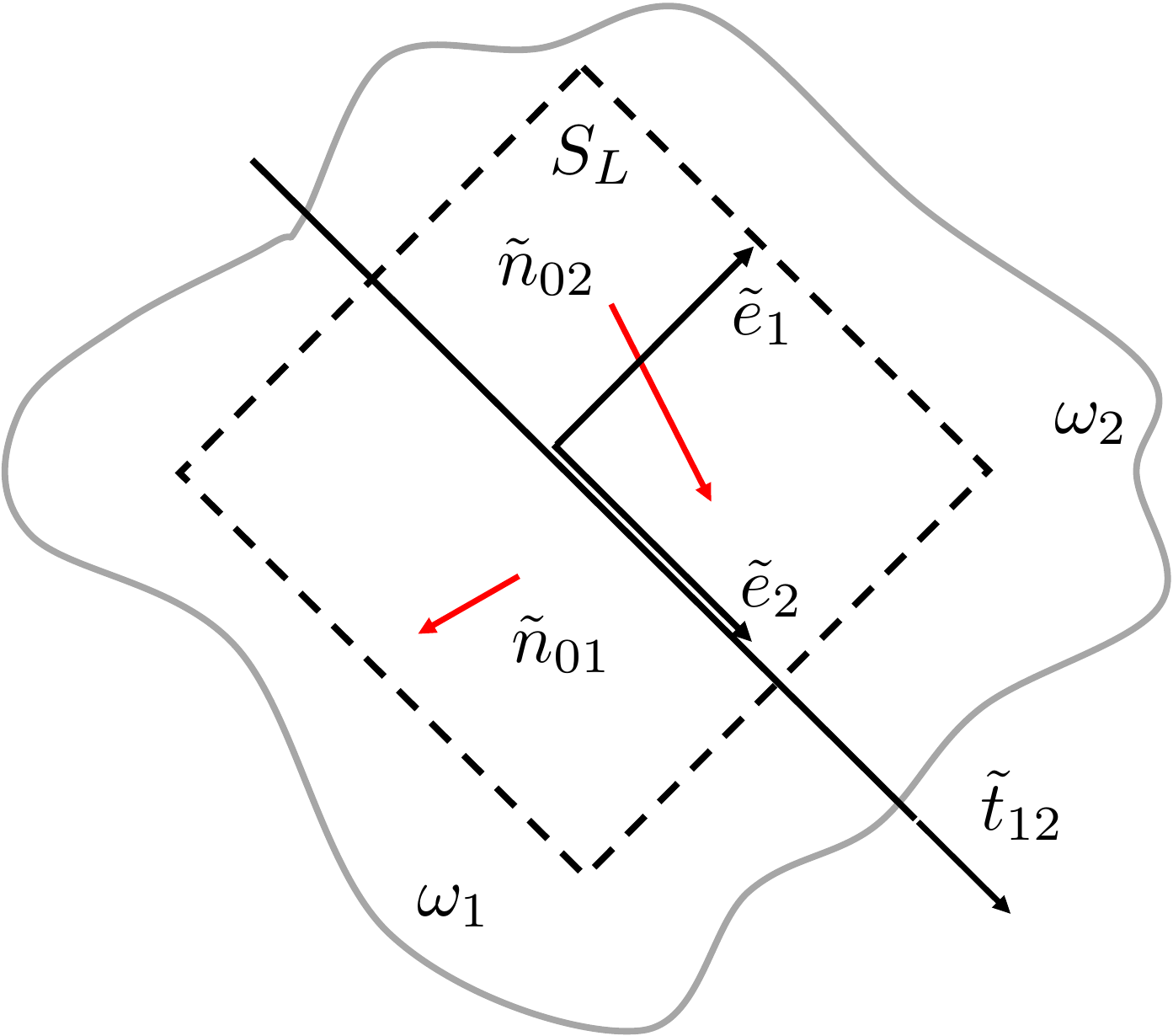}
  \caption{Case 1}
  \label{fig:C1LB}
\end{subfigure}%
\begin{subfigure}{.5\textwidth}
  \centering
  \includegraphics[width = 2.6in]{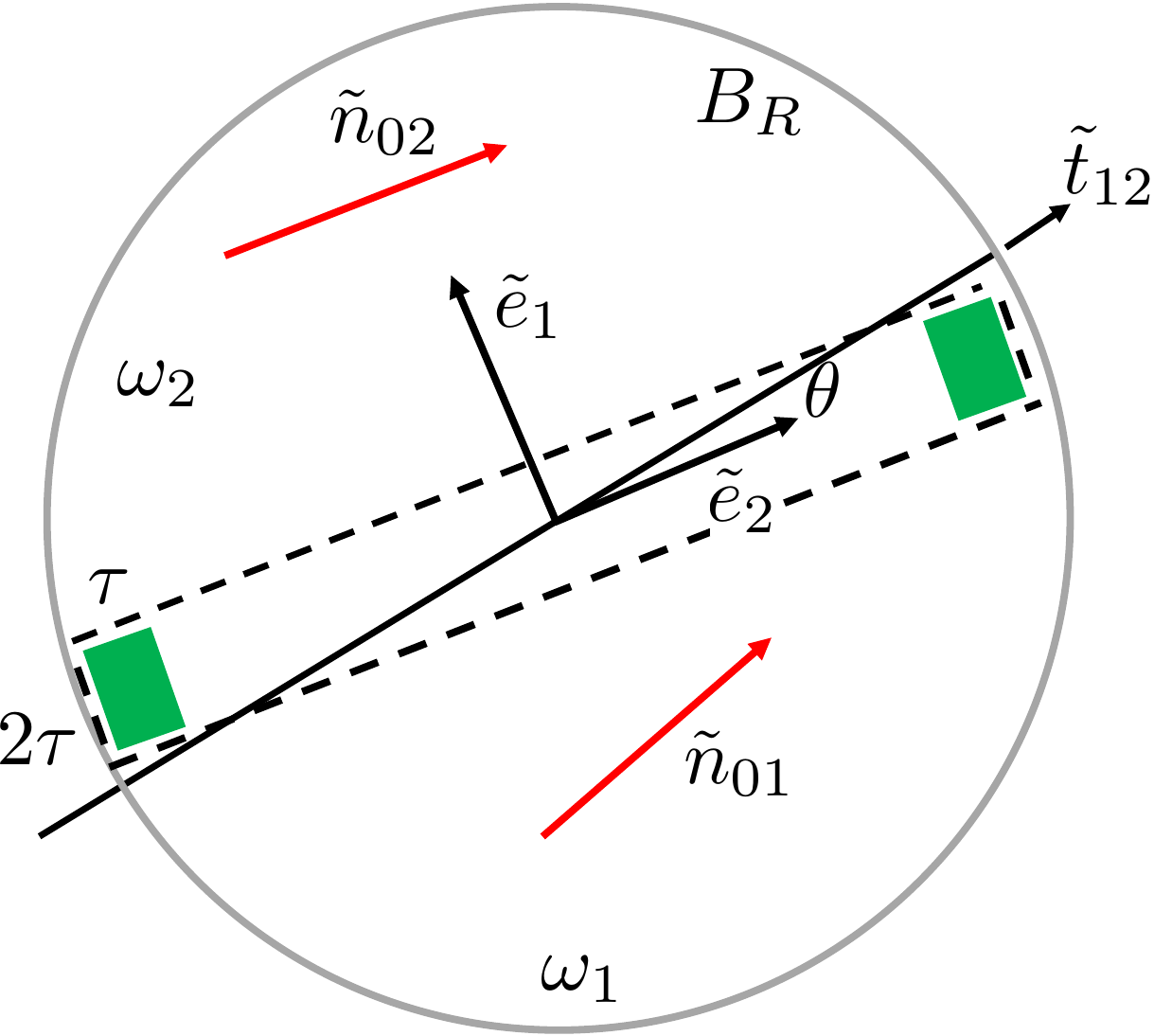}
  \caption{Case 2}
  \label{fig:C2LB}
\end{subfigure}
\caption{Schematic for canonical problem of Theorem \ref{LBTheoremTrue}}
\label{fig:LBSchematic}
\end{figure}

We assume $\omega$ and $n_0 \colon \omega \rightarrow \mathbb{S}^2$ satisfy Definition \ref{def:nonIsoDef}(i).  Then there exists a straight interface $\tilde{t}_{\alpha \beta}\in \mathbb{S}^1$ adjoining two regions $\omega_{\alpha}$ and $\omega_{\beta}$ such that $\tilde{n}_{0\alpha} \neq \pm \tilde{n}_{0\beta}$.  We let $\tilde{t}_{\alpha \beta}^{\perp} \in \mathbb{S}^1$ be the right-handed vector normal to $\tilde{t}_{\alpha \beta}$.  Focusing on this single interface, we have two cases to consider:
\begin{enumerate}
\item {\it Case 1.}  $(\tilde{n}_{0\alpha} \cdot \tilde{t}_{\alpha \beta})^2 \neq (\tilde{n}_{0\beta} \cdot \tilde{t}_{\alpha \beta})^2$ or $ (\tilde{n}_{0\alpha} \cdot \tilde{t}^{\perp}_{\alpha \beta})^2 \neq   (\tilde{n}_{0\beta} \cdot \tilde{t}^{\perp}_{\alpha \beta})^2$;
\item {\it Case 2.} $(\tilde{n}_{0\alpha} \cdot \tilde{t}_{\alpha \beta})^2 = (\tilde{n}_{0\beta} \cdot \tilde{t}_{\alpha \beta})^2$ and $ (\tilde{n}_{0\alpha} \cdot \tilde{t}^{\perp}_{\alpha \beta})^2 =  (\tilde{n}_{0\beta} \cdot \tilde{t}^{\perp}_{\alpha \beta})^2$.
\end{enumerate}

{\it Definition for Case 1:} In this case, we relabel so that $\alpha = 1$ and $\beta = 2$.  We fix a global frame so that $\tilde{e}_2$ lies on the $\tilde{t}_{12}$ interface and $\tilde{e}_1$ points in the direction of $\omega_2$.  We let the origin of this frame lie on the $\tilde{t}_{12}$ interface such that for some $L> 0$ there exists a $S_L := (-L,L)^2 \subset \omega_1 \cup \omega_2$.  A schematic of this description is provided in Figure \ref{fig:C1LB}.  We make the following observation in this case: 
\begin{proposition}\label{PropLBOri1}
If $\omega$ and $n_0$ have an interface as in the definition of Case 1 (see Figure \ref{fig:C1LB}), then for any $y \in W^{2,2}(\omega,\mathbb{R}^3)$, 
\begin{align*}
\tilde{\mathcal{I}}_{n_0}^h(y) \geq 2L^2 h M_1^h
\end{align*}
where 
\begin{align*}
M_1^h := &\inf \left \{ \int_{-1}^1 \left( (u^2 - \sigma(t))^2 + \frac{h^2}{L^2}(u')^2 \right ) dt  \colon u \in W^{1,2}((-1,1),\mathbb{R}) \;\; \text{ with} \;\; u \geq 0 \text{ a.e.} \right\} \\
&\alpha(t) = \begin{cases}
\alpha_1 &\text{ if } t < 0\\
\alpha_2 &\text{ if } t > 0.
\end{cases}
\end{align*}
Here $\alpha_1, \alpha_2 \geq 0$ and $\alpha_1 \neq \alpha_2$.
\end{proposition}
\begin{proof}
Let $y \in W^{2,2}(\omega,\mathbb{R}^3)$.  Since $S_{L} \subset \omega_1 \cup \omega_2 \subset \omega$ and the integrand in (\ref{eq:2DEnergy}) is non-negative, we have
\begin{align}\label{eq:FirstLBIneq}
\tilde{\mathcal{I}}_{n_0}^h(y) &\geq h\int_{S_L} \left(|(\tilde{\nabla} y)^T \tilde{\nabla} y - \tilde{\ell}_{n_0}|^2 + h^2|\tilde{\nabla} \tilde{\nabla} y|^2 \right) d\tilde{x}  \nonumber \\
&\geq h\int_{S_L} \left(|\tilde{e}_i \cdot ((\tilde{\nabla} y)^T \tilde{\nabla} y - \tilde{\ell}_{n_0})\tilde{e}_i|^2 + h^2|\partial_1 \partial_i y|^2 \right) d\tilde{x}  \nonumber \\
&= h \int_{S_L} \left(||\partial_i y|^2 - \sigma(x_1)|^2 + h^2|\partial_1 \partial_i y|^2 \right) d\tilde{x}
\end{align}
where $i \in \{1,2\}$ is chosen such that $(\tilde{n}_{01} \cdot \tilde{e}_i)^2 \neq (\tilde{n}_{02} \cdot \tilde{e}_i)^2$.  We see then that $\sigma$ is given by
\begin{align*}
\sigma(t) = \begin{cases}
r^{-1/3}(1+(r-1) (\tilde{n}_{01}\cdot \tilde{e}_i)^2) & \text{ if } t < 0 \\
r^{-1/3}(1+(r-1) (\tilde{n}_{02}\cdot \tilde{e}_i)^2) & \text{ if } t > 0.
\end{cases}
\end{align*}
Thus, we set $\sigma_1 := r^{-1/3}(1+(r-1) (\tilde{n}_{01}\cdot \tilde{e}_i)^2)$ and $\sigma_2 := r^{-1/3}(1+(r-1)(\tilde{n}_{02}\cdot \tilde{e}_i)^2)$, and note that $\sigma_1 \neq \sigma_2$ by definition of this case, and $\sigma_1, \sigma_2 > 0$ since $r > 0$.  

Given the chain of inequalities (\ref{eq:FirstLBIneq}), we deduce that 
\begin{align*}
\tilde{\mathcal{I}}_{n_0}^h(y) &\geq 2Lh \inf \left \{\int_{-L}^{L} \left((|w|^2 - \sigma(t))^2 + h^2 |w'|^2\right) dt \colon w \in W^{1,2}((-L,L),\mathbb{R}^3)  \right\} \\
&\geq 2Lh \inf \left \{\int_{-L}^{L} \left((|w|^2 - \sigma(t))^2 + h^2 (|w|')^2\right) dt \colon w \in W^{1,2}((-L,L),\mathbb{R}^3)  \right\} \\
&= 2Lh \inf \left \{\int_{-L}^{L} \left((v^2 - \sigma(t))^2 + h^2 (v')^2\right) dt \colon v \in W^{1,2}((-L,L), \mathbb{R}) \;\; \text{ with } \;\; v \geq 0 \text{ a.e.} \right\} \\
&= 2 L^2h M_1^h.
\end{align*}
The first inequality follows by replacing $\partial_2 y$ with a function $w$ which depends only on $x_1$ and taking the infimum amongst $W^{1,2}$ functions, and the second follows by noting $(|w|')^2 \leq |w'|^2$.  Finally, we simply replace $|w|$ by a function $v \geq 0$ for the first equality, and the second equality follows by a change of variables $v(t) = u(t/L)$.  This completes the proof.  
\end{proof}

{\it Definition for Case 2:} In this case, we again relabel so that $\alpha = 1$ and $\beta = 2$.  We note that $\tilde{n}_{02} \neq 0$ (otherwise, following the definition of Case 2, $\tilde{n}_{01} = 0$ and therefore $\tilde{n}_{01} = \tilde{n}_{02}$ which is not allowed).  Hence, we again fix a global Cartesian frame so that $\tilde{e}_2 = \tilde{n}_{02}/|\tilde{n}_{02}|$ and $\tilde{e}_1$ points in the direction of region $\omega_2$.  Next, for some $R > 0$, we find a ball $B_R \subset \omega_1 \cup \omega_2$ whose center intersects the interface $\tilde{t}_{12}$.  Note that $R = R(\omega)$ depends only on $\omega$.  We set $\theta \in (0, \pi/2]$ to be the acute angle between $\tilde{n}_{02}$ and $\tilde{t}_{12}$ (which is non-zero by definition of this case) and define 
\begin{align*}
L_1 := R\cos(\theta), \quad \quad \tau := L_1 \frac{\tan(\theta) }{1+ \tan(\theta)}.
\end{align*}
We note that by their very definition, $L_1$ and $\tau$ depend only on $\omega$ and $n_0$. Further, $\tau \in (0,L_1]$.  In particular, it cannot be zero since $\theta \neq 0$.  A schematic of this case is provided in Figure \ref{fig:C2LB}.  We make the following observation for this case:
\begin{proposition}\label{PropLBOri2}
If $\omega$ and $n_0$ have an interface as in the definition of Case 2 (see Figure \ref{fig:C2LB}), then for any $y \in W^{2,2}(\omega,\mathbb{R}^3)$, 
\begin{align*}
\tilde{\mathcal{I}}_{n_0}^h(y) \geq L_1 h \int_{-\tau}^{\tau} M_2^h(s) ds
\end{align*}
where 
\begin{align*}
M_2^h(s) := &\inf \left \{ \int_{-1}^1 \left( (u^2 - \sigma(s,t))^2 + \frac{h^2}{L_1^2}(u')^2 \right ) dt  \colon u \in W^{1,2}((-1,1),\mathbb{R}) \;\; \text{ with} \;\; u \geq 0  \text{ a.e.} \right\} \\
&\sigma(s,t) = \begin{cases}
\sigma_1 &\text{ if } t  < \max \{ -1 + \tau/L_1 , (1 -\tau/L_1)(s/\tau)\} \\
\sigma_2 &\text{ if } t > \min \{1 - \tau/L_1, (1 - \tau/L_1)(s/\tau)\}.
\end{cases}
\end{align*}
Here $\sigma_1, \sigma_2 \geq 0$ and $\sigma_1 \neq \sigma_2$.
\end{proposition}
\begin{proof}
Let $y \in W^{2,2}(\omega,\mathbb{R}^3)$.  Akin to the estimate in (\ref{eq:FirstLBIneq}), we reason that 
\begin{align}\label{eq:SecLBIneq}
\tilde{\mathcal{I}}_{n_0}^h(y) &\geq 2h \int_{-\tau}^\tau \int_{-L_1}^{L_1} \left( ||\partial_2 y|^2 - \sigma(\tilde{x})|^2 + h^2 |\partial_2^2 y|^2  \right) d\tilde{x}
\end{align}
for $\sigma(\tilde{x})$ depending on both coordinates and given by 
\begin{align*}
\sigma(\tilde{x}) = \begin{cases}
r^{-1/3}(1 + (r-1)|\tilde{n}_{02}|^2) & \text{ if } x_2 < \max\{ -L_1 + \tau, (L_1 -\tau)(x_1/\tau)\} \\
r^{-1/3}(1 + (r-1)\frac{(\tilde{n}_{01} \cdot \tilde{n}_{02})^2}{|\tilde{n}_{02}|^2}) & \text{ if } x_2 > \min\{ L_1 - \tau, (L_1 -\tau)(x_1/\tau)\}.
\end{cases}
\end{align*}
Since $\tilde{n}_{01} \neq \pm \tilde{n}_{02}$ by definition, $(\tilde{n}_{01} \cdot \tilde{n}_{02}) \neq |\tilde{n}_{02}|^2$.  Therefore, $\sigma_2 := r^{-1/3}(1 + (r-1)|\tilde{n}_{02}|^{-2}(\tilde{n}_{01} \cdot \tilde{n}_{02})^2)$ does not equal $\sigma_1 := r^{-1/3}(1 + (r-1)|\tilde{n}_{02}|^2)$.  Moreover, $\sigma_1 ,\sigma_2  > 0$ since $r > 0$.  

Now, given the inequality in (\ref{eq:SecLBIneq}), we again see that in this case 
\begin{align*}
\mathcal{I}_{n_0}^h(y) &\geq h \int_{-\tau}^{\tau} \inf \Bigg\{\int_{-L_1}^{L_1} \left((|w|^2 - \sigma(s,t))^2 + h^2 |w'|^2\right) dt \colon w \in W^{1,2}((-L_1,L_1),\mathbb{R}^3)  \Bigg\} ds \\
&\geq L_1 h \int_{-\tau}^{\tau} M_2^h(s) ds
\end{align*}
as desired.  This part of the argument is completely analogous to that of Proposition \ref{PropLBOri1}.  This completes the proof.
\end{proof}

\subsection{The Modica-Mortola analog  and proof of optimal scaling}\label{ssec:mmSec}
We have shown that given any design described by flat sheet $\omega \subset \mathbb{R}^2$ and $n_0 \colon \omega \rightarrow \mathbb{S}^2$ satisfying Definition \ref{def:nonIsoDef}(i), the problem of deducing a lower bound on the energy (\ref{eq:2DEnergy}) reduces to a canonical problem which has at most two flavors: Case 1 and Case 2 in Section \ref{ssec:unitCell}.  Actually though, following Proposition \ref{PropLBOri1} and \ref{PropLBOri2}, we find for the lower bound that one  only needs to consider the variational problem given by the one dimensional functionals
\begin{align*}
\mathcal{I}^h_s(u) := \int_{-1}^{1} \frac{1}{h}\left(( u^2 - \sigma(s,t) )^2 + c_1 h(u')^2 \right)dt , \quad s \in [-c_2, c_2]
\end{align*}
minimized amongst the functions $\{ u \in W^{1,2}((-1,1), \mathbb{R}) \colon u \geq 0  \}$ where 
\begin{align*}
\sigma(s,t) = \begin{cases}
\sigma_1 &\text{ if }   t < \max \{ -1+c_3, c_4 s\}  \\
\sigma_2 &\text{ if }  t >  \min\{ 1- c_3, c_4 s\} 
\end{cases}
\end{align*}
for $c_1, c_2 > 0 , c_3 \in (0,1]$ and  $c_4  \in [0, (1- c_3)/c_2]$.  In fact, the proof of Theorem \ref{LBTheoremTrue} follows from the observation that the infimum of $\mathcal{I}^h_s$ is bounded away from zero.  Precisely:
\begin{lemma}\label{LBLemma}
For any $c_1, c_2 >0 ,c_3 \in (0,1]$ and $c_4 \in [0,(1-c_3)/c_2]$, and for $h >0$ sufficiently small
\begin{align*}
\inf \left\{ \mathcal{I}^h_s(u) \colon u \in W^{1,2}((-1,1),\mathbb{R}) \text{ with } u \geq 0 \text{ a.e.} \right\} \geq c_L 
\end{align*}
where $c_L = c_L(c_1,c_2,c_3,c_4) > 0$ is independent of $s$ and $h$.
\end{lemma}

This is the crucial observation for the theorem.  Indeed:
\begin{proof}[Proof of Theorem \ref{LBTheoremTrue}]
We note following Section \ref{ssec:unitCell} that it suffices to restrict to the canonical problem given by the two cases in Figure \ref{fig:LBSchematic}.  From Proposition \ref{PropLBOri1} and \ref{PropLBOri2}, we have that for any $y \in W^{2,2}(\omega,\mathbb{R}^3)$,
\begin{align*}
\mathcal{I}^h_{n_0}(y) \geq \begin{cases}
2L^2h M_1^h &\text{ for Case 1}  \\
L_1 h\int_{-\tau}^\tau M_2^h(s)dx &\text{ for Case 2}.
\end{cases}
\end{align*}
In addition, we observe that 
\begin{align*}
M_1^h &= h \inf \left\{ \mathcal{I}^h_s(u) \colon u \in W^{1,2}((-1,1),\mathbb{R}) \text{ with } u \geq 0 \text{ a.e.}  \right\} \nonumber \\
&\text{ when } \;\; c_1 = L^{-1},\; c_3 = 1, \;c_4 = 0; \\
M_2^h(s) &= h \inf \left\{ \mathcal{I}^h_s(u) \colon u \in W^{1,2}((-1,1),\mathbb{R}) \text{ with } u \geq 0 \text{ a.e.}\right\} \nonumber \\
&\text{ when } \;\; c_1 = L_1^{-1},\;  c_2 = \tau, \; c_3 = \tau/L_1 , \; c_4 = (1/\tau - 1/L_1).
\end{align*}
Thus, by these observations and given Lemma \ref{LBLemma},
\begin{align*}
\mathcal{I}^h_{n_0}(y) \geq \begin{cases}
2L^2 c_L h^2 &\text{ for Case 1} \\
2L_1  \tau  c_L h^2 & \text{ for Case 2}
\end{cases}
\end{align*}
for $c_L = c_L(c_1,c_2,c_3,c_4) >0$ as in the lemma.  This completes the proof.  
\end{proof}

To close the argument, it remains to prove Lemma \ref{LBLemma}:
\begin{proof}[Proof of Lemma \ref{LBLemma}]
By the direct methods in the calculus of variations (see, for instance, Dacorogna \cite{d_dmcov_07}), we find that for any $s \in [-c_2,c_2]$ and $h > 0$, there exists a minimizer to $\mathcal{I}_s^h$ in the space $\{ u \in W^{1,2}((-1,1), \mathbb{R}) \colon u \geq 0 \text{ a.e.}  \}$.  For the lower bound, it suffices to restrict our attention to any such minimizer, which we label as $u_{s}^h$.  
Further, we may assume for some constant $M>0$ independent of $h$ and $s$ that 
\begin{align}\label{eq:boundedIhs}
\mathcal{I}^h_s(u_s^h)  < M.
\end{align}
Indeed, if for some $s \in [-c_2,c_2]$ and $h >0$ this does not hold, then we immediately establish a lower bound for this case since the reverse inequality holds.  

Now, since $c_4 \in [0,(1-c_3)/c_2]$, we have that  $\sigma(s, t) = \sigma_1$ when $t < -1+ c_3$ and $\sigma(s,t) = \sigma_2$ when $t > 1- c_3$.  Without loss of generality, we assume $\sigma_1 < \sigma_2$.  We let $\langle \sigma \rangle = (\sigma_1 + \sigma_2)/2$, and we claim that for any $h > 0$ sufficiently small,
\begin{align}\label{eq:conditionMin}
\begin{cases}
\text{ for some } \;\; t \in [-1, -1+ c_3/2], &u_{s}^h(t)^2 \in (\frac{1}{2}\sigma_1 , \frac{1}{2}(\sigma_1 + \langle \sigma \rangle) ); \\
\text{ for some } \;\; t \in [1- c_3/2, 1], &u_{s}^h(t)^2 \in (\frac{1}{2}(\sigma_2 + \langle \sigma \rangle) , \frac{3}{2}\sigma_2).
\end{cases}
\end{align}  
Indeed, suppose the first condition does not hold.  Then $(u_s^h(t)^2 - \sigma_1)^2 \geq \frac{1}{4}\min\{\sigma_1^2, (\langle \sigma \rangle - \sigma_1)^2\} > 0$ on the interval $[-1,-1+ c_3/2]$ which gives 
\begin{align*}
\mathcal{I}_s^h(u_s^h) \geq \int_{-1}^{-1+c_3/2} \frac{1}{h}((u_s^h)^2 - \sigma_1)^2 dt \geq \frac{c_3}{8h}\min\{\sigma_1^2, (\langle \sigma \rangle - \sigma_1)^2\}.
\end{align*}
Taking $h >0$ sufficiently small, we eventually arrive at a contradiction to (\ref{eq:boundedIhs}).  The second condition in (\ref{eq:conditionMin}) holds by an identical argument.  

Now, by the Sobolev embedding theorem $u_{s}^h \in W^{1,2}((-1,1),\mathbb{R})$ has a continuous representative.  This continuity and the observation that (\ref{eq:conditionMin}) holds leads to the non-zero lower bound on the energy.  Indeed, we have the estimate
\begin{align}\label{eq:lbpFinal}
\mathcal{I}^h_s(u_s^h) \geq 2\sqrt{c_1} \int_{-1}^{1} | (u_s^h)^2 - \sigma(s,t)| |(u_s^h)'| dt. 
\end{align}
Hence, we define 
\begin{align*}
a := \max \big\{ t \in [-1,1] \colon u_s^h(t)^2 = \frac{1}{2} (\sigma_1 + \langle \sigma \rangle ) \big \}, \quad b := \min \big \{ t \in (a,1] \colon u_s^h(t)^2 = \frac{1}{2} (\sigma_2 + \langle \sigma \rangle) \big \}. 
\end{align*}
By the continuity of $u_s^h$ and the observation (\ref{eq:conditionMin}), these quantities (as asserted) do, in fact, exist. Moreover,
\begin{align}\label{eq:lbFinal}
\int_{-1}^{1} | (u_s^h)^2 - \sigma(s,t)| |(u_s^h)'| dt &\geq \int_{a}^{b} | (u_s^h)^2 - \sigma(s,t)| |(u_s^h)'| dt \geq \frac{1}{2} \min_{1,2}\{ |\langle \sigma \rangle - \sigma_i|\}  \int_{a}^b |(u_s^h)'|dt \nonumber \\
&\geq \frac{1}{2}  \min_{1,2}\{ |\langle \sigma \rangle - \sigma_i|\} \Big|\int_{a}^{b} (u_s^h)' dt \Big| =\frac{1}{2}  \min_{1,2}\{ |\langle \sigma \rangle - \sigma_i|\} |u_s^h(b) - u_s^h(a)| \nonumber \\
& = \frac{1}{2 \sqrt{2}}  \min_{1,2}\{ |\langle \sigma \rangle - \sigma_i|\}|(\sigma_2 + \langle \sigma \rangle)^{1/2} - (\sigma_1 +\langle \sigma \rangle)^{1/2} |
\end{align} 
by the fundamental theorem of calculus.  Since this lower bound is positive and independent of $s$ and $h$, combining (\ref{eq:lbpFinal}) and (\ref{eq:lbFinal}) completes the proof.  
\end{proof}

\section{Compactness for bending configurations and the metric constraint}\label{sec:Compactness}
In this section, we prove that the metric constraint (\ref{eq:2DMetric}) is necessary for a configuration in pure bending when Frank elasticity is comparable to entropic elasticity at the bending scale (Theorem \ref{CompactnessTheorem}). In Section \ref{ssec:PreCompact}, we address some key preliminary results for this compactness, including a crucial lemma which is a consequence of geometric rigidity.  In Section \ref{ssec:CompactnessSec}, we prove Theorem \ref{CompactnessTheorem}.  

\subsection{Preliminaries for compactness}\label{ssec:PreCompact}
The key lemma which enables a proof of compactness in this setting is based on the result of geometric rigidity by Friesecke, James and M\"{u}ller \cite{fjm_cpam_02}, and generalization to non-Euclidean plates by Lewicka and Pakzad \cite{lp_esaim_11}.  

\begin{lemma}\label{GeomRigidityLemma}
Let $\omega \subset \mathbb{R}^2$ bounded and Lipschitz, and $r_f, r_0 >1$ and $\tau \geq 0$.  There exists a $C = C(\omega,r_f,r_0,\tau) > 0$ with the following property: For every $h > 0$, $\Omega_h := \omega \times (-h/2,h/2)$, $Y^h \in W^{1,2}(\Omega_h,\mathbb{R}^3)$, $N^h \in W^{1,2}(\Omega_h,\mathbb{S}^2)$ and $N_0^h$ as in (\ref{eq:n0tomidn0}) with $n_0 \in W^{1,2}(\omega,\mathbb{S}^2)$, there exists an associated matrix field $G^h: \omega \rightarrow \mathbb{R}^{3\times3}$ satisfying the estimates
\begin{align*}
&\frac{1}{h} \int_{\Omega_h} |G^h - (\ell^f_{N^h})^{-1/2}(\nabla Y^h)(\ell_{N_0^h}^0)^{1/2}|^2 dx \\
&\quad \quad  \leq  \frac{C}{h}\int_{\Omega_h} \left(\dist^2((\ell_{N^h}^f)^{-1/2}(\nabla Y^h)(\ell_{N_0^h}^0)^{1/2},SO(3)) + h^2 (|\nabla N^h|^2 + |\tilde{\nabla} n_0|^2 + 1) \right) dx , \\
&\int_{\omega} |\tilde{\nabla} G^h |^2 d\tilde{x} \\
&\quad \quad \leq \frac{C}{h^3}\int_{\Omega_h} \left(\dist^2((\ell_{N^h}^f)^{-1/2}(\nabla Y^h)(\ell_{N_0^h}^0)^{1/2},SO(3)) + h^2 (|\nabla N^h|^2 + |\tilde{\nabla} n_0|^2 + 1) \right) dx.
\end{align*} 
\end{lemma}
\noindent We address this result in Appendix \ref{sec:CompactAppendix}.  For similar results related to non-Euclidean plates in a different context, see Lewicka et al.\ \cite{bls_arma_15,lmp_prsa_11}.

Recall the rescaled variables $V^h$ and $M_0^h$ from Section \ref{ssec:metricNecessary}. We have:
\begin{proposition}\label{CompactProp}
Let $\omega \subset \mathbb{R}^2$ bounded and Lipschitz, $r_f, r_0 >1$, $\tau \geq 0$, and $\varepsilon_h$ as in (\ref{eq:varepsilonh}).  Let $V^h \in W^{1,2}(\Omega,\mathbb{R}^3)$, $M^h \in W^{1,2}(\Omega,\mathbb{S}^2)$ and $M_0^h$ as in (\ref{eq:M0h}) for $n_0 \in W^{1,2}(\omega,\mathbb{S}^2)$.  There exists an associated matrix field $G^h \colon \omega \rightarrow \mathbb{R}^{3\times3}$ such that 
\begin{align}\label{eq:CompactEst1}
&\int_{\Omega}|G^h - (\ell^f_{M^h})^{-1/2}(\nabla_h V^h)(\ell_{M_0^h}^0)^{1/2}|^2 dz \leq C h^2( \mathcal{J}_{M_0^h}^{h,\varepsilon_h}(V^h,M^h) + \|\tilde{\nabla} n_0\|_{L^2(\omega)}^2 + 1) \\
& \quad \quad \int_{\Omega} |\tilde{\nabla} G^h|^2 d\tilde{z} \leq C ( \mathcal{J}_{M_0^h}^{h,\varepsilon_h}(V^h,M^h) + \|\tilde{\nabla} n_0\|_{L^2(\omega)}^2 + 1) \label{eq:CompactEst2}
\end{align}
for $\varepsilon_h, c_l$ as in (\ref{eq:varepsilonh}) and some uniform $C = C(\omega,r_f, r_0,c_l,\tau)$ which is independent of $h$. 
\end{proposition}
\begin{proof}
Using Proposition \ref{LBProp} and the identity (\ref{eq:WeWnH}), we find that 
\begin{align*}
\int_{\Omega} \dist^2((\ell_{M^h}^f)^{-1/2}(\nabla_h V^h) (\ell_{M_0^h}^0)^{1/2},SO(3)) dz  \leq \int_{\Omega} \widehat{W}_{nH}((\ell_{M^h}^f)^{-1/2}(\nabla_h V^h) (\ell_{M_0^h}^0)^{1/2})dz \\
\qquad \leq \int_{\Omega} \widehat{W}^e(\nabla_h V^h, M^h, M_0^h)dz \leq \int_{\Omega} \widehat{W}(\nabla_h V^h, M^h, M_0^h)dz
\end{align*}
where $\widehat{(\cdot)} = (2/\mu)(\cdot)$ and the last inequality follows given $W^{ni}$ in (\ref{eq:Wni}) is $\geq 0$.  Since $\varepsilon_h$ as in (\ref{eq:varepsilonh}), we also find that 
\begin{align*}
\int_{\Omega} h^2 |\nabla_h M^h|^2 dz \leq \frac{1}{c^2_l} \int_{\Omega} \varepsilon_h^2 |\nabla_h M^h|^2 dz.
\end{align*}
Combining these two estimates, we find that 
\begin{align}\label{eq:upperBoundCompact}
\int_{\Omega} \left(\dist^2((\ell_{M^h}^f)^{-1/2}(\nabla_h V^h) (\ell_{M_0^h}^0)^{1/2},SO(3)) + h^2|\nabla_h M^h|^2\right) dz \leq  Ch^2\mathcal{J}_{M_0^h}^{h,\varepsilon_h}(V^h,M^h)
\end{align}
for some uniform $C = C(c_l)$.  

To obtain the desired estimates (\ref{eq:CompactEst1}) and (\ref{eq:CompactEst2}), we change variables via $z(x) := (x_1,x_2,x_3/h)$ for $x \in \Omega_h$, set the functions as defined in (\ref{eq:chgVar2}), apply Lemma \ref{GeomRigidityLemma}, and the estimates follow from the bound (\ref{eq:upperBoundCompact}). 
\end{proof}

\subsection{Compactness for comparable entropic and Frank elasticity in bending}\label{ssec:CompactnessSec}
We turn now to the proof of Theorem \ref{CompactnessTheorem}.  For clarity, we break up the proof into several steps.  

Recall that for this theorem, we suppose $M_0^h$ as in (\ref{eq:M0h}) with $n_0 \in W^{1,2}(\omega,\mathbb{S}^2)$ and $\varepsilon_h$ as in (\ref{eq:varepsilonh}).  We consider a sequence $\{V^h, M^h\} \subset W^{1,2}(\Omega,\mathbb{R}^3) \times W^{1,2}(\Omega,\mathbb{S}^2)$ such that 
\begin{align}\label{eq:hypEnergy}
\mathcal{J}_{M_0^h}^{h,\varepsilon_h}(V^h, M^h) \leq C, 
\end{align}
for all $h$ small and for some $C$ independent of $h$.  The convergences stated in each step are for a suitably chosen subsequence as $h \rightarrow 0$. 

\begin{proof}[Step 1.] {\it $M_0^h \rightarrow n_0$ in $L^2(\Omega,\mathbb{S}^2)$, and $(\ell_{M_0^h}^0)^{\pm1/2} \rightarrow (\ell_{n_0}^0)^{\pm1/2}$ in $L^2(\Omega,\mathbb{R}^{3\times3})$.} 

The first convergence is a trivial consequence of the definition of $M_0^h$ in (\ref{eq:M0h}).  The second follows from the estimate $|(\ell_{\nu_1}^0)^{\pm1/2} - (\ell_{\nu_2}^0)^{\pm1/2}| \leq C(r_0)|\nu_1- \nu_2|$  for $\nu_{1,2} \in \mathbb{S}^2$ and the first convergence.  
\end{proof}
\begin{proof}[Step 2.]{\it $M^h \rightharpoonup n$ in $W^{1,2}(\Omega,\mathbb{S}^2)$ for some $n$ independent of $z_3$ and $(\ell_{M^h}^f)^{\pm1/2} \rightarrow (\ell_{n}^f)^{\pm1/2}$ in $L^2(\Omega,\mathbb{R}^{3\times3})$.}

For $h$ sufficiently small, we have 
\begin{align}\label{eq:estMhGrad}
\frac{1}{c_l^2}\int_{\Omega} |\nabla M^h|^2 dz \leq \frac{1}{c_l^2}\int_{\Omega}\left( |\tilde{\nabla} M^h|^2 + \frac{1}{h^2} |\partial_3 M^h|^2 \right) dz \leq \mathcal{J}_{M_0^h}^{h,\varepsilon_h}(V^h,M^h) \leq C
\end{align}
for $C$ independent of $h$ by (\ref{eq:hypEnergy}).  Thus, up to a subsequence $M^h \rightharpoonup n$ in $W^{1,2}(\Omega,\mathbb{R}^3)$.  By Rellich's theorem, taking a further subsequence (if necessary), we have strong convergence, $M^h \rightarrow n$ in $L^2(\Omega,\mathbb{R}^3)$.  Since $M^h \in \mathbb{S}^2$ a.e., we deduce that $n \in \mathbb{S}^2$ a.e. by this strong convergence.  Further, $n$ is independent of $z_3$ since by the estimate (\ref{eq:estMhGrad}), we find $\partial_3 M^h \rightarrow 0$ in $L^2(\Omega,\mathbb{R}^3)$, and therefore $\partial_3 n = 0$ a.e. by the uniqueness of the weak $W^{1,2}$ limit.  The convergences of $(\ell_{M^h}^f)^{\pm1/2}$ follow by an argument similar to the convergences of $(\ell_{M_0^h}^0)^{\pm1/2}$ in Step 1.
\end{proof}

\begin{proof}[Step 3.]{\it $(V^h - \frac{1}{|\Omega|} \int_{\Omega} V^h dz) \rightharpoonup y$ in $W^{1,2}(\Omega, \mathbb{R}^3)$ for some $y$ independent of $z_3$.  Also, $h^{-1} \partial_3 V^h \rightharpoonup b$ in $L^2(\Omega,\mathbb{R}^3)$.  }

For $h$ sufficiently small, we have 
\begin{align}\label{eq:estStep3}
\frac{1}{c}\int_{\Omega} (|\tilde{\nabla} V^h|^2 +|h^{-1} \partial_3 V^h|^2 -1) dz \leq \int_{\Omega} \widehat{W}^e(\nabla_h V^h, M^h,M_0^h)dz \leq \mathcal{J}_{M_0^h}^{h,\varepsilon_h}(V^h,M^h) \leq C
\end{align}
by Proposition \ref{UpLowProp} and (\ref{eq:hypEnergy}).  Thus, since $|\nabla V^h| \leq |\nabla_h V^h|$ for $h$ small, we conclude the first convergence (up to a subsequence) given the estimate (\ref{eq:estStep3}) and an application of the Poincar\'{e} inequality.  We again use (\ref{eq:estStep3}) to conclude that up to a subsequence, $h^{-1} \partial_3 V^h \rightharpoonup b$ in $L^2(\Omega,\mathbb{R}^3)$ for some vector valued function $b$, and that the limit $y$ is independent of $z_3$ (exactly the same argument as in Step 2 for $n$ independent of $z_3$).  
\end{proof}

\begin{proof}[Step 4.] {\it There exists a sequence of matrix fields $\{G^h\}$ with $G^h \colon \omega \rightarrow \mathbb{R}^{3\times3}$ such that 
\begin{align}\label{eq:estStep4}
\int_{\Omega} | G^h  - (\ell_{M^h}^f)^{-1/2} (\nabla_h V^h)(\ell_{M_0^h}^0)^{1/2}|^2 dz \leq C h^2, \quad \quad  \int_{\omega} |\tilde{\nabla} G^h|^2 d\tilde{z} \leq C 
\end{align}
for $C$ independent of $h$.  Moreover, $G^h \rightharpoonup R$ in $W^{1,2}(\omega,\mathbb{R}^{3\times3})$ with $R \in SO(3)$ a.e.} 

To obtain the estimates in (\ref{eq:estStep4}), we first apply Proposition \ref{CompactProp} to obtain each matrix field $G^h$, and then observe that the estimates follow from the bound on the energy (\ref{eq:hypEnergy}) and the fact that by hypothesis $n_0 \in W^{1,2}(\omega,\mathbb{S}^2)$.

For the convergence, we note the first estimate in (\ref{eq:estStep4}) implies 
\begin{align*}
\int_{\omega} |G^h|^2 d\tilde{z} \leq Ch^2 + 2c(r_f,r_0) \int_{\Omega} |\nabla_h V^h|^2 dz.
\end{align*}
The constant $c(r_f,r_0)$ is from estimating the step-length tensors.  From Step 3, $\nabla_h V^h$ is bounded uniformly in $L^2$, and therefore using the above estimate and the second estimate in (\ref{eq:estStep4}), we conclude that up to a subsequence $G^h \rightharpoonup R$ in $W^{1,2}(\omega,\mathbb{R}^{3\times3})$.  Now, to deduce that $R \in SO(3)$ a.e., we estimate via two applications of the triangle inequality
\begin{align*}
\int_{\omega} \dist^2(R,SO(3)) d\tilde{z} &\leq 2 \int_{\Omega} \left(\dist^2(G^h,SO(3)) dz  + |G_h -R|^2\right) dz \\
&\leq C \left( h^2 + \int_{\Omega}( \dist^2( (\ell_{M^h}^f)^{-1/2} (\nabla_h V^h)(\ell_{M_0^h}^0)^{1/2},SO(3))  + |G_h -R|^2)dz\right)  \\
&\leq C \left( h^2 + h^2 \mathcal{J}_{M_0^h}^{h,\varepsilon_h}(V^h, M^h) + \int_{\omega} |G_h - R|^2 d \tilde{z}\right).  
\end{align*}
In the second estimate, we also use the first estimate in (\ref{eq:estStep4}).  For the third estimate, we recall (\ref{eq:upperBoundCompact}).  Now, by Rellich's theorem, we have $G^h \rightarrow R$ in $L^2(\Omega,\mathbb{R}^{3\times3})$ for a subsequence.  Thus, it is clear given (\ref{eq:hypEnergy}) that the upper bound above vanishes as $h \rightarrow 0$.  This implies $R \in SO(3)$ a.e. as desired. 
\end{proof}

\begin{proof}[Step 5.] {\it $(\ell_{M^h}^f)^{-1/2} (\nabla_h V^h) (\ell_{M_0^h}^0)^{1/2} \rightarrow R$ in $L^2(\Omega,\mathbb{R}^{3\times3})$ for $R$ from Step 4.}

Since 
\begin{align*}
&\int_{\Omega} |(\ell_{M^h}^f)^{-1/2} (\nabla_h V^h) (\ell_{M_0^h}^0)^{1/2} - R|^2 dz \\
&\quad  \quad \quad \quad \leq 2\int_{\Omega}\left(|G^h -R|^2 + |G^h-(\ell_{M^h}^f)^{-1/2} (\nabla_h V^h) (\ell_{M_0^h}^0)^{1/2}|^2\right) dz, 
\end{align*}
we conclude that $(\ell_{M^h}^f)^{-1/2} (\nabla_h V^h) (\ell_{M_0^h}^0)^{1/2} \rightarrow R$ in $L^2(\Omega,\mathbb{R}^{3\times3})$ using Step 4.
\end{proof}

\begin{proof}[Step 6.]{\it Actually, $R = (\ell_{n}^f)^{-1/2} (\tilde{\nabla} y |b) (\ell_{n_0}^0)^{-1/2}$ a.e.\ for the limiting fields above.  In particular, $(\ell_{n}^f)^{-1/2} (\tilde{\nabla} y |b) (\ell_{n_0}^0)^{1/2} \in W^{1,2}(\omega,SO(3))$.}

We observe that $\| (\ell_{M^h}^f)^{-1/2} (\nabla_h V^h)\|_{L^2(\Omega)} \leq c(r_f) \|\nabla_h V^h\|_{L^2(\Omega)} \leq C$ by the compactness of the step-length tensor on $\mathbb{S}^2$ and following Step 3.  So up to a subsequence $(\ell_{M^h}^f)^{-1/2} (\nabla_h V^h)$ converges weakly in $L^2(\Omega,\mathbb{R}^{3\times3})$.   In addition, the results of Step 2 and 3 imply $(\ell_{M^h}^f)^{-1/2} (\nabla_h V^h) \rightharpoonup (\ell_{n}^f)^{-1/2} (\tilde{\nabla} y|b)$ in $L^1(\Omega,\mathbb{R}^{3\times3})$.  Hence, in combination and by the uniqueness of the $L^1$ limit, we also have weak convergence to this limiting field in $L^2$ (rather than just $L^1$). 

Given the weak-$L^2$ convergence just established and the convergence in Step 1, we deduce 
\begin{align*}
(\ell_{M^h}^f)^{-1/2} (\nabla_h V^h)(\ell_{M_0^h}^{0})^{1/2} \rightharpoonup (\ell_{n}^f)^{-1/2}(\tilde{\nabla} y|b)(\ell_{n_0}^0)^{1/2} \quad \text{ in } L^1(\Omega,\mathbb{R}^{3\times3}).
\end{align*} 
By the convergence in Step 5 and the uniqueness of the weak-$L^1$ limit $R = (\ell_{n}^f)^{-1/2}(\tilde{\nabla} y|b)(\ell_{n_0}^0)^{1/2}$ a.e.  To complete the proof, we recall from Step 4 that $R \in W^{1,2}(\omega,\mathbb{R}^{3\times3})$ and that $R \in SO(3)$ a.e.
\end{proof}

\begin{proof}[Step 7.]{\it The sequences in Step 3 actually converge strongly in their respective spaces.  In addition, we have improved regularity: $y \in W^{2,2}(\omega,\mathbb{R}^3)$ and $b$ is independent of $z_3$ and in $W^{1,2}(\omega,\mathbb{R}^3)$. }

For the strong $L^2$ convergence, we have the estimate
\begin{align*}
&\int_{\Omega} |\nabla_h V^h - (\tilde{\nabla} y|b)|^2 dz\\
& \leq 2\int_{\Omega} \left( |\nabla_h V^h - (\ell_{M^h}^f)^{1/2}R (\ell_{M_0^h}^0)^{-1/2}|^2 + |(\ell_{M^h}^f)^{1/2}R (\ell_{M_0^h}^0)^{-1/2} - (\tilde{\nabla} y|b)|^2 \right) dz \\
&\leq C\int_{\Omega} \left( |(\ell_{M^h}^f)^{-1/2} (\nabla_h V^h)(\ell_{M_0^h}^0)^{-1/2} - R|^2 + |(\ell_{M^h}^f)^{1/2} - (\ell_{n}^f)^{1/2}|^2 + |(\ell_{M_0^h}^0)^{-1/2} - (\ell_{n_0}^0)^{-1/2}|^2 \right)dz
\end{align*}
using that $R = (\ell_{n}^f)^{-1/2}(\tilde{\nabla} y|b)(\ell_{n_0}^0)^{1/2}$ a.e. from Step 6, and that the step-length tensors are compact and invertible on $\mathbb{S}^2$.  It is clear that the upper bound $\rightarrow 0$ as $h \rightarrow 0$ due to the strong-$L^2$ convergences of each term (established in the previous steps). Thus, $\nabla_h V^h \rightarrow (\tilde{\nabla} y|b)$ in $L^2(\Omega,\mathbb{R}^{3\times3})$ as desired. 

For the improved regularity, we see that 
\begin{align*}
(\tilde{\nabla} y | b) = (\ell_{n}^f)^{1/2} R (\ell_{n_0}^0)^{-1/2} \quad \text{ a.e. on } \omega.
\end{align*}
Note that $R \in W^{1,2}$ from Step 4, $n \in W^{1,2}$ from Step 2, and $n_0 \in W^{1,2}$ by assumption. By the structure of the step-length tensors, we also have that $(\ell_{n_0}^0)^{-1/2}, (\ell_{n}^f)^{1/2} \in W^{1,2}$.  Thus, the improved regularity is clear from differentiating the right side using the product rule for these Sobolev functions. Finally, $b$ is independent of $z_3$ since $(\ell_{n}^f)^{1/2} R (\ell_{n_0}^0)^{-1/2} e_3$ is independent of $z_3$. 
\end{proof}

\begin{proof}[Step 8.]{\it Actually, 
\begin{align}\label{eq:nCompact}
&\Big\|M^h- \sigma \frac{(\nabla_h V^h) M_0^h}{|(\nabla_h V^h) M_0^h|} \Big\|_{L^2(\Omega)} \rightarrow 0, \quad \text{ and } \quad   n = \sigma \frac{(\tilde{\nabla} y|b) n_0}{|(\tilde{\nabla} y |b) n_0|}\text{ a.e. on } \omega.
\end{align}
for $\sigma$ a fixed constant from the set $\{1,-1\}$.}

Since $M_0^h \in \mathbb{S}^2$ a.e.\ by definition and $\|(\tilde{\nabla} y|b)n\|_{L^{\infty}} \leq C(r_f,r_0)$ given Step 6, 
\begin{align*}
\int_{\Omega}& \big|(I_{3 \times 3} - M_0^h \otimes M_0^h)(\nabla_h V^h)^T M^h - (I_{3\times 3} - n_0 \otimes n_0)(\tilde{\nabla} y|b)^Tn\big|^2 dz \\
& \leq C(r_f,r_0) \int_{\Omega} \Big(|(\nabla_h V^h)^T M^h - (\tilde{\nabla} y|b)^Tn|^2 + |M_0^h \otimes M_0^h - n_0 \otimes n_0|^2 \Big)dz
\end{align*}
Given the convergences from the previous steps, we conclude $(\nabla_h V^h)^T M^h \rightarrow (\tilde{\nabla} y|b)^T n$  and $M_0^h \otimes M_0^h \rightarrow n_0 \otimes n_0$ both in $L^2(\Omega)$.  Thus following the estimate above,
\begin{align*}
(I_{3 \times 3} - M_0^h \otimes M_0^h)(\nabla_h V^h)^T M^h \rightarrow (I_{3\times 3} - n_0 \otimes n_0)(\tilde{\nabla} y|b)^Tn \quad \text{ in } L^{2}(\Omega,\mathbb{R}^3).
\end{align*}
Notice also that 
\begin{align*}
\int_{\Omega} \widehat{W}^{ni}(\nabla V^h, M^h, M_0^h) dz \leq h^2 \mathcal{J}_{M_0^h}^{h,\varepsilon_h}(V^h, M^h) \leq Ch^2.
\end{align*}
Consequently, $(I_{3\times3}-M_0^h \otimes M_0^h)(\nabla_h V^h)^T M^h$ actually converges strongly to zero in $L^2$.   Hence, by the uniqueness of the $L^2$ limit, and using the identity for $(\tilde{\nabla} y|b)$ in Step 6, 
\begin{align*}
 r_f^{1/3} r_0^{1/6} (I_{3 \times 3} - n_0 \otimes n_0) R^T n = (I_{3\times 3} - n_0 \otimes n_0)(\tilde{\nabla} y|b)^Tn =  0 \quad \text{ a.e. on } \omega
\end{align*}
for $R \in SO(3)$ a.e.\ defined from Step 6.  Thus, it must be that $n = Rn_0$ or $n = - Rn_0$ a.e.\ on $\omega$ (note, the sign cannot flip since $R \in W^{1,2}$, $n_0 \in W^{1,2}$ and $n \in W^{1,2}$ from the previous steps).  This follows, for instance, from de Giorgi's lemma, which makes precise the intuition that $W^{1,2}$ functions cannot have jumps. We denote this sign by $\sigma$ as in the statement. Again using the identity for $R$ in Step 6, we conclude the a.e.\ equality in (\ref{eq:nCompact}).  As a consequence, $M^h \rightarrow \sigma (\tilde{\nabla} y |b) n_0 /|(\tilde{\nabla} y|b)n_0|$ in $L^2(\Omega,\mathbb{S}^2)$.  Further,  $(\nabla_hV^h)M_0^h/|(\nabla_h V^h)M_0^h| \rightarrow (\tilde{\nabla} y |b) n_0 /|(\tilde{\nabla} y|b)n_0|$ in $L^2(\Omega,\mathbb{S}^2)$ (using, for instance, the incompressibility of $\nabla_h V^h$, the fact that $M_0^h \in \mathbb{S}^2$ a.e., the $L^2$/pointwise a.e.\ convergence of $(\nabla V^h) M_0^h$ and the Lebesgue dominated convergence theorem).  The convergence in (\ref{eq:nCompact}) follows since $\sigma$ is $1$ or $-1$.
\end{proof} 

\begin{proof}[Step 9.]{\it Finally, 
\begin{align}\label{eq:Metric2DProof}
(\tilde{\nabla} y)^T \tilde{\nabla} y = \ell_{n_0} \text{ a.e. on } \omega.  
\end{align}}

From Step 6, $(\ell_{n}^f)^{-1/2} (\tilde{\nabla} y|b) (\ell^0_{n_0})^{1/2} \in W^{1,2}(\omega,SO(3))$, and from Step 7, $n$ as in (\ref{eq:nCompact}).  Hence,
\begin{align*}
\int_{\omega} W^e\Big((\tilde{\nabla} y|b),\frac{(\tilde{\nabla} y|b)n_0}{|(\tilde{\nabla} y|b)n_0|},n_0\Big) d\tilde{z} &= \int_{\omega} W^e\Big((\tilde{\nabla} y|b),\sigma \frac{(\tilde{\nabla} y|b)n_0}{|(\tilde{\nabla} y|b)n_0|},n_0\Big) d\tilde{z} \\
&= \int_{\omega} W_{nH}((\ell_{n}^f)^{-1/2} (\tilde{\nabla} y|b) (\ell^0_{n_0})^{1/2}) d \tilde{z} = 0
\end{align*}
using that $\sigma$ from Step 8 is either $1$ or $-1$, the definitions of $W^e$ and $W_{nH}$, and since $W_{nH}$ vanishes on $SO(3)$.  The conclusion (\ref{eq:Metric2DProof}) follows by Proposition \ref{EquivalenceProp}.
\end{proof}

\begin{proof}[Proof of Theorem \ref{CompactnessTheorem}.]
The proof follows by the collection of steps above.  In particular, Step 7 shows the strong convergence of $\{V^h\}$ and the desired regularity of the limiting field $y$ as a consequence of (\ref{eq:hypEnergy}).  Step 8 shows the convergence of the director $\{M^h\}$ as required.  Step 9 shows that the metric constraint must also be satisfied.  This is the proof.
\end{proof}

\section{On approximating origami deformations by $\delta$-smoothings}\label{sec:SmoothingSection}
When constructing three dimensional deformations $Y^h$ for nonisometric origami in Section \ref{ssec:nonIsoConst}, we assumed the existence of a $\delta$-smoothing.  We now supplement the proof of this existence in several exemplary cases.    The basic idea of our construction is (i) interfaces can be smoothed trivially and (ii) the existence problem at a junction can be reduced to a linear algebra constraint related to piecewise constant deformation gradients at each junction.  The linear algebra constraint can be found in Theorem \ref{DSmoothTheorem} below.   


\subsection{Formulation of a single junction}\label{ssec:IntroAffine}

We first consider the smoothing of a piecewise affine and continuous deformation of a polygonal region $\omega \subset \mathbb{R}^2$ containing a single junction.

{\it Definition of a single junction.} We fix a right-handed frame with standard basis $\{\tilde{e}_1, \tilde{e}_2 \} \subset \mathbb{R}^2$, and we set  $\tilde{x} := x_1 \tilde{e}_1 + x_2 \tilde{e}_2$.  We suppose $\omega$ contains $K$ interfaces  merging at a point $\tilde{p} \in \mathbb{R}^2$, each separating regions of distinct constant deformation gradient.  For each $\alpha \in \{ 1,\ldots, K\}$, the vector defining the interface (and pointing away from the junction $\tilde{p}$) is called $\tilde{t}_{\alpha} \in \mathbb{S}^1$ with $\tilde{t}_{\alpha}^{\perp} \in \mathbb{S}^1$ the right-handed vector normal to $\tilde{t}_{\alpha}$.  A schematic of this is shown in Figure \ref{fig:RefRegions} for (a) an exterior junction (i.e., where the junction $\tilde{p}$ lies on $\partial \omega$) and (b) an interior junction.  We write 
\begin{align}
\omega = \bigcup_{\alpha \in \{1,\ldots,K\}} S_{\alpha} 
\end{align}
where each $S_{\alpha}$ is a polygonal sector containing the $\tilde{t}_{\alpha}$ interface whose boundaries merging to $\tilde{p}$ either bisect the angle between $\tilde{t}_{\alpha-1}$ and $\tilde{t}_{\alpha}$, bisect the angle between $\tilde{t}_{\alpha}$ and $\tilde{t}_{\alpha + 1}$ or form the boundary of $\omega$.  A schematic of this is shown in Figure \ref{fig:RefRegions} for (c) an exterior junction and (d) an interior junction.  Finally, we let $\theta_{\alpha} >0$ denote the angle between $\tilde{t}_{\alpha}$ and $\tilde{t}_{\alpha + 1}$ for each $\alpha \in \{1,\ldots K-1\}$ (and $\theta_K>0$ the angle between $\tilde{t}_K$ and $\tilde{t}_1$ if $\tilde{p}$ is an interior junction) and  define 
\begin{align}\label{eq:thetaAst}
\theta^{\ast} := \min_{\alpha} \theta_{\alpha}.
\end{align}

\begin{figure}
\centering
\begin{subfigure}{\linewidth}
\centering
\includegraphics[width = 4.8 in]{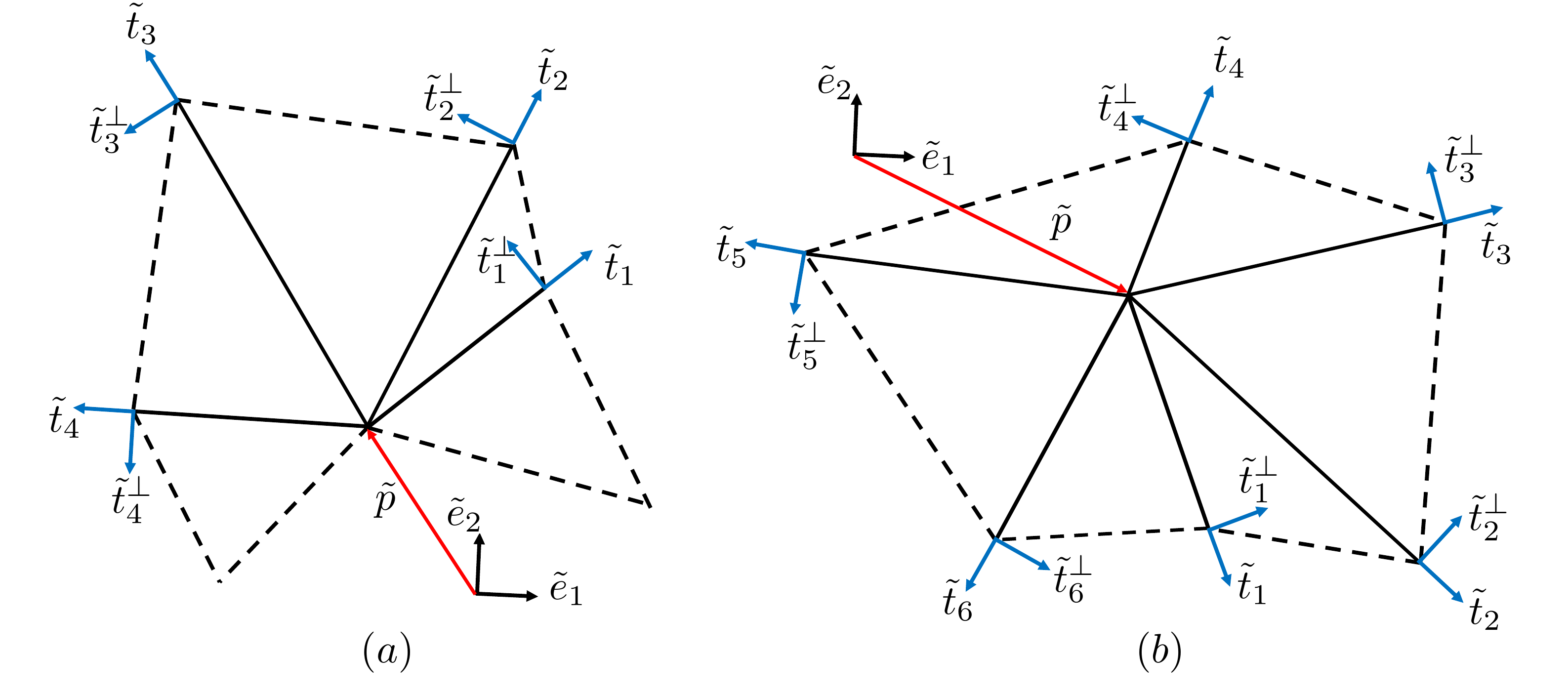}
\end{subfigure}
\begin{subfigure}{\linewidth}
\centering
\includegraphics[width = 4.8 in]{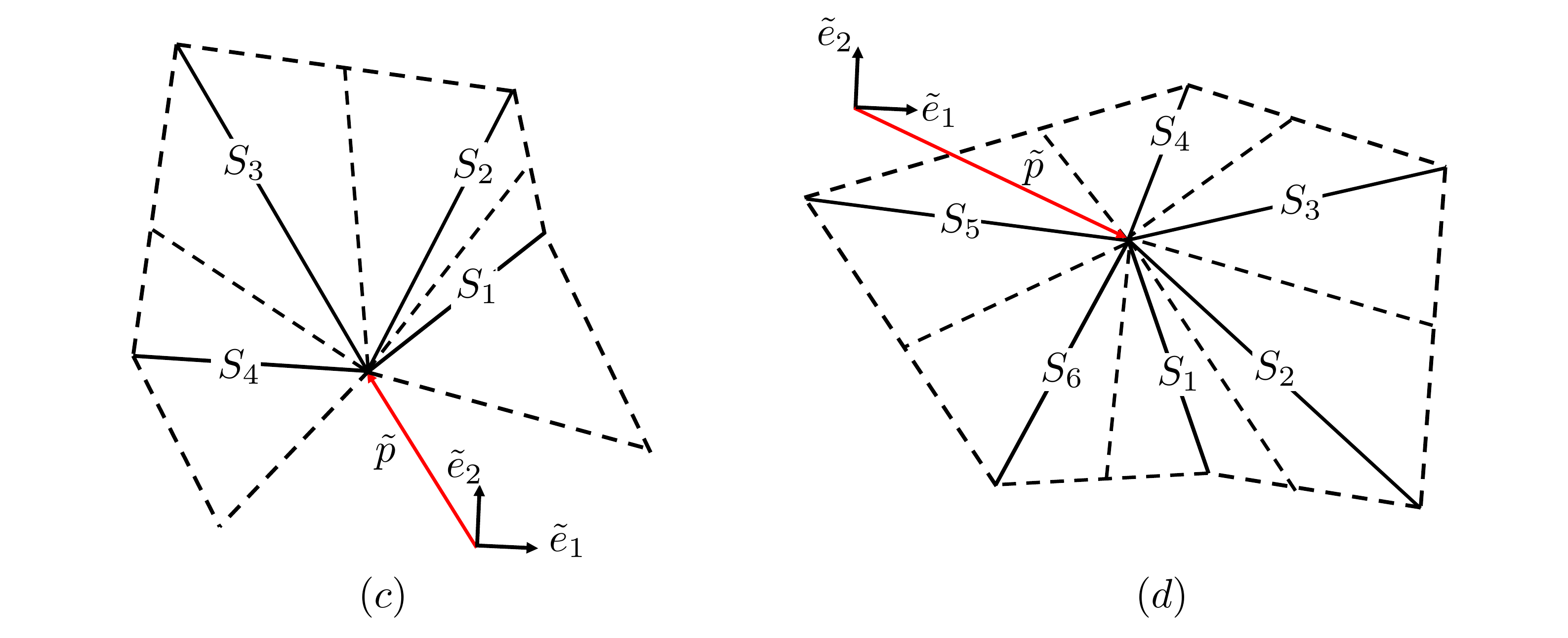}
\end{subfigure}
\caption{Schematic on a single junction at point $\tilde{p}$: Exterior (a)/(c) and interior (b)/(d). }
\label{fig:RefRegions}
\end{figure}

{\it Definition of origami deformation of a single junction.} We consider a general piecewise affine and continuous deformation $y \colon \omega \rightarrow \mathbb{R}^3$ of this single junction at point $\tilde{p}$.  This is defined as
\begin{align}\label{eq:yAffine}
y(\tilde{x}) = \gamma_{\alpha}((\tilde{x}-\tilde{p}) \cdot \tilde{t}_{\alpha}^{\perp}) + ((\tilde{x}-\tilde{p}) \cdot \tilde{t}_{\alpha}) \tilde{F}_{\alpha} \tilde{t}_{\alpha} + y(\tilde{p}) \quad \text{ if } \quad \tilde{x} \in S_{\alpha}, \quad \alpha \in \{ 1, \ldots K\},  
\end{align}
for some $y(\tilde{p}) \in \mathbb{R}^{3}$, for $\gamma_{\alpha} \colon \mathbb{R} \rightarrow \mathbb{R}^3$ satisfying
\begin{align}\label{eq:gammaAlpha}
\gamma_{\alpha}(s) = \begin{cases}
s \tilde{F}_{\alpha -1} \tilde{t}_{\alpha}^{\perp} & \text{ if  } s < 0 \\
s \tilde{F}_{\alpha} \tilde{t}_{\alpha}^{\perp} & \text{ if } s > 0
\end{cases}, \quad \alpha \in \{1, \ldots K\},
\end{align}
and for any set of matrices $\tilde{F}_0, \tilde{F}_1, \ldots, \tilde{F}_K \in \mathbb{R}^{3 \times 2}$ having the properties:
\begin{equation}
\begin{aligned}\label{eq:CompatCons}
&\tilde{F}_{\alpha-1} \neq \tilde{F}_{\alpha}, \quad (\tilde{F}_{\alpha} - \tilde{F}_{\alpha - 1}) \tilde{t}_{\alpha} = 0 \quad \text{ and }  \\
&\lambda \adj \tilde{F}_{\alpha-1} + (1-\lambda) \adj \tilde{F}_{\alpha} \neq 0, \quad \forall \;\;  \lambda \in [0,1]
\end{aligned}
\end{equation}
for each $\alpha \in \{ 1,\ldots K\}$ (if $\tilde{p}$ is an interior junction, then $\tilde{F}_0 = \tilde{F}_K$).  Here, we introduce the notation 
\begin{align}
\adj \tilde{F} := \tilde{F} \tilde{e}_1 \times \tilde{F} \tilde{e}_2.
\end{align} 

The first condition in (\ref{eq:CompatCons}) ensures that each $\tilde{t}_{\alpha}$ interface is a non-trivial (i.e., there is a jump in the deformation gradient across the interface).  The second condition in (\ref{eq:CompatCons}) is the rank-one compatibility condition which ensures that $y$ is continuous across each $\tilde{t}_{\alpha}$ interface.  Finally, the latter condition in (\ref{eq:CompatCons}) ensures that adjoining regions do not fold into themselves.

\subsection{On $\delta$-smoothings of a single junction}

We now show the existences of a $\delta$-smoothing for a special class of origami junctions where the $F_{\alpha}$ satisfy an algebraic condition.  The general problem of finding $\delta$-smoothings for any junction that satisfies (\ref{eq:yAffine})-(\ref{eq:CompatCons}) remains open.  However, our special class covers the examples of physical interest.  

To introduce our result, we recall that the convex hull of a finite collection of $\mathbb{R}^{3\times2}$ matrices is
\begin{align}
\co \{ \tilde{F}_{1}, \ldots, \tilde{F}_{N}\} := \Big\{ \sum_{i = 1}^{N} \lambda_{i} \tilde{F}_{i} \colon \lambda_{i} \geq 0 \;\; \text{ for each $i$ and }\;\; \sum_{i = 1}^{N} \lambda_{i} = 1\Big\}.
\end{align}
In addition, for any collection $\mathcal{S}$ of $\mathbb{R}^{3 \times 2}$ matrices, we denote the lower rank of the matrices in this set as 
\begin{align}
\rank_l \mathcal{S} := \min\{ \rank \tilde{F} \colon \tilde{F} \in \mathcal{S} \}. 
\end{align}
Our main result on $\delta$-smoothings of generic origami deformation of a single junction is as follows:
\begin{theorem}\label{DSmoothTheorem}
Let $\omega$ be a single junction (as defined above) and let $y$ be an origami deformation of this junction defined by (\ref{eq:yAffine}), (\ref{eq:gammaAlpha}) and (\ref{eq:CompatCons}).  Consider the set 
\begin{align}\label{eq:AySet}
\mathcal{A}_y := \left\{ \tilde{F} \in \mathbb{R}^{3\times2}  \colon \rank_{l} \left(\bigcup_{\alpha = 1}^{K} \co\{ \tilde{F}, \tilde{F}_{\alpha}, \tilde{F}_{\alpha-1}\} \right) = 2 \right\}.
\end{align}
If $\mathcal{A}_y$ is non-empty, then $y$ has a $\delta$-smoothing.  
\end{theorem}
\noindent We prove this result in two steps Proposition \ref{1DMollProp} and \ref{propReduce} below.  

First, consider a mollification of each $\gamma_{\alpha}$ in (\ref{eq:gammaAlpha}), i.e., $\gamma_{\alpha}^{\delta} \in C^{\infty}(\mathbb{R},\mathbb{R}^3)$ given by 
\begin{align}\label{eq:gammaAlphaDelta}
\gamma_{\alpha}^{\delta} := \gamma_{\alpha} \ast \eta_{\delta}, \quad \alpha \in \{ 1, \ldots, K\}
\end{align}
for $\eta_{\delta} \in C^{\infty}(\mathbb{R},\mathbb{R})$ the standard symmetric mollifier supported on the interval $(-\delta/2,\delta/2)$.   For any $\delta > 0$, we define the function $y_0^{\delta} \colon \omega \rightarrow \mathbb{R}^3$ given by 
\begin{align}\label{eq:y0delta}
y_0^{\delta}(\tilde{x}) := \gamma^{\delta}_{\alpha}((\tilde{x}-\tilde{p}) \cdot \tilde{t}_{\alpha}^{\perp}) + ((\tilde{x}-\tilde{p}) \cdot \tilde{t}_{\alpha}) \tilde{F}_{\alpha} \tilde{t}_{\alpha} + y(\tilde{p}) \quad \text{ if } \quad \tilde{x} \in S_{\alpha}, \quad \alpha \in \{ 1, \ldots K\}.
\end{align}
This is a $\delta$-smoothing of $y$ outside of a small neighborhood of the junction.  
\begin{proposition}\label{1DMollProp}
Let $\omega$ and $y$ be as in Theorem \ref{DSmoothTheorem}.  Set $m > 1/\sin(\theta^{\ast}/2)$ for $\theta^{\ast}$ in (\ref{eq:thetaAst}). For any $\delta > 0$, define $y_0^{\delta}$ as in (\ref{eq:y0delta}).  Then $y_0^{\delta}$ restricted to $\omega \setminus B_{m \delta}(\tilde{p})$ is a $\delta$-smoothing of $y$.  Moreover,
\begin{align}\label{eq:identGy0}
\tilde{\nabla} y_0^{\delta}(\tilde{x}) = (1- \lambda_{\delta}((\tilde{x}-\tilde{p}) \cdot \tilde{t}_{\alpha}^{\perp})) \tilde{F}_{\alpha-1} + \lambda_{\delta}((\tilde{x}- \tilde{p}) \cdot \tilde{t}_{\alpha}^{\perp}) \tilde{F}_{\alpha}, \quad \text{ if } \quad \tilde{x} \in S_{\alpha} \setminus B_{m\delta}(\tilde{p})
\end{align}
where $\lambda_{\delta} \in C^{\infty}(\mathbb{R}, [0,1])$ is given by $\lambda_{\delta}(s) := \int_{-\delta/2}^s \eta_{\delta}(t) dt$.
\end{proposition}

The issue with this construction, however, is that $y_0^{\delta}$ is not even continuous on $B_{m\delta}(\tilde{p})$, so we require a modification of this deformation in a neighborhood of $\tilde{p}$ for a $\delta$-smoothing on all of $\omega$.  From now on, we assume that the set $\mathcal{A}_y$ in (\ref{eq:AySet}) is non-empty. We replace $y_0^{\delta}$ on the $B_{m\delta}(\tilde{p})$ by 
\begin{align}\label{eq:yc}
y_c(\tilde{x}) = \tilde{F}_c (\tilde{x} - \tilde{p}) + y(\tilde{p}), \quad \tilde{F}_{c} \in \mathcal{A}_y.
\end{align} 
Then for any $\delta > 0$, we define $y^{\delta} \colon \omega \rightarrow \mathbb{R}^3$ as 
\begin{align}\label{eq:ydelta}
y^{\delta}(\tilde{x}) = y_0^{\delta}(\tilde{x}) + \psi_{\delta}(|\tilde{x}- \tilde{p}|) (y_c(\tilde{x}) - y_0^{\delta}(\tilde{x})) 
\end{align}
for some cutoff function $\psi_{\delta}$ such that  
\begin{equation}
\begin{aligned}\label{eq:cutoff}
\psi_{\delta} \in C^{3}(\mathbb{R}, [0,1]), \quad \psi_{\delta}(s) = \begin{cases}
1 & \text{ if } s < m \delta \\
0 & \text{ if } s > M \delta
\end{cases}, \quad M > m > 1/ \sin(\theta^{\ast}/2) 
\end{aligned}
\end{equation} 
We make the following observation about this construction:
\begin{proposition}\label{propReduce}
Let $\omega$ and $y$ be as in Theorem \ref{DSmoothTheorem}.  Let $\mathcal{A}_y$ be non-empty.  For any $\delta > 0$ define $y^{\delta}$ as (\ref{eq:ydelta}) for $y_0^{\delta}$ in (\ref{eq:y0delta}) with $y_c$ as in (\ref{eq:yc}).   There exists a $\psi_{\delta}$ satisfying (\ref{eq:cutoff}) such that $y^{\delta}$ is a $\delta$-smoothing.  
\end{proposition}

\begin{proof}[Proof of Theorem \ref{DSmoothTheorem}.]
The theorem follows directly from Proposition \ref{propReduce}.
\end{proof}
\begin{figure}
  \centering
  \includegraphics[height = 2.3 in]{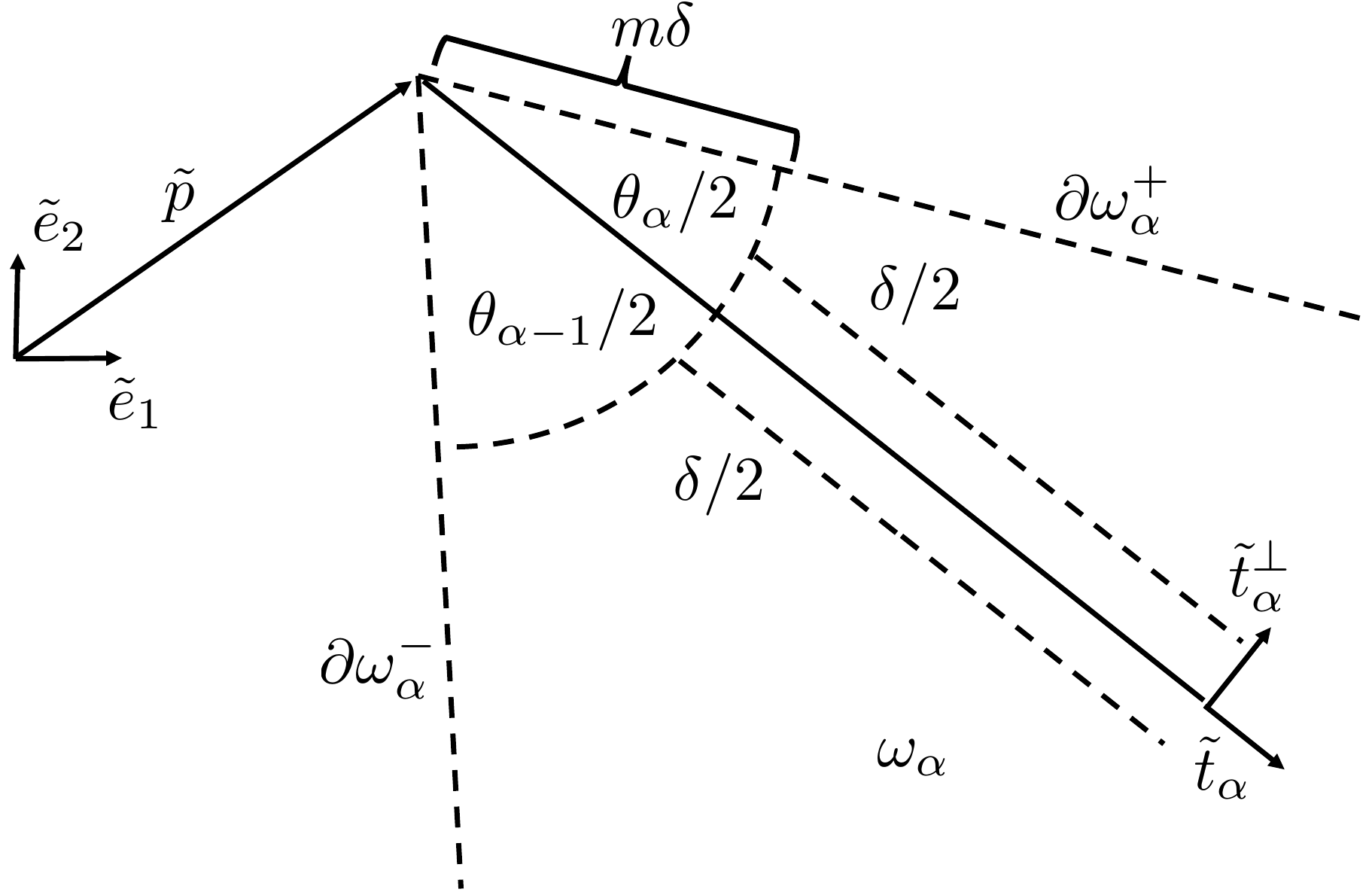}
  \caption{Schematic of an arbitrary $\tilde{t}_{\alpha}$ interface for a junction at point $\tilde{p}$.}
  \label{fig:1DMollifier}
\end{figure}

It remains to prove Propositions \ref{1DMollProp} and \ref{propReduce}:
\begin{proof}[Proof of Proposition \ref{1DMollProp}.]
Let $\omega_{m} := \omega \setminus B_{m\delta}(\tilde{p})$, and consider any $\delta$ sufficiently small so that $\omega_m$ is non-empty.  Since $m \geq 1/\sin(\theta^{\ast}/2) \geq 1/\sin(\theta_{\alpha}/2)$ for each $\alpha \in \{1,\ldots,K-1\}$ (and $K$ if $\omega$ is an interior junction), $y_0^{\delta}$ defined in (\ref{eq:y0delta}) is equal to $y$ across each $\partial S_{\alpha}^{\pm} \cap \omega_m$ (see the schematic in Figure \ref{fig:1DMollifier}) and $y$ is smooth across these interfaces.  Therefore, we need only to show that $y_0^{\delta}$ is a $\delta$-smoothing of $y$ on each $\omega_{m} \cap S_{\alpha}$ to prove it is a $\delta$-smoothing of $y$ on $\omega_m$.  

Fix an $S_{\alpha} \subset \omega$.  Since $\gamma_{\alpha}^{\delta}$ is a $\delta$-mollification of a piecewise affine and continuous function $\gamma_{\alpha}$ as defined by (\ref{eq:gammaAlpha}) and (\ref{eq:gammaAlphaDelta}), it follows that 
\begin{align*}
&y_0^{\delta} = y \quad \text{ except on a } \delta\text{-strip} \subset \omega_m \cap S_{\alpha} \quad \text{ and } \\
&|\tilde{\nabla} y_0^{\delta}| \leq C, \quad |\tilde{\nabla} \tilde{\nabla}  y_0^{\delta}| \leq C\delta^{-1}, \quad |\tilde{\nabla}^{(3)} y_0^\delta| \leq C \delta^{-2} \quad \text{ on } \omega_m \cap S_{\alpha}
\end{align*}
for some $C> 0$ independent of $\delta$.  For the lower bound constraint on the cross-product, we observe that 
\begin{align*}
\gamma_{\alpha}^{\delta}(s) =  \tilde{F}_{\alpha} \tilde{t}_{\alpha}^{\perp} \int_{-\delta/2}^{s} \eta_{\delta}(t) (s-t) dt  + \tilde{F}_{\alpha-1} \tilde{t}_{\alpha}^{\perp} \int_{s}^{\delta/2} \eta_{\delta}(t) (s-t) dt,
\end{align*}
and thus
\begin{align*}
(\gamma_{\alpha}^{\delta})'(s) = (1- \lambda_{\delta}(s)) \tilde{F}_{\alpha-1} \tilde{t}_{\alpha}^{\perp} + \lambda_{\delta}(s) \tilde{F}_{\alpha} \tilde{t}_{\alpha}^{\perp}
\end{align*}
for $\lambda_{\delta}$ defined in the proposition.  Since $\tilde{\nabla} y_0^{\delta}(\tilde{x}) = (\gamma_{\alpha}^{\delta})'((\tilde{x}-\tilde{p}) \cdot \tilde{t}_{\alpha}^{\perp}) \otimes \tilde{t}_{\alpha}^{\perp} + \tilde{F}_{\alpha} \tilde{t}_{\alpha} \otimes \tilde{t}_{\alpha}$ for  $\tilde{x} \in S_{\alpha}$, we obtain (\ref{eq:identGy0}) by direct substitution of $(\gamma_{\alpha}^{\delta})'$ above and using the fact that $\tilde{F}_{\alpha-1}\tilde{t}_{\alpha}= \tilde{F}_{\alpha} \tilde{t}_{\alpha}$.  Finally, noting that $\adj(\tilde{F}\tilde{R}) = \det(\tilde{R}) \adj(\tilde{F}) = \adj\tilde{F}$ for any $\tilde{F} \in \mathbb{R}^{3\times2}$ and $\tilde{R} \in SO(3)$, we find
\begin{align*}
\adj(\tilde{\nabla} y_0^{\delta}(\tilde{x})) &= \tilde{\nabla}y_0^{\delta}(\tilde{x}) \tilde{t}_{\alpha} \times \tilde{\nabla}y_0^{\delta}(\tilde{x}) \tilde{t}_{\alpha}^{\perp} \\
&= \tilde{F}_{\alpha} \tilde{t}_{\alpha} \times \left((1-\lambda_{\delta}((\tilde{x}-\tilde{p})\cdot \tilde{t}_{\alpha}^{\perp})) \tilde{F}_{\alpha-1} \tilde{t}_{\alpha}^{\perp}  + \lambda_{\delta}((\tilde{x}-\tilde{p}) \cdot \tilde{t}_{\alpha}^{\perp} ) \tilde{F}_{\alpha} \right) \\
&= (1- \lambda_{\delta}((\tilde{x}-\tilde{p}) \cdot \tilde{t}_{\alpha}^{\perp})) \adj \tilde{F}_{\alpha-1} + \lambda_{\delta}((\tilde{x}-\tilde{p}) \cdot \tilde{t}_{\alpha}^{\perp}) \adj \tilde{F}_{\alpha},
\end{align*}
for any $\tilde{x} \in \omega_{m} \cap S_{\alpha}$ making repeated use of the fact that $\tilde{F}_{\alpha-1} \tilde{t}_{\alpha} = \tilde{F}_{\alpha} \tilde{t}_{\alpha}$.   By hypothesis (\ref{eq:CompatCons}), this is bounded away from zero since $\lambda_{\delta} \in [0,1]$.  That is, we conclude $|\adj \tilde{\nabla} y_{0}^{\delta}| \geq c_{\alpha} > 0$  on $S_{\alpha}$.  Here, $\alpha \in \{1,\ldots,K\}$ was arbitrary and so the proof is complete.   
\end{proof}

\begin{proof}[Proof of Proposition \ref{propReduce}.]
We note that for any $m,M$ such that $M > m > 1/\sin(\theta^{\ast}/2)$ and for any $\delta> 0$ sufficiently small, 
\begin{align}\label{eq:gradyDelta}
\tilde{\nabla} y^{\delta} = \tilde{F}^{\delta} + \tilde{G}^{\delta} \quad \text{ on } \omega,
\end{align}
where $\tilde{F}^{\delta}, \tilde{G}^{\delta} \colon \omega \rightarrow \mathbb{R}^{3\times2}$ are given by 
\begin{equation}\label{eq:Gdelta}
\begin{aligned}
&\tilde{F}^{\delta}(\tilde{x}) := (1- \psi_{\delta}(|\tilde{x}- \tilde{p}|)) \tilde{\nabla} y_0^{\delta}(\tilde{x}) + \psi_{\delta}(|\tilde{x}- \tilde{p}|) \tilde{F}_c, \\
&\tilde{G}^{\delta}(\tilde{x}) := |\tilde{x}-\tilde{p}|^{-1} \psi_{\delta}'(|\tilde{x}-\tilde{p}|) (y_c(\tilde{x}) - y_0^{\delta}(\tilde{x})) \otimes (\tilde{x}-\tilde{p}).
\end{aligned}
\end{equation}

Focusing first on $\tilde{F}^{\delta}$, we note that for any $S_{\alpha} \subset \omega$, 
\begin{align*}
\tilde{F}^{\delta}(\tilde{x}) = (1- \psi_{\delta}(|\tilde{x} -\tilde{p}|))& \left((1-\lambda_{\delta}((\tilde{x}-\tilde{p}) \cdot  \tilde{t}_{\alpha}^{\perp})) \tilde{F}_{\alpha-1} + \lambda_{\delta}((\tilde{x} -\tilde{p}) \cdot \tilde{t}_{\alpha}^{\perp}) \tilde{F}_{\alpha} \right) \\
&+ \psi_{\delta}(|\tilde{x}-\tilde{p}|)\tilde{F}_{c}, \quad \tilde{x} \in S_{\alpha}
\end{align*}
using (\ref{eq:identGy0}) from Proposition \ref{1DMollProp}. Consequently, 
\begin{align}\label{eq:FdeltaProp}
\tilde{F}^{\delta}(\tilde{x}) \in \text{co}\{ \tilde{F}_c, \tilde{F}_{\alpha-1}, \tilde{F}_{\alpha}\}, \quad \tilde{x} \in S_{\alpha}, \quad \alpha \in \{ 1, \ldots, K-1\}
\end{align}
(and $K$ if $\omega$ is an interior junction) since $\psi_{\delta}, \lambda_{\delta}$ map to $[0,1]$.  We claim that  (\ref{eq:FdeltaProp}) implies 
\begin{align}\label{eq:adjFdIdent}
|\adj \tilde{F}^{\delta} | \geq c^{\ast} \quad \text{ on } \omega
\end{align}
for some $c^{\ast} > 0$.  

To see this, we define $f_{\alpha} \colon \Lambda_3 \rightarrow \mathbb{R}$ as 
\begin{align*}
f_{\alpha}(\lambda_1, \lambda_2, \lambda_3) := |\adj(\lambda_1 \tilde{F}_c + \lambda_2 \tilde{F}_{\alpha-1} + \lambda_3 \tilde{F}_{\alpha})| 
\end{align*}
where $\Lambda_3 := \{(\lambda_1, \lambda_2, \lambda_3) \in \mathbb{R}^3 \colon \lambda_i \geq 0 \;\; \text{ for each $i$ and } \sum_{i =1}^{3} \lambda_i = 1\}$. $\Lambda_3$ is a compact subset of $\mathbb{R}^3$ and each $f_{\alpha}$ is continuous on $\Lambda_3$.  Thus, the infimum of each $f_{\alpha}$ is attained.  We denote 
\begin{align*}
(\lambda_1^{\alpha}, \lambda_2^{\alpha}, \lambda_{3}^{\alpha}) := \arg \min_{\Lambda_3} f_{\alpha}, \quad \alpha \in \{1,\ldots,K-1\}. 
\end{align*}
(and $K$ if $\omega$ is an interior junction).  Since $\tilde{F}_c \in \mathcal{A}_y$, each $\lambda_1^{\alpha} \tilde{F}_c + \lambda_2^{\alpha} \tilde{F}_{\alpha-1} + \lambda_3^{\alpha} \tilde{F}_{\alpha} \in \mathbb{R}^{3\times2}$ is full rank and thus, we can take 
\begin{align}\label{eq:cast}
c^{\ast} := \min_{\alpha \in \{1,\ldots,K-1\} } \min_{\Lambda_3} f_{\alpha} > 0
\end{align}
(again minimizing over $K$ as well if $\omega$ is an interior junction)  to achieve the identity (\ref{eq:adjFdIdent}).

Now, for the lower bound estimate of a $\delta$-smoothing, we notice that given the representations (\ref{eq:gradyDelta}) and (\ref{eq:Gdelta}), 
\begin{align}\label{eq:lbAdj}
|\adj \nabla y^{\delta}| &\geq |\adj \tilde{F}^{\delta}| - |\tilde{G}^{\delta}|(2|\tilde{F}^{\delta}| + |\tilde{G}^{\delta}|)  \nonumber \\
& \geq c^{\ast} - C|\tilde{G}^{\delta}|(1+ |\tilde{G}^{\delta}|) \quad \text{ on } \omega.
\end{align}
The latter constant $C$ is independent of $M,m$ and $\delta$ since $\|\tilde{F}^{\delta}\|_{L^{\infty}}$ can be bounded uniformly independent of these quantities following (\ref{eq:FdeltaProp}).  

Now, for estimating $\tilde{G}^{\delta}$ in (\ref{eq:Gdelta}) with this cutoff function, we notice first that 
\begin{align*}
|y_0^{\delta} - y_c| \leq |y - y(\tilde{p})| + |y_c - y(\tilde{p})| + |y_0^{\delta} - y|  \leq C(\delta + |\tilde{x}-\tilde{p}|) \quad \text{ on }  B_{M \delta}(\tilde{p}) \setminus B_{m\delta}(\tilde{p}).
\end{align*}
To obtain this estimate, we used that $|y_0^{\delta} -y | = |\gamma_{\alpha}^{\delta} - \gamma_{\alpha}| \leq C\delta$ on each $S_{\alpha} \setminus B_{m\delta}(\tilde{p})$ since $\gamma_{\alpha}$ is Lipschitz continuous, and we used that both $y$ and $y_c$ are equal to $y(\tilde{p})$ at $\tilde{x} = \tilde{p}$ and Lipschitz continuous with uniform Lipschitz constant on $B_{M\delta}(\tilde{p}) \setminus B_{m\delta}(\tilde{p})$.  Moreover, $\tilde{G}^{\delta}$ is only non-zero on that annulus $B_{M\delta}(\tilde{p}) \setminus B_{m\delta}(\tilde{p})$ since $\psi'_{\delta}(|\tilde{x} - \tilde{p}|) = 0$ outside this annulus.  Hence, we observe that 
\begin{align}\label{eq:gdeltaEst}
|\tilde{G}^{\delta}(\tilde{x})| &\leq |\psi_{\delta}'(|\tilde{x}-\tilde{p}|)||y_{0}^{\delta}(\tilde{x}) - y_c(\tilde{x})| \leq C|\psi_{\delta}'(|\tilde{x}-\tilde{p}|)|(\delta + |\tilde{x}-\tilde{p}|) \nonumber \\
& \leq C(1/m + 1) ||\tilde{x}-\tilde{p}|\psi_{\delta}'(|\tilde{x}-\tilde{p}|)|  \leq C\|s \psi_{\delta}'(s)\|_{L^{\infty}}.
\end{align}
where in the second to last estimate we use that $|\tilde{x}-\tilde{p}|/m\delta > 1$ on the annulus $B_{M\delta}(\tilde{p})\setminus B_{m\delta}(\tilde{p})$, and all constants $C> 0$ above can be chosen uniform independent of $\delta$ and $M > m > 1/\sin(\theta^{\ast}/2)$. Hence, by applying Lemma \ref{cutoffLemma} (below), we suitably choose $m,M$ and the cutoff function $\psi_{\delta}$ to establish the estimate 
\begin{align*}
|\adj \tilde{\nabla} y^{\delta}| \geq c^{\ast}/2 \quad \text{ on } \omega 
\end{align*} 
for all $\delta > 0$ sufficiently small.  Here, we made use of (\ref{eq:lbAdj}) and (\ref{eq:gdeltaEst}).

With this lower bound established, the other properties which show $y^{\delta}$ is a $\delta$-smoothing of $y$ are easily verified: Indeed, $y = y^{\delta}$ except on a set of measure $O(\delta)$ for all $\delta$ sufficiently small since $y^{\delta}$ deviates from $y_0^{\delta}$ only on a set of $O(\delta^2)$.  Moreover, the derivative estimates follow from the chain rule and using the estimates for the cutoff function $\psi_{\delta}$ established in Lemma \ref{cutoffLemma} below.  This completes the proof.
\end{proof}

\begin{lemma}\label{cutoffLemma}
Fix $\epsilon >0$.  There is a $\Delta_{\epsilon} > 0$ such that for any $M> m>1/\sin(\theta^{\ast}/2)$ satisfying $M/m \geq \Delta_{\epsilon}$  and $M-m > 2$, there exists for all $\delta> 0$ a cutoff function $\psi_{\delta}$ satisfying (\ref{eq:cutoff}) with the properties
\begin{equation}\label{eq:lemmaCutoff}
\begin{aligned}
&\|s \psi_{\delta}'(s)\|_{L^{\infty}} \leq \epsilon \\
&\|\psi_{\delta}'\|_{L^{\infty}}  \leq C \delta^{-1}, \quad \|\psi_{\delta}''\|_{L^{\infty}} \leq C \delta^{-2}  \quad \text{ and }\quad \|\psi_{\delta}'''\|_{L^{\infty}} \leq C \delta^{-3}
\end{aligned}
\end{equation}
for $C \equiv C(M,m) > 0$ independent of $\delta$.  
\end{lemma}
\begin{proof}
Consider the cutoff function $\tilde{\psi}_{\delta} \colon \mathbb{R} \rightarrow [0,1]$ given by 
\begin{align*}
\tilde{\psi}_{\delta}(s) := \begin{cases}
1 & \text{ if } s < (m+1/2) \delta \\
\frac{\log(s/(M-1/2)\delta)}{\log((m+1/2)/(M-1/2))} & \text{ if } s \in [(m+1/2)\delta,(M-1/2)\delta] \\
0 & \text{ if } s > (M-1/2) \delta
\end{cases}.
\end{align*}
Here, $\tilde{\psi}_{\delta}$ is Lipschitz continuous since $M-m >2$ and equal to $1$ in a neighborhood of the origin since $m \geq 1/\sin(\theta^{\ast}/2) \geq 1$.   This is not a cutoff function with the properties (\ref{eq:cutoff}).  However, importantly 
\begin{align*}
|s \tilde{\psi}_{\delta}'(s)| = |\log\Big(\frac{M -1/2}{m+1/2}\Big)|^{-1} =: \epsilon_{M,m}
\end{align*}
which can be made arbitrarily small (independent of $\delta$) by choosing $M/m$ sufficiently large.  By mollification, we can retain a similar estimate for a $\psi_{\delta}$ as in (\ref{eq:cutoff}).  

Indeed, we let $\psi_{\delta} := \eta_{\delta} \ast \tilde{\psi}_{\delta}$ for $\eta_{\delta} \in C^{\infty}(\mathbb{R}, \mathbb{R})$ the standard symmetric mollifier supported on the interval $(-\delta/2,\delta/2)$.  Since $\tilde{\psi}_{\delta}$ is Lipschitz continuous, equal to $1$ for $s < (m+1/2) \delta$ and equal to $0$ for $s > (M-1/2)\delta$, $\psi_{\delta}$ is a cutoff function satisfying all the properties in (\ref{eq:cutoff}) and it satisfies the latter estimates in (\ref{eq:lemmaCutoff}).   

It remains to prove the first estimate in (\ref{eq:lemmaCutoff}).  To this end, note $\psi_{\delta}' = \eta_{\delta} \ast \tilde{\psi}_{\delta}'$, and so explicitly
\begin{align*}
|s \psi_{\delta}'(s)| \leq  \epsilon_{M,m}\begin{cases}
 \int_{(m+1/2) \delta}^{s+ \delta/2} \eta_{\delta}(s-t) \big| \frac{s}{t} \big| dt & \text{ if } s \in (m\delta, (m+1) \delta] \\
 \int_{s - \delta/2}^{s+ \delta/2} \eta_{\delta}(s-t) \big| \frac{s}{t} \big| dt & \text{ if } s \in ((m+1)\delta, (M-1) \delta] \\
 \int_{s - \delta/2}^{(M-1/2)\delta} \eta_{\delta}(s-t) \big| \frac{s}{t} \big| dt & \text{ if } s \in ((M-1)\delta, M \delta]\\ 
0 & \text{ otherwise}.  
\end{cases}
\end{align*}
From this, we deduce that 
\begin{align}\label{eq:tpsiD}
\|s \psi_{\delta}'(s)\|_{L^{\infty}} \leq \left(\frac{m+1}{m+1/2}\right) \epsilon_{M,m} \leq 2 \epsilon_{M,m}.
\end{align}
Thus, there is a $\Delta_{\epsilon}$ such that if $M/m \geq \Delta_{\epsilon}$, $\epsilon_{M,m} \leq \epsilon/2$.  This completes the proof given (\ref{eq:tpsiD}).
\end{proof}

\subsection{Examples of noniosmetric origami and their $\delta$-smoothings.}
In this section, we examine the nonisometric origami to actuate a box, rhombic dodecahedron and rhombic triacontahedron.  We will show that each of these designs has a corresponding $\delta$-smoothing.  In this direction, consider Figure \ref{fig:SolveExamples} showing for the case of cooling a nematic elastomer sheet: (a) the design to actuate a box, (b) part of the design to actuate the rhombic dodecahedron and (c) part of the design to actuate the rhombic triacontahedron.  In each case, there are only two non-trivial junctions to consider, each highlighted in red.  That is, once the deformation (both the origami and $\delta$-smoothing deformations) are constructed for these junctions, then the entire deformation can be built as rotations and translations of these constructions. 

\begin{figure}
\centering
\includegraphics[width = 6in]{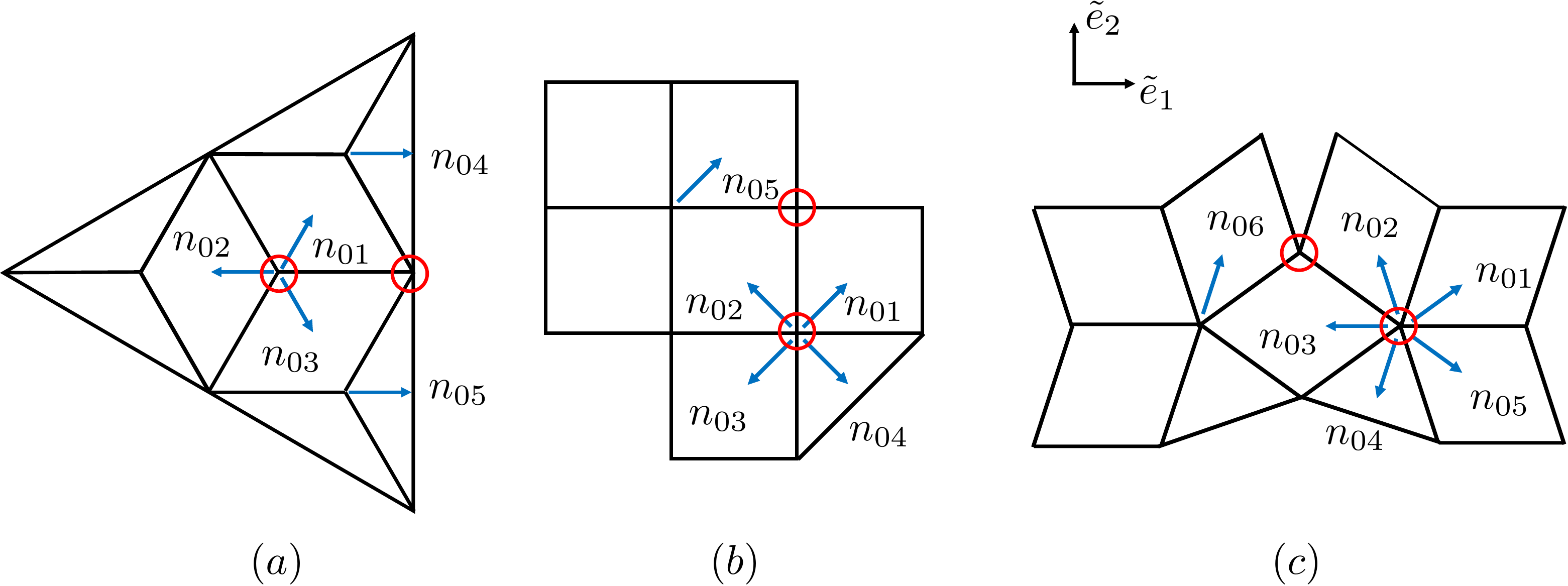}
\caption{Designs to actuate a box (a), rhombic dodecahedron (b) and rhombic triacontahedron (c) upon cooling (i.e., $r > 1$).  Only a portion of the design is shown in (b) and (c).}
\label{fig:SolveExamples}
\end{figure}  

As a first step towards constructing a $\delta$-smoothing for these actuations, we identify the deformation gradients associated with the origami.  This makes use of the notion of compatibility discussed in Appendix \ref{app:no}.
\begin{proposition}\label{identDefs}
Consider the designs depicted in Figure \ref{fig:SolveExamples}.  Up to a rigid body rotation, the deformation gradients corresponding to each region are given by
\begin{enumerate}[(a)]
\item For the box
\begin{align*}
 \tilde{F}_{\alpha} = \begin{cases} R_{n_{0\alpha}}( \ell_{n_{0\alpha}}^{1/2})_{3 \times 2}, &\alpha \in \{1,2,3\}  \\
R_{n_{01}}^2 R_{n_{04}} (\ell_{n_{04}}^{1/2})_{3\times2} & \alpha = 4\\
R_{n_{03}}^2 R_{n_{05}} (\ell_{n_{05}}^{1/2})_{3\times2} & \alpha = 5;
\end{cases}
\end{align*}
\item for the rhombic dodecahedron
\begin{align*}
 \tilde{F}_{\alpha} = \begin{cases} R_{n_{0\alpha}}( \ell_{n_{0\alpha}}^{1/2})_{3 \times 2} &\alpha \in \{1,\ldots,4\}  \\
R_{n_{02}}^2 R_{n_{05}} (\ell_{n_{05}}^{1/2})_{3\times2} & \alpha = 5;
\end{cases}
\end{align*}
\item for the rhombic triacontahedron
\begin{align*}
 \tilde{F}_{\alpha} = \begin{cases} R_{n_{0\alpha}}( \ell_{n_{0\alpha}}^{1/2})_{3 \times 2} &\alpha \in \{1,\ldots,5\}  \\
R_{n_{03}}^2 R_{n_{06}} (\ell_{n_{06}}^{1/2})_{3\times2} & \alpha = 6.
\end{cases}
\end{align*}
\end{enumerate}
Here, each $R_{n_{0\alpha}} \in SO(3)$ satisfies $R_{n_{0\alpha}} n_{0\alpha} = r^{-1/2} n_{0\alpha} + \sqrt{1-r^{-1}} e_3$ for $r \in [1,r_{max}]$.
\end{proposition}
\begin{proof} 
Clearly in each region $\tilde{F}_{\alpha}^T \tilde{F}_{\alpha} = \tilde{\ell}_{n_{0\alpha}}$.  Thus, we need only to show that the deformation gradients are rank-one compatible at each interface.  Let $\tilde{t}_{\alpha} \in \mathbb{S}^1$ denote the outward normal (to the junction) for the interface separating $n_{0\alpha}$ and $n_{0(\alpha-1)}$ for the interior junctions as depicted in Figure \ref{fig:SolveExamples} (i.e., $\alpha \in \{1,\ldots,K\}$ with $K = 3$ (box), $4$ (rhombic dodecahedron), $5$ (rhombic triacontahedron) where we set $n_{00} = n_{0K}$ in each case).   We have that $\tilde{n}_{0\alpha} = \cos(\theta) \tilde{t}_{\alpha} + \sin(\theta) \tilde{t}_{\alpha}^{\perp}$ where $\theta \in \{\pi/3,\pi/4,\pi/5\}$ for the box, rhombic dodecahedron and rhombic triacontahedron respectively with $\tilde{t}_{\alpha}^{\perp}$ is the right-hand orthonormal vector to $\tilde{t}_{\alpha}$.  

Now, to verify interface compatibility, let us first assume only that $R_{n_{0\alpha}} n_{0\alpha} = \cos(\varphi_r) n_{0\alpha} + \sin(\varphi_r) e_3$.   By explicit computation, we find that 
\begin{align}
\left(R_{n_{0\alpha}} (\ell_{n_{0\alpha}}^{1/2})_{3\times2} - R_{n_{0(\alpha-1)}} (\ell_{n_{0(\alpha-1)}}^{1/2})_{3\times2}\right) \tilde{t}_{\alpha} = 2 \cos(\theta) \sin(\theta)( r^{1/3} \cos(\varphi_r) - r^{-1/6}) t_{\alpha}^{\perp}.
\end{align}
Thus, interface compatibility (i.e., this quantity being equal to zero) is achieved with $\cos(\varphi_r) = r^{-1/2}$, and this gives the condition on each $R_{n_{0\alpha}} \in SO(3)$ defined in the proposition.  

It remains to verify compatibility for the exterior junctions.  Let us focus on the $n_{04}$ case for the box in (a).  Notice that if we consider the interior junction which contains the $n_{04}$ sector without compatibility of the entire origami structure, by the previous argument on interior junctions, we find that the junction in isolation is compatible given
\begin{align*}
\tilde{F}^{\ast}_{1} = R_{n_{01}}^T (\ell_{n_{01}}^{1/2})_{3\times2}, \quad \tilde{F}^{\ast}_4 = R_{n_{04}} (\ell_{n_{04}}^{1/2})_{3 \times 2}.
\end{align*}
The transpose for $\tilde{F}^{\ast}_1$ is since $n_{01}$ points toward this junction and not away from it.   For compatibility of the whole structure, we notice that $\tilde{F}_1 = R_{n_{01}}^2 \tilde{F}_1^{\ast}$, and so rigidly rotating this isolated compatible junction by $R_{n_{01}}^2$ achieves a fully compatible structure.   This gives $\tilde{F}_4$ in the proposition for (a).  An analogous argument holds for all the other exterior junction cases. 
\end{proof}

Now, for a $\delta$-smoothing of the deformation, we claim first:
\begin{proposition}\label{InterProp}
Each interior junction in (a), (b) and (c) has a $\delta$-smoothing.
\end{proposition}
\begin{proof}
By Proposition \ref{identDefs} and Theorem \ref{DSmoothTheorem}, we prove this result if we can find a $\tilde{F}_K \in \mathbb{R}^{3 \times 2}$ such that the set 
\begin{align*}
\mathcal{S}_{int}^K := \bigcup_{\alpha = 1}^{K} \co \big\{ \tilde{F}_K, R_{n_{0\alpha}} (\ell_{n_{0\alpha}}^{1/2})_{3 \times 2}, R_{n_{0(\alpha-1)}} (\ell_{n_{0(\alpha-1)}}^{1/2})_{3 \times 2}\big\}
\end{align*}
contains only matrices of full rank.  Here, $K \in \{ 3,4,5\} $ is for the box, rhombic dodecahedron and rhombic triacontahedron respectively and $n_{00} = n_{0K}$ for each $K$.  

The choice of $\tilde{F}_K$ which gives $\rank_{l} (\mathcal{S}_{int}^K) =2$ is facilitated by the following observation.  Consider $\tilde{v} := \beta_1 \tilde{n}_{0\alpha} + \beta_2 \tilde{n}_{0\alpha}^{\perp}$ for any $(\beta_1, \beta_2) \in \mathbb{R}^2$ (i.e., $\tilde{v}$ is an arbitrary vector on $\mathbb{R}^2$).  We observe that 
\begin{align*}
R_{n_{0\alpha}} (\ell_{n_{0\alpha}}^{1/2})_{3 \times 2} \tilde{v} &= R_{n_{0\alpha}} (r^{1/3} \beta_1 n_{0\alpha} + r^{-1/6} \beta_2 n_{0\alpha}^{\perp})  \\
&= r^{-1/6}(\beta_1 n_{0\alpha} + \beta_2 n_{0\alpha}^{\perp}) + r^{1/3} \sqrt{1-r^{-1}} e_3 \\
&=r^{-1/6} v + r^{1/3} \sqrt{1-r^{-1}} e_3.
\end{align*}
Thus, since $\tilde{v}$ was arbitrary and $\alpha$ was arbitrary, we conclude that 
\begin{align*}
\mathbb{P}_{e_3}( R_{n_{0\alpha}} (\ell_{n_{0\alpha}}^{1/2})_{3 \times 2}) = r^{-1/6} I_{2 \times 2}, \quad \text{ for each $\alpha$ and $K$}
\end{align*}
where $\mathbb{P}_{e_3} \colon \mathbb{R}^{3\times2} \rightarrow \mathbb{R}^{2\times2}$ projects any $\mathbb{R}^{3\times2}$ matrix to that plane normal to $e_{3}$.

Now, if we choose $\tilde{F}_{K} = r^{-1/6} I_{3\times2}$ for each $K \in \{3,4,5\}$, we notice that for any $\lambda, \mu \in [0,1]$
\begin{align*}
\mathbb{P}_{e_3}&\left(\lambda r^{-1/6} I_{3\times2} +  (1-\lambda) \left(\mu R_{n_{0\alpha}} (\ell_{n_{0\alpha}}^{1/2})_{3 \times 2} +  (1-\mu)R_{n_{0(\alpha-1)}} (\ell_{n_{0(\alpha-1)}}^{1/2})_{3 \times 2}\right)\right)  \\
&= \lambda r^{-1/6} \mathbb{P}_{e_3}(I_{3 \times2}) + (1-\lambda)\left(  \mu  \mathbb{P}_{e_3}(R_{n_{0\alpha}} (\ell_{n_{0\alpha}}^{1/2})_{3 \times 2}) +  \mathbb{P}_{e_3}(R_{n_{0(\alpha-1)}} (\ell_{n_{0(\alpha-1)}}^{1/2})_{3 \times 2})\right) \\
&= \lambda r^{-1/6} I_{2 \times 2} + (1-\lambda) (\mu r^{-1/6} I_{2\times2} + (1-\mu) r^{-1/6} I_{2\times2}) = r^{-1/6} I_{2\times2}.
\end{align*}
That is, this $\mathbb{R}^{2\times2}$ projection of any convex combination of these $\mathbb{R}^{3\times2}$ matrices is full-rank.  Therefore, any convex combination is also full-rank.  This result did not depend on $\alpha$ or $K$, so in particular, it shows that $\rank_l(\mathcal{S}_{int}^K) = 2$ for each $K$.  Thus, these interior junctions have a $\delta$-smoothing.
\end{proof}

Now, with regards to the exterior junctions, the case of only two interfaces is trivial.  In particular:
\begin{proposition}
For any $r \geq 1$ prior to self-intersection, each exterior junction for (b) the rhombic dodecahedron and (c) the rhombic triacontahedron has a $\delta$-smoothing.
\end{proposition}
\begin{proof}
By Proposition \ref{identDefs} and Theorem \ref{DSmoothTheorem}, we prove this result if we can find a $\tilde{F}_{b,c} \in \mathbb{R}^{3 \times 2}$ such that the sets 
\begin{align*}
&\mathcal{S}_{ext}^{(b)} := \co \big\{ \tilde{F}_b, \tilde{F}_2, \tilde{F}_5 \big\}  \cup \co \big\{ \tilde{F}_b, \tilde{F}_2, \tilde{F}_1\}\\
&\mathcal{S}_{ext}^{(c)} := \co \big\{ \tilde{F}_c, \tilde{F}_6, \tilde{F}_3\big\} \cup \co \big\{\tilde{F}_c, \tilde{F}_3, \tilde{F}_2 \big\}
\end{align*}
contains only matrices of full rank where the $\tilde{F}_{\alpha}$ are as described in Proposition \ref{identDefs} for (b) and (c).  Hence, we choose $\tilde{F}_b = \tilde{F}_2$ and $\tilde{F}_c = \tilde{F}_3$.   With these choices actually 
\begin{align*}
&\mathcal{S}_{ext}^{(b)} = \co \big\{ \tilde{F}_2,\tilde{F}_5\big\} \cup \co \big\{\tilde{F}_2, \tilde{F}_1 \big\}  \\
&\mathcal{S}_{ext}^{(c)} = \co \big\{ \tilde{F}_3, \tilde{F}_6 \big\} \cup  \co \big\{ \tilde{F}_3,\tilde{F}_2\big\}.
\end{align*}
Focusing on  (b), we note that $\rank_{l} (\co \{ \tilde{F}_2, \tilde{F}_5\}) = 2$.   This follows from the fact that $\tilde{F}_2$ and $\tilde{F}_5$ are, by Proposition \ref{identDefs}, rank-one compatible (which implies, in actuating the rhombic dodecahedron prior to self intersection, $\lambda \adj \tilde{F}_2 \times (1-\lambda)\adj \tilde{F}_5 \neq 0$ for all $\lambda \in [0,1]$).  The same argument applies to other convexified set in $\mathcal{S}_{ext}^{(b)}$.  Thus, $\rank_l (\mathcal{S}_{ext}^{(b)}) = 2$.  The argument is the same also for $\mathcal{S}_{ext}^{(c)}$.  This completes the proof. 
\end{proof}

The exterior junction for the box has three interfaces separating four regions of distinct deformation gradient.  Therefore, the previous proof technique is not applicable.  Instead, we resort to explicit computation of the deformations gradients.  Nevertheless:
\begin{proposition}
For any $r \geq 1$ prior to self-intersection at $r >3$, the exterior junction for (a) the box also has a $\delta$-smoothing.
\end{proposition} 
\begin{proof}
By explicit computation in the $\{\tilde{e}_1, \tilde{e}_2\}$ basis shown (with $e_3$ the outward normal) in Figure \ref{fig:SolveExamples}, 
\begin{align}\label{eq:explicitMat}
\tilde{F}_{1,3} = r^{-1/6}\left( \begin{array}{cc} 1 & 0 \\ 0 & 1 \\ \frac{1}{2} \sqrt{r-1} & \pm \frac{\sqrt{3}}{2} \sqrt{r-1} \end{array}\right), \quad \tilde{F}_{4,5} = r^{-1/6}\left( \begin{array}{cc} \frac{3-r}{2r} & \mp \frac{\sqrt{3}}{2}\left(\frac{r-1}{r}\right) \\ \mp 3 \frac{\sqrt{3}}{2}\left(\frac{r-1}{r}\right)  & \frac{3-r}{2r} \\ (3-r)\frac{\sqrt{r-1}}{r} & \pm  \sqrt{3}\frac{\sqrt{r-1}}{r} \end{array}\right).
\end{align}
Now, we claim that the set 
\begin{align}\label{eq:SetA}
\mathcal{S}_{ext}^{(a)}:= \co \{ \tilde{F}_a, \tilde{F}_1, \tilde{F}_4\} \cup \co \{ \tilde{F}_a, \tilde{F}_3, \tilde{F}_5 \} \cup \co \{\tilde{F}_a, \tilde{F}_1, \tilde{F}_3\} 
\end{align}
contains only matrices of full-rank if $\tilde{F}_a = r^{-1/6} I_{3\times2}$.  

To see this, first we note that we need only consider the first two sets since $\rank_l( \co \{r^{-1/6} I_{3\times2}, \tilde{F}_1, \tilde{F}_3\}) = 2$ from Proposition \ref{InterProp}.  In addition, we notice by explicit calculation that
\begin{align*}
&\tilde{F}^{\mp}(\lambda,\mu,r) := \lambda \left((1-\mu) r^{-1/6} I_{3\times2} + \mu \tilde{F}_{1,3}\right) + (1-\lambda) \tilde{F}_{4,5}  = \left( \begin{array}{cc}\xi(\lambda,r) &  \mp \kappa(\lambda, r) \\ \mp 3 \kappa(\lambda,r) & \xi(\lambda,r) \\ \gamma(\lambda,\mu,r) & \pm \phi(\lambda,\mu,r) \end{array}\right)
\end{align*}
where $\xi,\kappa,\gamma,\phi \geq 0$ for all $\lambda,\mu \in [0,1]$ and $r \in [1,3]$.  Thus, 
\begin{align*}
\adj \tilde{F}^{\mp}  = \left(\begin{array}{c} - (3 \kappa \phi + \xi \gamma) \\ \mp (\gamma \kappa + \xi \phi) \\ \xi^2 - 3 \kappa^2 \end{array}\right)
\end{align*}
where we have suppressed the dependence on $\lambda,\mu$ and $r$.  We will simply require the  non-negativity of each parameter as stated above for our argument. 

Suppose for the sake of a contradiction that $\adj \tilde{F}^{\mp} = 0$.   Using the non-negativity of the parameters, we have 
\begin{align}\label{eq:Implications}
\adj \tilde{F}^{\mp} = 0 \quad \Rightarrow  \quad \xi = \sqrt{3} \kappa \quad \Rightarrow \quad  \text{either:}\;\;  \begin{cases} 
 \gamma = -\sqrt{3} \phi \quad \Rightarrow \quad  \phi,\gamma = 0 \\ 
 \xi, \kappa = 0 \end{cases}
 \end{align}
Let us assume it is the case $\xi, \kappa = 0$, and notice that $\xi=0$ implies $r = 3$ and $\lambda = 0$.  However, in this case $ \tilde{F}^{\mp} = \tilde{F}_{4,5}$, and $\tilde{F}_{4,5}$ is full-rank.  Thus if $\adj \tilde{F}^{\mp}= 0$, it must be that $\phi,\gamma =  0$. However, we find additionally that 
\begin{align*}
\phi(\lambda, \mu, r) = \left(\lambda \mu \frac{1}{2} + (1-\lambda)\frac{1}{r} \right)\sqrt{3(r-1)}.
\end{align*}
Thus, we see that $\phi = 0$ for $\lambda,\mu \in [0,1]$ and $r \in [1,3]$ in only two cases: if $\lambda = 1$ and $\mu = 0$ or if $r = 1$.  For the first case though, $\tilde{F}^{\mp} = \tilde{F}_{1,3}$ which is full-rank.  For the second case,  $\tilde{F}^{\mp} = I_{3\times2}$ which is also full-rank.  So $\phi$ is only equal to zero on full-rank matrices.  This is the desired contradiction.  Indeed given this fact, $\adj \tilde{F}^{\pm}$ can never be zero due to (\ref{eq:Implications}).
\end{proof}

\section{Applications}\label{sec:Applications}

\subsection{Nonisometric origami: Compatibility and examples}\label{app:no}

\begin{figure}[t!]
   \begin{subfigure}{0.5\textwidth}
   \centering
    \includegraphics[width=0.47\linewidth ]{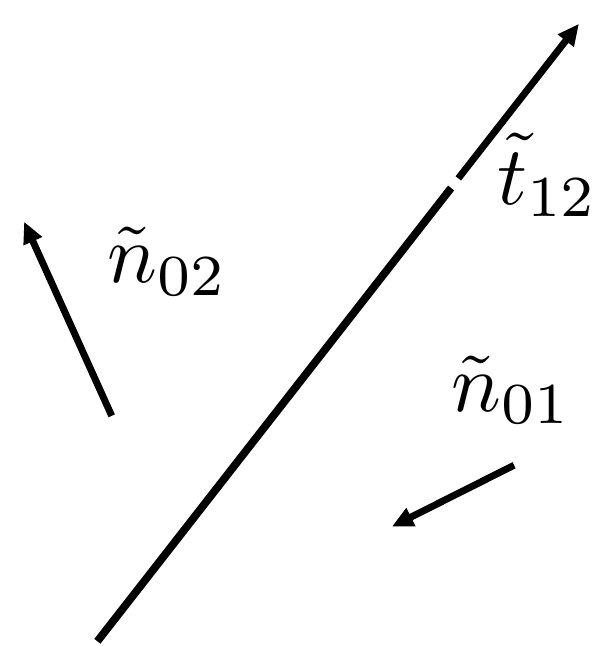} 
    \caption{A single interface}
     \label{fig:OneInterface}
     \end{subfigure}
     \begin{subfigure}{0.5\textwidth}
     \centering
\includegraphics[width=0.49\linewidth ]{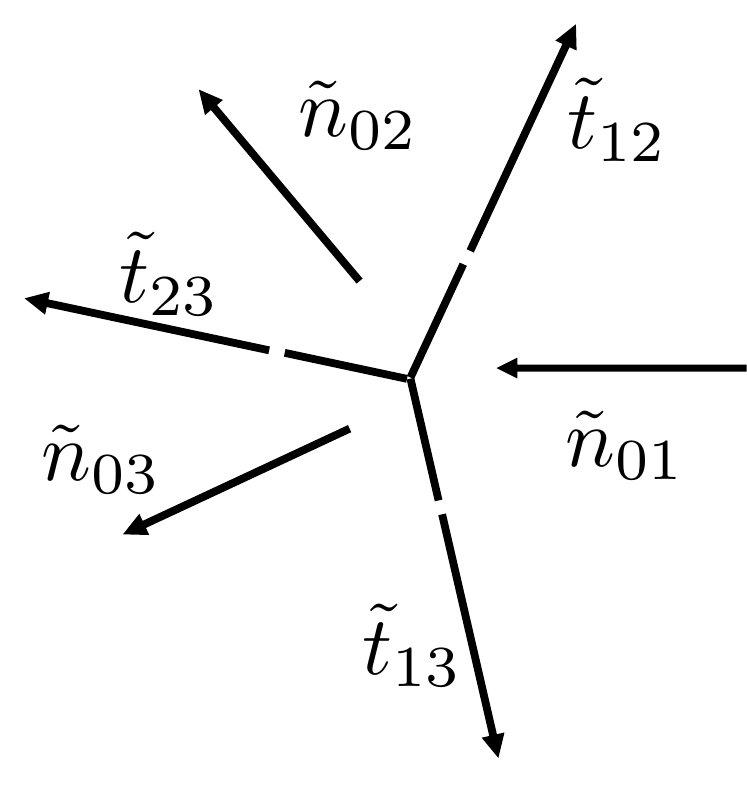}
\caption{Three-faced junction}
\label{fig:ThreeFacedJunc}
\end{subfigure}
\caption{Schematic of interfaces and junctions in nonisometric origami}
\label{fig:InterfaceFig}
\end{figure}

The actuation of complex shape stems from piecewise polygonal regions satisfying the nonisometric condition in Definition \ref{def:nonIsoDef}(ii), hence the term {\it nonisometric origami}.  In particular, the compatibility of interfaces separating regions of distinct constant director (Figure \ref{fig:OneInterface}) combined with the compatibility of junctions where these interfaces merge at a single point (Figure \ref{fig:ThreeFacedJunc}) play the key role in actuation.  To address this with mathematical precision, we note that the nonisometric condition in Definition \ref{def:nonIsoDef}(ii) is equivalent to 
\begin{align}\label{eq:equivMetric}
\tilde{F}_{\alpha} = R_{\alpha} (\ell_{n_{0\alpha}}^{1/2})_{3\times2} \quad \text{ for some $R_{\alpha} \in SO(3)$ }, \quad \alpha \in \{1,\ldots,N\}
\end{align} 
where $(\ell^{1/2}_{n_{0\alpha}})_{3\times2} = r^{-1/6}(I_{ 3 \times2} + (r^{1/2}-1) n_{0\alpha} \otimes \tilde{n}_{0\alpha})$ for the projection $\tilde{n}_{0\alpha} \in B_1(0) \subset \mathbb{R}^2$. (Clearly $\Leftarrow$ holds.  For $\Rightarrow$, since $(\tilde{F}_{\alpha})^T \tilde{F}_{\alpha} = \tilde{\ell}_{n_{0\alpha}}$ for each $\alpha$, there exists for each $\alpha$ a $b_{\alpha} \in \mathbb{R}^3$ as in Proposition \ref{bPropDef}. We deduce (\ref{eq:equivMetric}) using the polar decomposition theorem.) Thus for compatibility, the deformation $y$ in (\ref{eq:ybarAffine}) must be continuous across each interface separating regions of distinct constant director.  This occurs if and only if 
\begin{align}\label{eq:generalCompat}
R_{\alpha}  (\ell_{n_{0\alpha}}^{1/2})_{3\times2}\tilde{t}_{\alpha \beta}= R_{\beta} (\ell_{n_{0\beta}}^{1/2})_{3\times2} \tilde{t}_{\alpha \beta}
\end{align}
for every interface tangent $\tilde{t}_{\alpha \beta}\in \mathbb{S}^1$.  Explicitly, $\tilde{t}_{\alpha \beta}$ represents the tangent vector to the interface separating regions $\omega_{\alpha}$ with director $n_{0\alpha}$ and $\omega_{\beta}$ with director $n_{0\beta}$ as depicted in Figure \ref{fig:InterfaceFig}.  This condition is akin to the rank-one condition studied in the context of fine-scale twinning during the austenite martensite phase transition and actuation active martensitic sheets \cite{balljames_analcontmech_89, bhatta_03, bhattajames_jmps_99}.   More recently, this compatibility has been appreciated as a means of actuation for nematic elastomer and glass sheets \cite{mw_PRE_11,mw_phystoday_16} using planar programming of the director.  Here though, (\ref{eq:generalCompat}) describes the most general case of compatibility in thin nematic sheets as $n_{0\alpha}, n_{0\beta} \in \mathbb{S}^2$ need not be planar.  

While (\ref{eq:generalCompat}) encodes a complete characterization of nonisometric origami as defined in Definition \ref{def:nonIsoDef}, more useful criterion are gleamed from examining necessary and sufficient conditions associated with this constraint.  In particular, taking the norm of both sides of (\ref{eq:generalCompat}) yields, after some manipulation, a necessary condition for nonisometric origami, 
\begin{align}\label{eq:pieceNecessary}
|\tilde{n}_{0\alpha} \cdot \tilde{t}_{\alpha \beta}| = |\tilde{n}_{0\beta} \cdot \tilde{t}_{\alpha \beta}|
\end{align}
for every interface tangent $\tilde{t}_{\alpha \beta}$ (when $r \neq 1$).  We emphasize again that $\tilde{n}_{0\alpha} \in B_{1}(0) \subset \mathbb{R}^2$ is the projection of $n_{0\alpha}$ onto the tangent plane of $\omega$.  That this need not be a unit vector is a direct consequence of allowing for non-planar programming.  

A director program satisfying (\ref{eq:pieceNecessary}) is not, however, sufficient to ensure the existence of a deformation $y$ satisfying Definition \ref{def:nonIsoDef}(ii).  To illustrate this point, consider the design in Figure \ref{fig:NotCompatible}.  Here, we have a junction with three sectors of equal angle $2\pi/3$, and the director is programmed to bisect the sector angle (respectively, perpendicular to the bisector) on heating (respectively, cooling).  This program satisfies the necessary condition (\ref{eq:pieceNecessary}).  However in this case, due to the stretching part of the deformation upon actuation, the base of each triangle expands while the height contracts.  Thus, it is clear geometrically that no series of rotations and/or translations of the three deformed triangles can bring about a continuous piecewise affine deformation of the entire junction.  Conversely, if thermal actuation is reversed, as illustrated in Figure \ref{fig:Compatible} with the color change of the director program, then the base of each triangle contracts and the height expands.  In this case, a continuous piecewise affine deformation is realized by rotating each of the deformed triangles out-of-plane to form a 3-sided pyramid.  

\begin{figure}[t!]
\centering
\begin{subfigure}{.5\textwidth}
  \centering
  \includegraphics[width = 2.8 in]{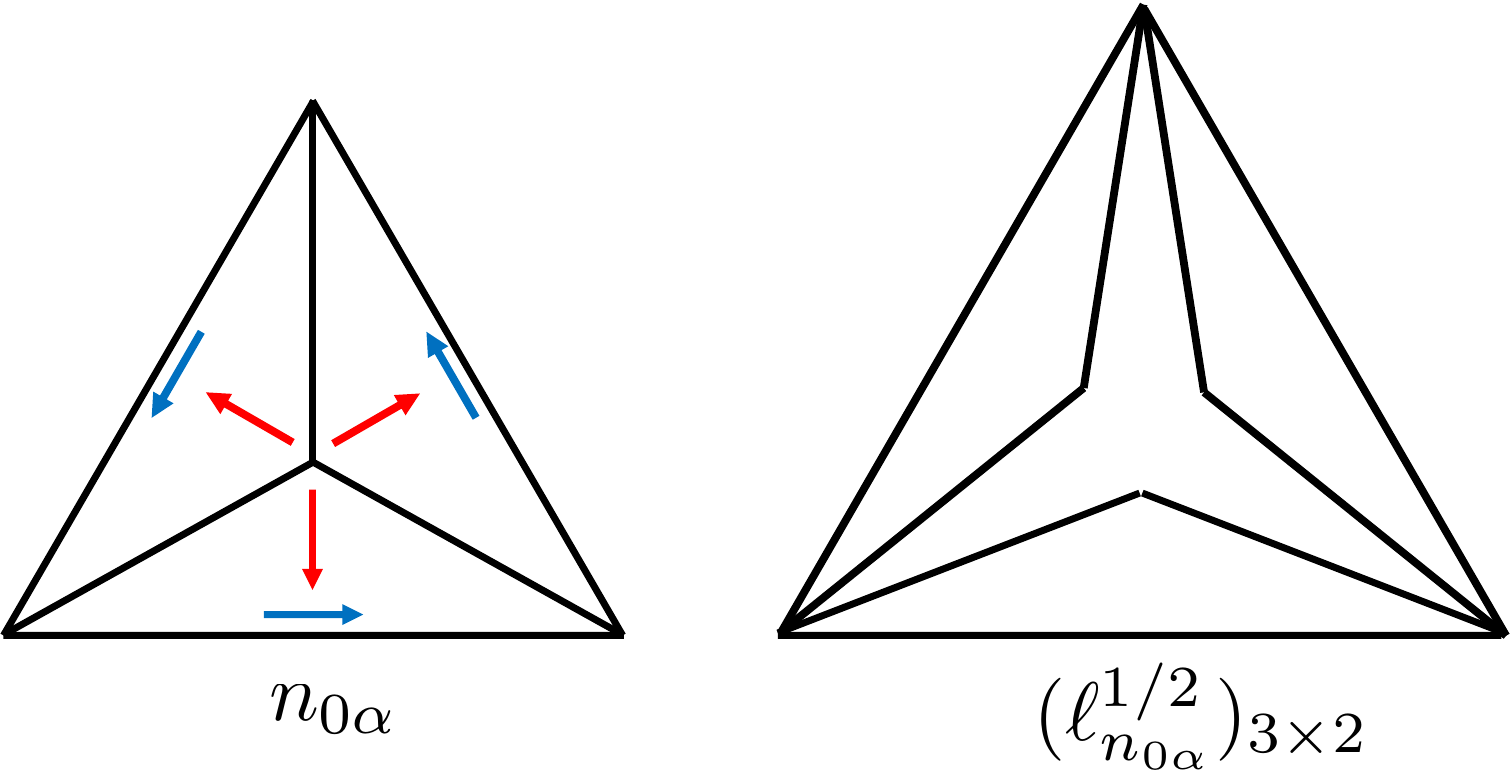}
  \caption{Incompatible Junction}
  \label{fig:NotCompatible}
\end{subfigure}%
\begin{subfigure}{.5\textwidth}
  \centering
  \includegraphics[width =2.8 in]{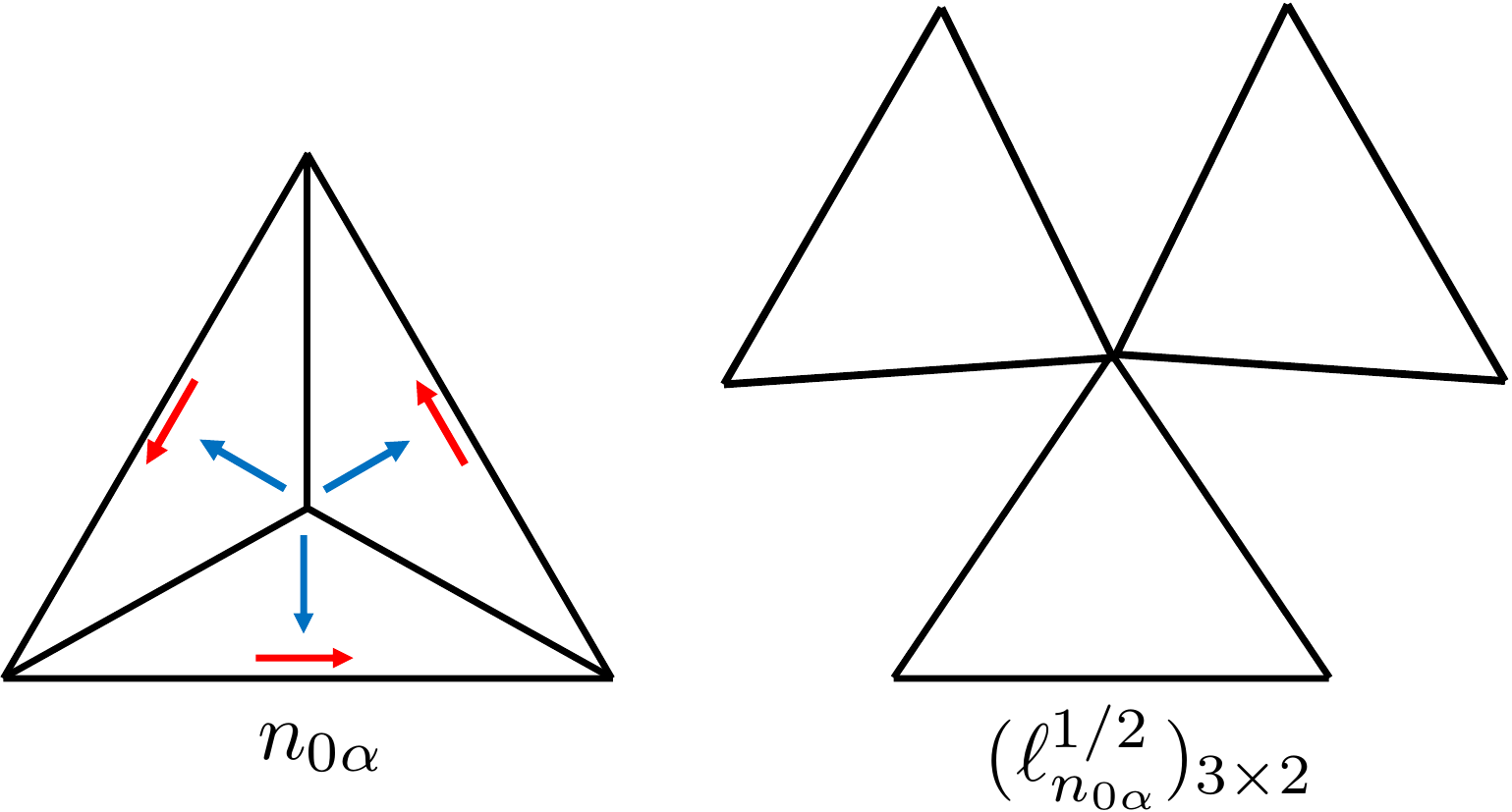}
  \caption{Compatible Junction}
  \label{fig:Compatible}
\end{subfigure}
\caption{Two junctions with director programs satisfying (\ref{eq:pieceNecessary}).  Blue represents the design for cooling and red represents the design for heating. The stretch part of the deformation upon thermal actuation is plotted.  }
\end{figure}

Figure \ref{fig:Compatible}, by way of example, also highlights a simple scheme to form a compatible pyramidal junction.  Indeed, if a junction has $K \geq 3$ sectors of equal angle $2\pi/K$  as in Figure \ref{fig:UnitCellEx}, then programming the director to bisect this angle upon cooling (respectively, perpendicular to the bisector on heating) alway leads to a compatible $K$-sided pyramid.  There are, of course, an infinite number of these types of junction, as emphasized with the designs in the right part of Figure \ref{fig:UnitCellEx}.  Most importantly though, these junctions can be used as unit cells to actuate more complex structures from nematic sheets.  This is shown in Figure \ref{fig:ExNonisometricOrigami} with designs for actuating a box (a), rhombic triacontehedron (b) and azimuthally periodic structures (c).   Each design incorporates a unit cell in Figure \ref{fig:UnitCellEx} as the building block.  

\begin{figure}[t!]
\centering
\includegraphics[width = 6 in]{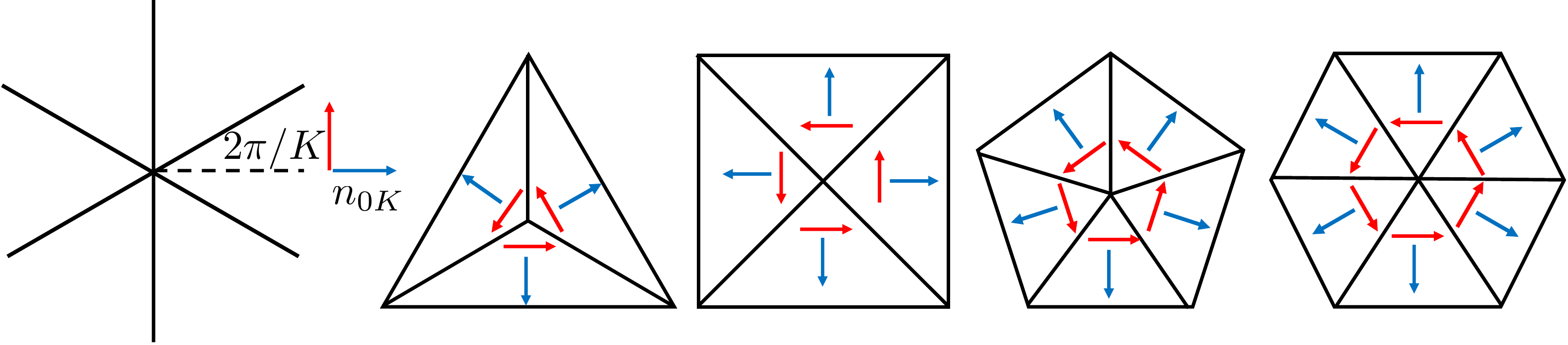}
\caption{Simple scheme for a compatible junction. Blue represents the design for cooling and red represents the design for heating.}  
\label{fig:UnitCellEx}
\end{figure}

\begin{figure}[t!]
  \begin{minipage}[b]{0.5\linewidth}
    \begin{subfigure}{0.95\textwidth}
    \centering
    \includegraphics[width=0.90\linewidth]{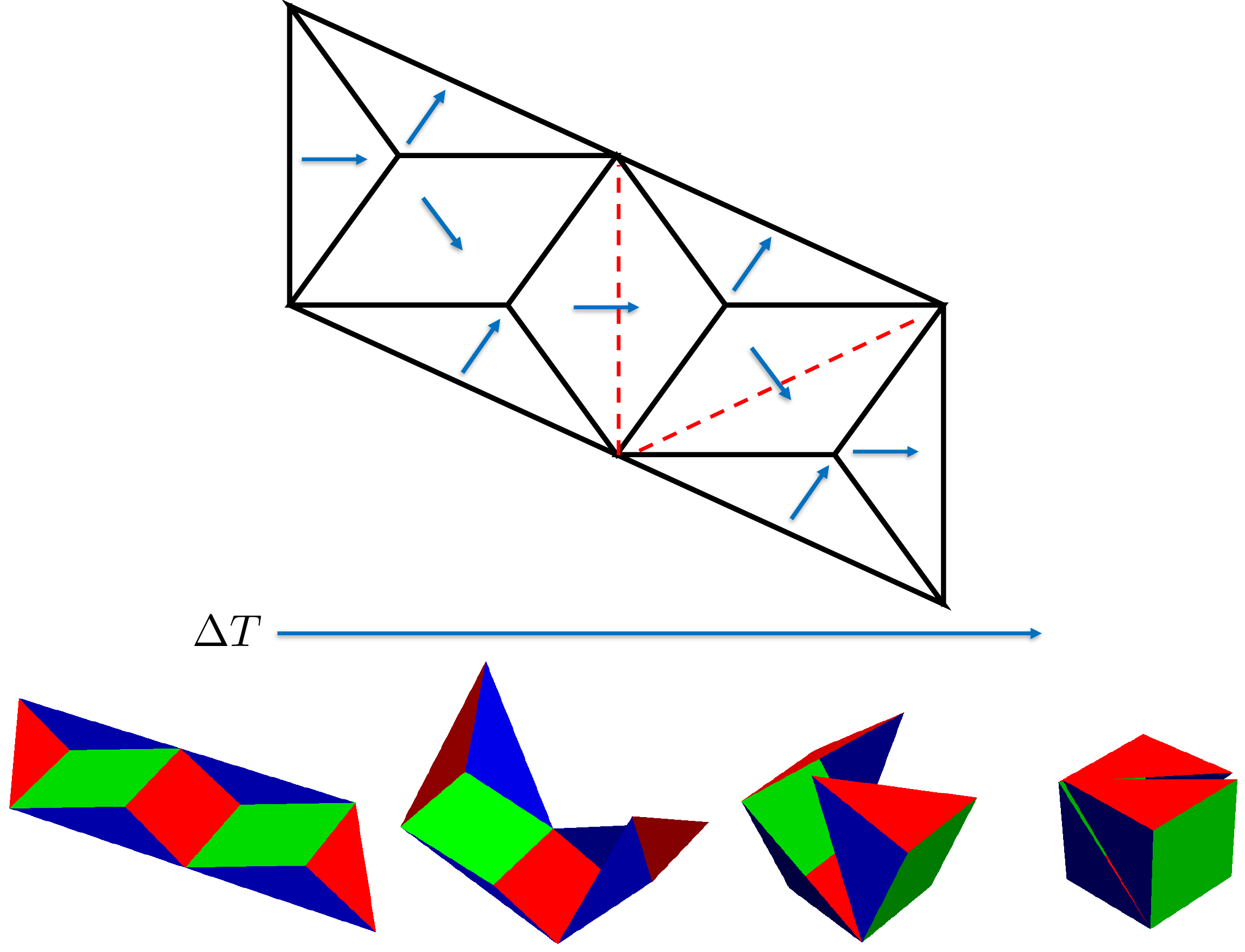} 
    \caption{Box construction} 
    \end{subfigure}
    \vspace{4ex}
  \end{minipage}
  \begin{minipage}[b]{0.5\linewidth}
    \begin{subfigure}{1\textwidth}
    \centering
    \includegraphics[width=1.0\linewidth]{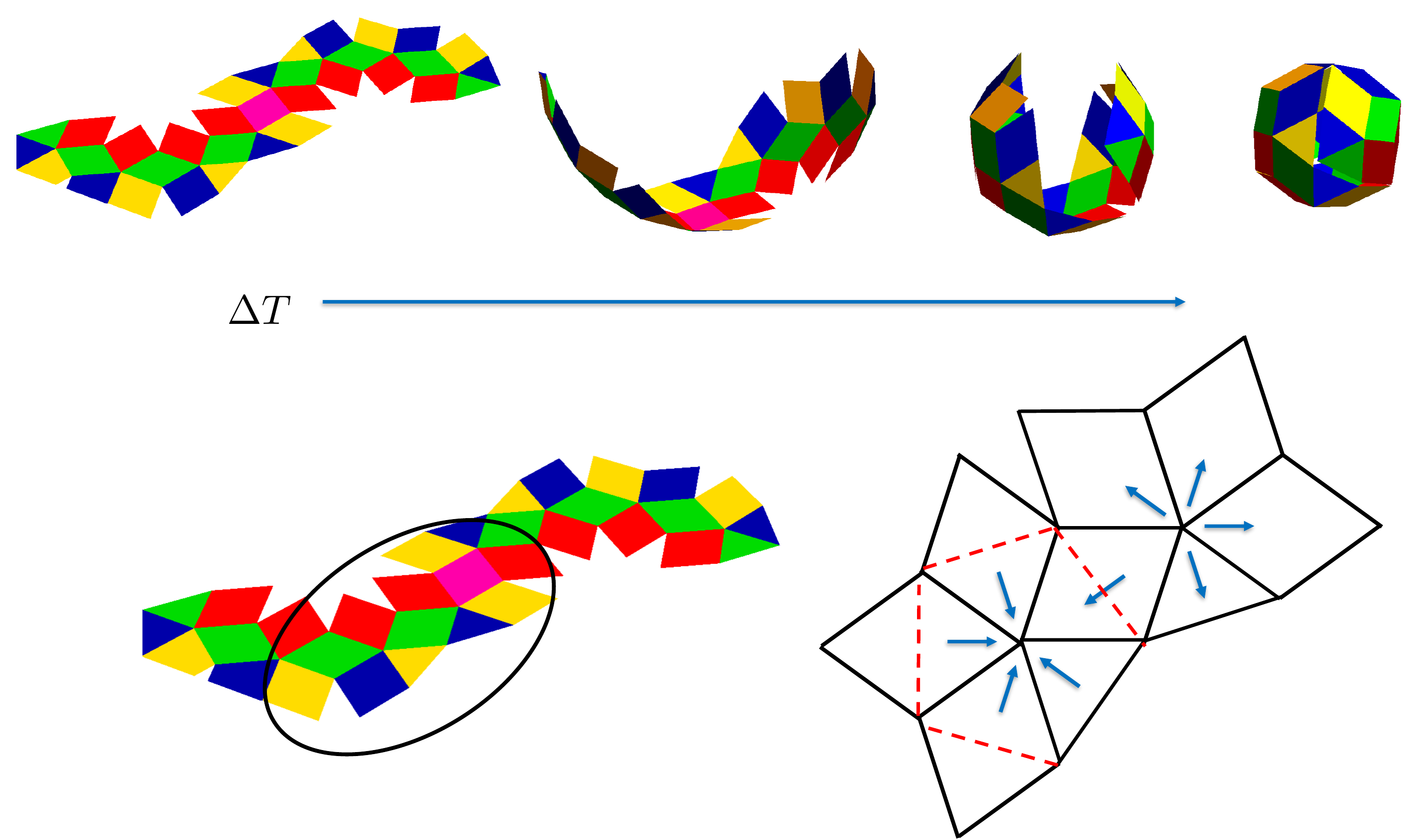} 
    \caption{Rhombic tricontahedron construction}
    \end{subfigure}
    \vspace{4ex}
  \end{minipage} 
  \begin{minipage}[b]{0.5\linewidth}
    \begin{subfigure}{0.92\textwidth}
    \centering
    \includegraphics[width=0.92\linewidth]{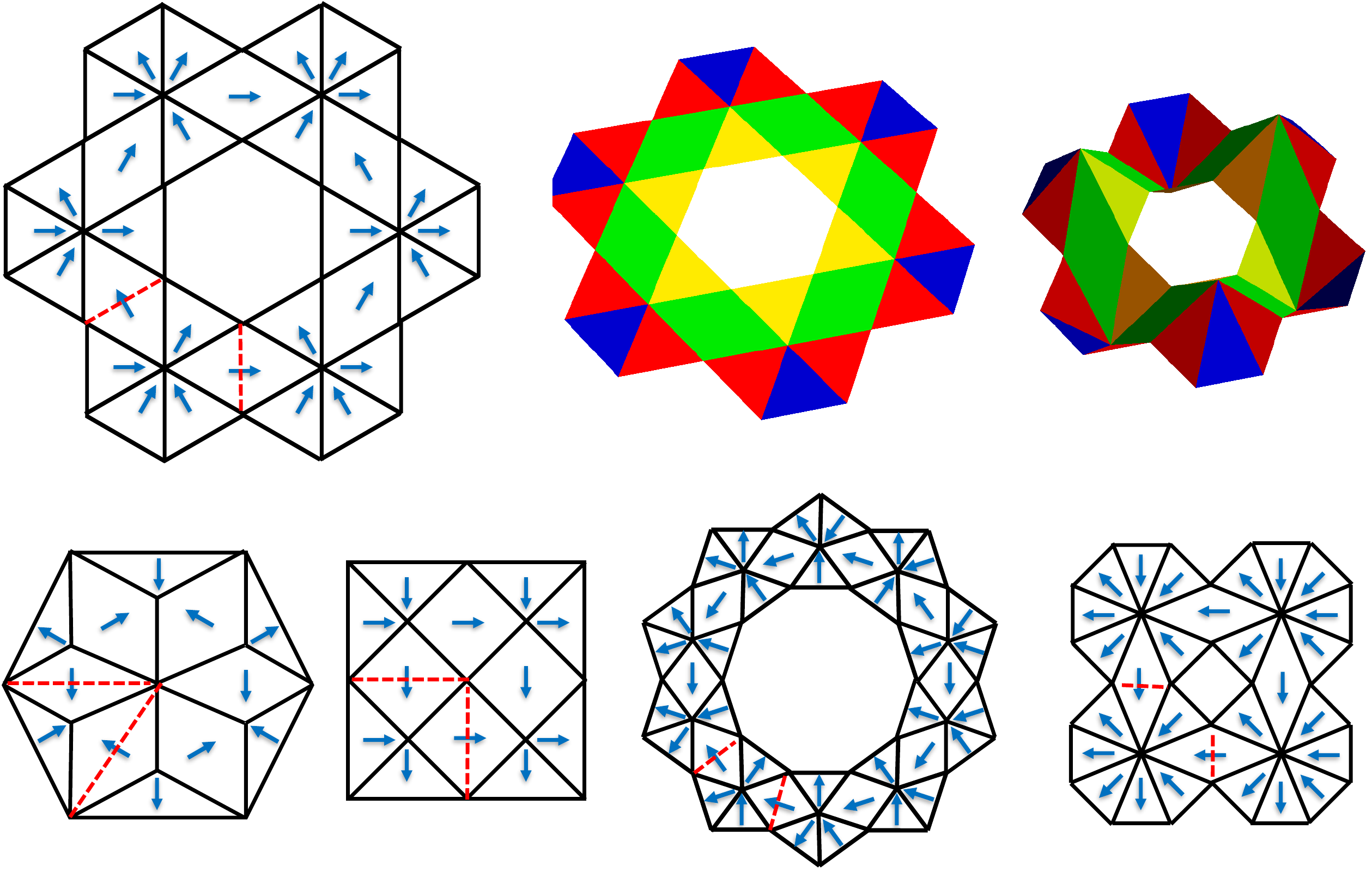} 
    \caption{Azimuthally Periodic Construction} 
    \label{fig:AzPeriodic}
    \end{subfigure}
    \vspace{4ex}
  \end{minipage}
  \begin{minipage}[b]{0.5\linewidth}
     \begin{subfigure}{1.0\textwidth}
    \centering
    \includegraphics[width=0.95\linewidth]{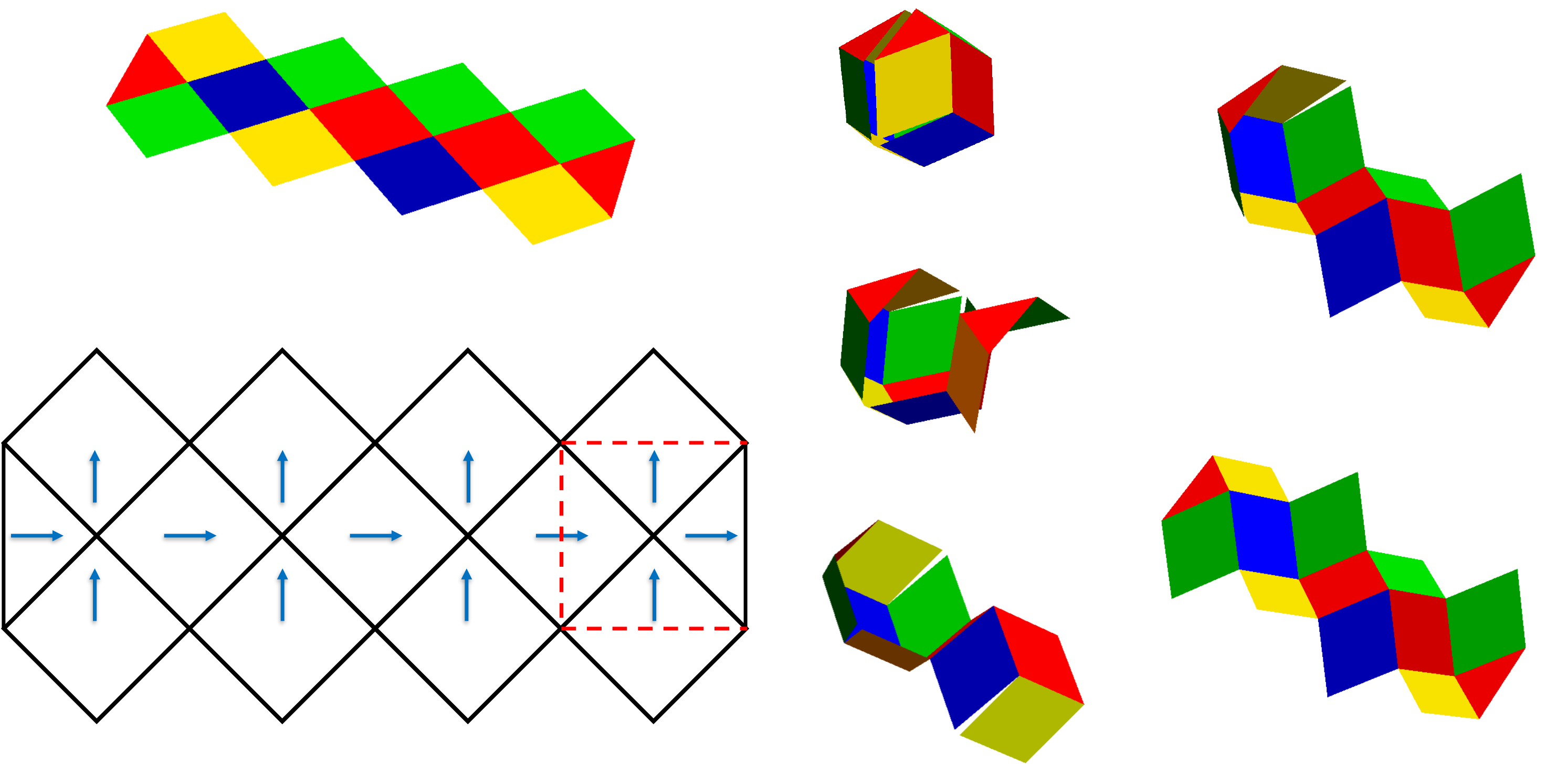} 
    \caption{Degeneracies}
    \label{fig:Degenerate}
    \end{subfigure}
    \vspace{4ex}
  \end{minipage} 
  \caption{Examples of nonisometric origami.  The design for cooling is shown.  For heating, each director field is replaced by its respective perpendicular in the plane. The unit cells which form the building blocks for these constructions are highlighted in red.}
  \label{fig:ExNonisometricOrigami} 
\end{figure}

The examples highlighted in Figure \ref{fig:ExNonisometricOrigami} illustrate that for even the simplest of building blocks, there is a richness of shape changing deformations of nematic elastomer sheets described by nonisometric origami. It should be noted, however, that these structure are in general degenerate.  This is shown in Figure \ref{fig:Degenerate} where we design a program to actuate a rhombic dodecahedron upon cooling. Here though, we have done nothing to break the reflection symmetry associated with the building block.  Thus, each interior junction is free to actuate either up or down.  Therefore, in addition to possibly actuating the rhombic dodecahedron, the actuation of four alternative surfaces is a completely equivalent outcome given this framework.  Such degeneracy was observed actuating conical defects by Ware et al.\ \cite{wetal_science_15}, where it was shown that each defect could actuate either up or down.  However, it may be possible to suppress these degeneracies by introducing a slight bias in the thru thickness director orientation via twisted nematic prescription.  This was seen, for instance, in Fuchi et al.\ \cite{fwbrvwj_sm_2015} (see also Gimenez-Pinto et al.\ \cite{getal_sr_17}), where actuation of a box like structure was achieved through folds biased in the appropriate direction using such prescription.  Thus, biasing would appear a promising means of breaking the reflection symmetry.  Nevertheless, we did not address this here as it is difficult to analyze to the level of rigor intended for this work.  

As a final comment on the design landscape for these constructions, recall that the relations associated with (\ref{eq:generalCompat}) provide a complete, but not particularly transparent, description  of nonisometric origami.  Further, the more useful condition (\ref{eq:pieceNecessary}) is only necessary as we provided a counterexample to sufficiency in Figure \ref{fig:NotCompatible}.  In fact, to our knowledge, a complete characterization of the geometry of configurations satisfying (\ref{eq:generalCompat}) remains open.  Nevertheless, we do expect an immense richness to such a characterization.  For instance,  in \cite{p_thesis_16} a more general, but by no means complete, characterization of compatible three-faced junctions is worked out, and numerous non-trivial examples of compatibility emerge from the analysis.  For these reasons, we feel a further pursuit in this direction appealing, though we did not delve deeper herein due to length considerations.  

\subsection{Lifted surfaces, and a recipe for design}

The idea for lifted surfaces (i.e., the ansatz (\ref{eq:ansatz}), (\ref{eq:graphVarphi}) and (\ref{eq:designRef})) is based on an equivalent rewriting of the metric constraint $(\tilde{\nabla} y)^T \tilde{\nabla} y = \tilde{\ell}_{n_0}$. (This equivalent form also yields a concrete design scheme for the actuation of nematic elastomers sheets in general.) Essentially, we take the picture of $y$ being a solution to $(\tilde{\nabla} y)^T \tilde{\nabla} y = \tilde{\ell}_{n_0}$ defined by a predetermined $n_0$ and turn it on its head. That is, we first identify the set of deformation gradients that are consistent with (\ref{eq:2DMetric}) for any director field and then we identify the director associated with that deformation gradient.

\be{theorem}
\label{prop:equivalent}
 Let $r>1$. The metric constraint (\ref{eq:2DMetric}) holds if and only if 
\begin{align}\label{eq:equivalence2DMet}
\tilde{\nabla} y(\tilde{x}) = (\partial_1y|\partial_2 y)(\tilde{x}) \in \mathcal{D}_{r}, \quad n_0(\tilde{x}) \in \mathcal{N}_{\tilde{\nabla} y(\tilde{x})}^{\; r} \quad \text{ a.e. } \tilde{x} \in \omega. 
\end{align}
Here,
\begin{equation}
\begin{aligned}\label{eq:Drg1}
\mathcal{D}_{r >1} := \Big\{ \tilde{F} \in \mathbb{R}^{3\times2} \colon &|\tilde{F}|^2 \leq r^{-1/3} + r^{2/3}, \;\; r^{-1/3} \leq |\tilde{F} \tilde{e}_{\alpha}|^2 \leq r^{2/3}, \alpha = 1,2, \\
& (\tilde{F} \tilde{e}_1 \cdot \tilde{F} \tilde{e}_2)^2 = (|\tilde{F}\tilde{e}_1|^2 - r^{-1/3})(|\tilde{F} \tilde{e}_2|^2 - r^{-1/3}) \Big \} 
\end{aligned}
\end{equation}
and 
\begin{equation}
\begin{aligned}\label{eq:Nrg1}
\mathcal{N}_{\tilde{F}}^{\;r>1} := \Big \{ \nu_0 \in \mathbb{S}^2 \colon &(\nu_0 \cdot e_{\alpha})^2 = \frac{|\tilde{F}\tilde{e}_{\alpha}|^2 -r^{-1/3}}{r^{2/3} - r^{-1/3}}\;\; \alpha = 1,2,  \\
& \text{sign}( (\nu_0 \cdot e_1) (\nu_0 \cdot e_2)) = \text{sign} (\tilde{F} \tilde{e}_1 \cdot \tilde{F} \tilde{e}_2) \Big\}
\end{aligned}
\end{equation}
for $\text{sign} \colon \mathbb{R} \rightarrow \{-1,0,1\}$ the sign function with $\text{sign}(0) = 0$.  (For $r < 1$, the inequalities in (\ref{eq:Drg1}) and the sign in (\ref{eq:Nrg1}) are reversed, i.e., for the latter: $\text{sign}(\tilde{F} \tilde{e}_1 \cdot \tilde{F} \tilde{e}_2) \mapsto -\text{sign}(\tilde{F} \tilde{e}_1 \cdot \tilde{F} \tilde{e}_2)$.)  

In addition, if $y \colon \omega \rightarrow \mathbb{R}^3$ such that $\tilde{\nabla} y(\tilde{x})  \in \mathcal{D}_{r}$ a.e., then there exists an $n \colon \omega \rightarrow \mathbb{S}^2$ such that $n(\tilde{x}) \in \mathcal{N}_{\tilde{\nabla} y(\tilde{x})}^{\;r}$ a.e.
\e{theorem}

\begin{figure}
\centering
\begin{subfigure}{.5\textwidth}
  \centering
  \includegraphics[height = 2.1 in]{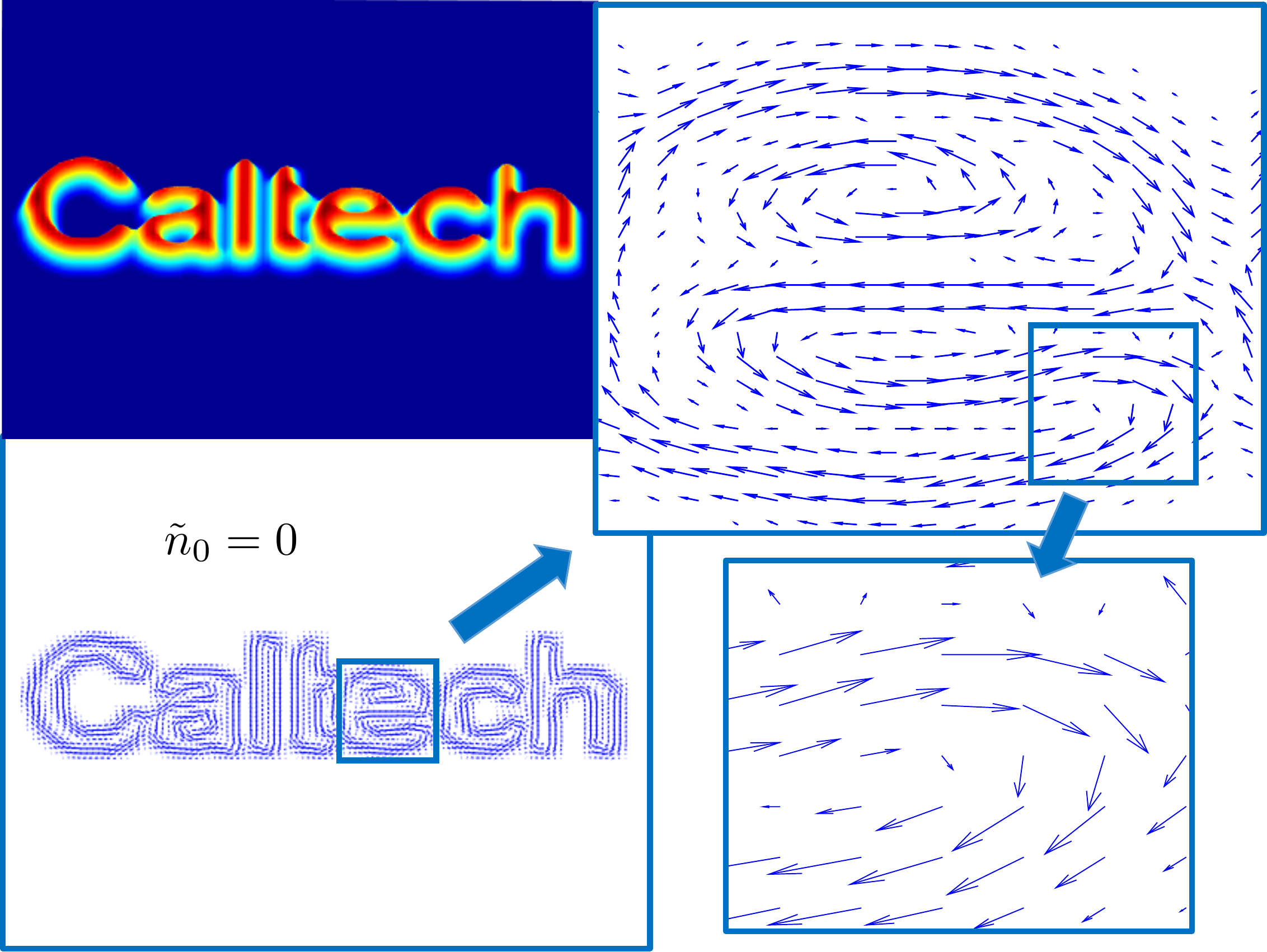}
  \caption{Caltech}
  \label{fig:Caltech}
\end{subfigure}%
\begin{subfigure}{.5\textwidth}
  \centering
  \includegraphics[height = 2.1 in]{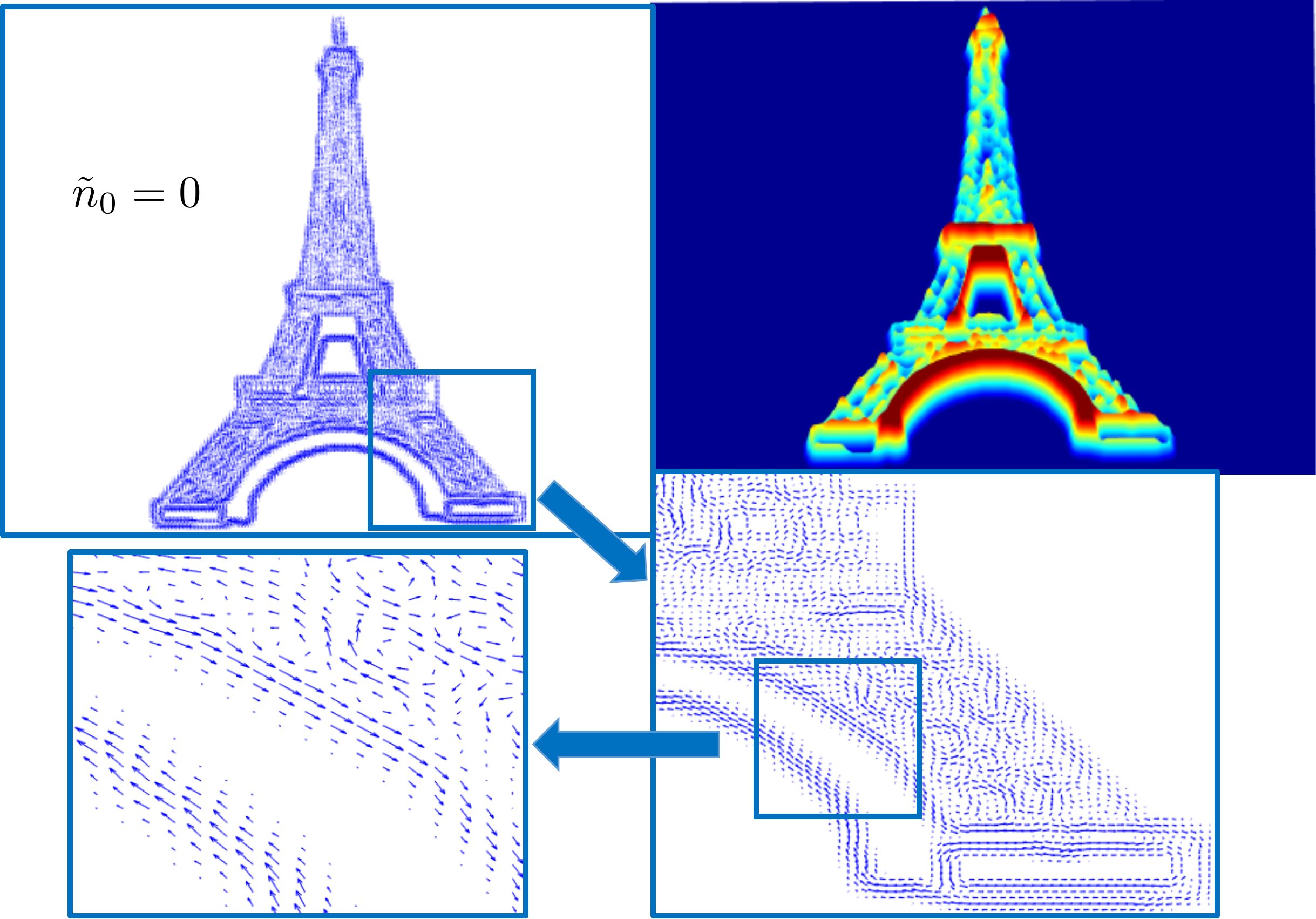}
  \caption{Eiffel Tower}
  \label{fig:Eiffel}
\end{subfigure}
\caption{The deformed shape and designs for lifted surfaces.  The  planar part of the director $\tilde{n}_0$ is plotted.}
\label{fig:ComplexShape}
\end{figure}

We prove this equivalence below.  Regarding the last point of the theorem (i.e., if $y \colon \omega \rightarrow \mathbb{R}^3$ such that  $\tilde{\nabla} y(\tilde{x}) \in \mathcal{D}_r \ldots$), we note that this means that for characterizing the geometry of surfaces which satisfy the metric constraint (\ref{eq:2DMetric}), we need only to consider the set of deformation gradients from a flat sheet $\omega$ which satisfy $\tilde{\nabla} y (\tilde{x}) \in \mathcal{D}_{r}$ a.e. $\tilde{x} \in \omega$.  Unfortunately, such a broad characterization remains open.  Of particular difficulty is the fact that this condition on the deformation gradient implies the equality
\begin{align}\label{eq:hardEqualDesign}
(\partial_1 y \cdot \partial_2 y)^2 = (|\partial_1 y|^2 -r^{-1/3})(|\partial_2 y|^2 -r^{-1/3}), \quad \text{ a.e. on } \omega. 
\end{align}
{\it Lifted surfaces} constitute a broad class of deformations such that this constraint holds trivially. 

Regarding applications, the lifted surfaces ansatz allows for actuation of a large variety of shapes, since the limitations imposed by (\ref{eq:graphVarphi}) are not very restrictive.  Since $r$ can be significantly different from 1 in nematic elastomers, one can form shapes with significant displacement like spherical caps and sinusoidally rough surfaces.   Figure \ref{fig:ComplexShape} shows two additional examples with complex surface relief.  These are but a small sample of the designs amenable to this framework.   Indeed, given any arbitrary greyscale image $\mathcal{G}$, we can program a nematic sheet so that the surface of the sheet upon cooling corresponds to this image.  We do this by smearing $\mathcal{G}$ (for instance by mollification or by averaging over a small square twice) and taking this as $\varphi$.

Nevertheless, the key ingredient to the design of lifted surfaces is the ability to program the director three dimensionally.  To our knowledge, experimental studies on nematic elastomer sheets such as Ware et al.\ \cite{wetal_science_15} have examined planar inscription of the director only. We hope that promising designs such as lifted surfaces will inspire future experimentation to realize three dimensional programming.   In any case, the theory and design scheme are easily adapted to the planar case.  Specifically in the case of a \textit{planar program}, the metric constraint (\ref{eq:2DMetric}) reduces to the metric underlying Aharoni et al.\ \cite{ask_prl_14}, and the spaces above reduce to 
\begin{align*}
&\mathcal{D}_{r>1} \equiv \tilde{\mathcal{D}}_{r>1}:=  \\
&\qquad  \Big\{ \tilde{F} \in \mathbb{R}^{3\times2} \colon |\tilde{F}|^2 = r^{2/3} + r^{-1/3}, \;\; r^{-1/3} \leq |\tilde{F} \tilde{e}_{1}|^2 \leq r^{2/3}, \;\; \det((\tilde{F})^T \tilde{F}) = r^{1/3}  \Big\}, \\
&\mathcal{N}^{\; r >1}_{\tilde{F}} \equiv \tilde{\mathcal{N}}^{\; r >1}_{\tilde{F}} :=  \\
&\qquad  \Big\{ \tilde{\nu}_0 \in \mathbb{S}^1 \colon (\tilde{\nu}_0 \cdot \tilde{e}_1)^2 = \frac{|\tilde{F} \tilde{e}_1|^2 - r^{1/3}}{r^{2/3} - r^{-1/3}}, \;\; \text{sign} ((\tilde{\nu}_0 \cdot \tilde{e}_1)(\tilde{\nu}_0 \cdot \tilde{e}_2)) = \text{sign}(\tilde{F} \tilde{e}_1 \cdot \tilde{F} \tilde{e}_2) \Big\},
\end{align*}
where again the inequalities above and the sign in $\tilde{\mathcal{N}}_{\tilde{F}}^{\;r}$ are reversed for $r <1$ (as in the theorem). 

\begin{proof}[Proof of Theorem \ref{prop:equivalent}.]
Let $\tilde{F} \in \mathbb{R}^{3\times2}$ and $\nu_0 \in \mathbb{S}^2$ satisfy $\tilde{F}^T \tilde{F} = \tilde{\ell}_m$.  Equivalently,
\begin{align}\label{eq:ExpandConst}
\left(\begin{array}{cc} |\tilde{F} \tilde{e}_1|^2 & (\tilde{F} \tilde{e}_1 \cdot \tilde{F} \tilde{e}_2) \\ (\tilde{F} \tilde{e}_1 \cdot \tilde{F} \tilde{e}_2) & |\tilde{F} \tilde{e}_2|^2 \end{array}\right) =  r^{-1/3} \left(\begin{array}{cc}  1 + (r - 1)(\nu_0 \cdot e_1)^2 & (r - 1)(\nu_0 \cdot e_1)(\nu_0 \cdot e_2) \\  (r -1)(\nu_0 \cdot e_1)(\nu_0 \cdot e_2) &  1 + (r - 1)(\nu_0 \cdot e_2)^2 \end{array}\right)
\end{align}
for $\{ \tilde{e}_1, \tilde{e}_2\}$ and $\{e_1, e_2, e_3\}$ the standard basis on $\mathbb{R}^2$ and $\mathbb{R}^3$ respectively.  Now, since $\nu_0 \in \mathbb{S}^2$, $(\nu_0 \cdot e_{\alpha})^2 \in [0,1]$ and  
\begin{align}\label{eq:1nem}
\begin{aligned}
|\tilde{F} \tilde{e}_{\alpha}|^2 &\in [r^{-1/3}, r^{2/3}] \;\; \text{ if } \;\; r > 1, \\
&\in [r^{2/3}, r^{-1/3}] \;\; \text{ if } \;\; r < 1,
\end{aligned}
\end{align}
from (\ref{eq:ExpandConst}) for $\alpha = 1,2$.  In addition, $(\nu_0 \cdot e_1)^2 + (\nu_0 \cdot e_2)^2 \leq 1$, and so 
\begin{align}\label{eq:2nem}
\begin{aligned}
|\tilde{F}|^2 = |\tilde{F} \tilde{e}_1|^2 + |\tilde{F} \tilde{e}_2|^2 &\leq r^{2/3} + r^{-1/3} \;\; \text{ if } \;\; r > 1, \\
&\geq r^{2/3} + r^{-1/3} \;\; \text{ if } \;\; r < 1,
\end{aligned}
\end{align}
also from (\ref{eq:ExpandConst}).  Now note that substituting the diagonal terms into the square of the off diagonal term in (\ref{eq:ExpandConst}) results in 
\begin{align}\label{eq:3nem}
(\tilde{F} \tilde{e}_1 \cdot \tilde{F} \tilde{e}_2)^2 =(|\tilde{F} \tilde{e}_1|^2 - r^{-1/3})(|\tilde{F} \tilde{e}_2|^2 - r^{-1/3}).
\end{align}
Combining (\ref{eq:1nem}), (\ref{eq:2nem}) and (\ref{eq:3nem}), we conclude $\tilde{F} \in \mathcal{D}_{r}$ as desired.  To prove $\nu_0 \in \mathcal{N}_{\tilde{F}}^{\;r}$, note that since $r \neq 1$, rearranging the diagonal terms in (\ref{eq:ExpandConst}) gives 
\begin{align}\label{eq:mFormula}
(\nu_0 \cdot e_{\alpha})^2 = \frac{ |\tilde{F} \tilde{e}_\alpha|^2 - r^{-1/3}}{r^{2/3} - r^{-1/3}}, \;\; \alpha = 1,2.
\end{align}
Further, since $r > 0$ and $\neq 1$, taking the sign of the off diagonal term in (\ref{eq:ExpandConst}) gives
\begin{align}\label{eq:mSign}
\text{sign}( (\nu_0 \cdot e_1) (\nu_0 \cdot e_2)) = \text{sign}(r -1) \text{sign}( \tilde{F} \tilde{e}_1 \cdot \tilde{F} \tilde{e}_2). 
\end{align}
Since $\nu_0 \in \mathbb{S}^2$, combining (\ref{eq:mFormula}) and (\ref{eq:mSign}) yields $\nu_0 \in \mathcal{N}_{\tilde{F}}^{\; r}$.  

Now, let $\tilde{F} \in \mathcal{D}_{r}$ and $\nu_0 \in \mathcal{N}_{\tilde{F}}^{\;r}$.  To prove $\tilde{F}^T \tilde{F} = \tilde{\ell}_{\nu_0}$, we need to show (\ref{eq:ExpandConst}).  By hypothesis, we have (\ref{eq:mFormula}).  By rearranging this formula, we obtain the diagonal terms in (\ref{eq:ExpandConst}).  For the off diagonal term, we note that in addition to (\ref{eq:mFormula}), we have (\ref{eq:3nem}) by hypothesis.  Combining these relations, we find
\begin{align*}
(\tilde{F} \tilde{e}_1 \cdot \tilde{F} \tilde{e}_2)^2 = (r^{2/3} - r^{-1/3})^2(\nu_0 \cdot e_1)^2(\nu_0 \cdot e_2)^2.  
\end{align*}
Taking the square root, we have the off diagonal term up to the sign.  The correct choice of sign is guaranteed since $\nu_0$ and $\tilde{F}$ satisfy (\ref{eq:mSign}), again by hypothesis.  

Finally, we let $\tilde{F} \in \mathcal{D}_{r}$, and show $\mathcal{N}_{\tilde{F}}^{\;r}$ is non-empty.  Indeed by definition, $\tilde{F}$ satisfies (\ref{eq:1nem}) and (\ref{eq:2nem}).  Thus, the right side of (\ref{eq:mFormula}) is non-negative.  From this, we may find an $\nu_0 \in \mathbb{R}^3$ satisfying (\ref{eq:mFormula}) and (\ref{eq:mSign}).  Further by (\ref{eq:2nem}), $(\nu_0 \cdot e_1)^2 + (\nu_0 \cdot e_2)^2 \leq 1$.  Thus, we can choose $(\nu_0 \cdot e_3)$ such that $\nu_0 \in \mathbb{S}^2$.  It follows that $\mathcal{N}^{\;r}_{\tilde{F}}$ is non-empty. 

For functions, $\tilde{\nabla} y(\tilde{x}) \equiv \tilde{F}$ and $n(\tilde{x}) \equiv \nu_0$, and so all these results should hold pointwise a.e.
\end{proof}

\appendix
\addcontentsline{toc}{section}{Appendices}
\section*{Appendices}
\section{Some facts about the entropic energy}

\subsection{Some estimates on the energy densities}
\begin{proposition}\label{UpLowProp}
If $F \in \mathbb{R}^{3\times3}$ such that $\det F = 1$ and $\nu_0, \nu \in \mathbb{S}^2$, then 
\begin{align}\label{eq:uplower boundWDe}
\frac{1}{c} (|F|^2 - 1) \leq W^e(F, \nu,\nu_0) \leq c(|F|^2 +1)
\end{align}
for some $c = c(r_f,r_0) > 0$.
\end{proposition}
\begin{proof}
Since $\det F =1$ and $\nu_0 \in \mathbb{S}^2$, we have $\nu = F\nu_0/|F \nu_0| \in \mathbb{S}^2$ and 
\begin{align*}
W^e(F,\nu,\nu_0) &= W_{nH}((\ell_{\nu}^f)^{-1/2} F (\ell_{\nu_0}^0)^{1/2}) \\
&\geq \frac{\mu}{2} \left( \frac{1}{3}\sigma_{min}^2((\ell_{\nu_0}^0)^{1/2})|(\ell_{\nu}^f)^{-1/2} F|^2 - 3\right) \\
& \geq \frac{\mu}{2} \left( \frac{1}{9} \sigma_{min}^2((\ell_{\nu_0}^0)^{1/2}) \sigma_{min}^2((\ell_n^f)^{-1/2})|F|^2 - 3 \right).
\end{align*}
We note that $\sigma_{min}((\ell_{\nu_0}^0)^{1/2})$ is nonzero and depends only on $r_0$.  Similarly, $\sigma_{min}((\ell_{\nu}^f)^{-1/2})$ is nonzero and depends only on $r_f$.  Thus, the lower bound in (\ref{eq:uplower boundWDe}) follows.  The upper bound is similar. 
\end{proof}

\begin{proposition}\label{DumbProp}
Let $G \in \mathbb{R}^{3\times3}$ such that $\det(I_{3\times3} + G) = 1$.   We find
\begin{align}
&W_{nH}(I_{3\times3} + G) \leq C( |G|^2 + |G|^3) \label{eq:Dumb2},
\end{align}
for $W_{nH}$ in (\ref{eq:WnH}) and for some uniform constant $C> 0$.  
\end{proposition}
\begin{proof}
For the inequality on $W_{nH}$, we note that since $\det(I_{3\times3} + G) = 1$, 
\begin{align*}
\Tr(G) = - \Tr(\cof G) - \det(G)
\end{align*}
and $W_{nH}$ is finite.  Hence,
\begin{align}\label{eq:WnHIdent}
W_{nH}(I + G) &= \frac{\mu}{2} \left(|I + G|^2 - 3\right)  \nonumber \\
&= \frac{\mu}{2}\left( |G|^2 + 2 \Tr(G) \right) \nonumber  \\
&=\frac{\mu}{2}\left( |G|^2 - 2\Tr(\cof G) - 2 \det(G)\right).  
\end{align}
Since there exists a $C> 0$ independent of $G$ such that $|\Tr(\cof G)|\leq C|G|^2$ and $|\det(G)|\leq C|G|^3$, we conclude (\ref{eq:Dumb2}) following the identity (\ref{eq:WnHIdent}).  
\end{proof}

\subsection{Relating the metric and the step-length tensor}
\begin{proposition}\label{LBProp}
The energy density $W_{nH}$ in (\ref{eq:WnH}) satisfies $W_{nH}(F) \geq \frac{\mu}{2} \dist^2(F,SO(3))$ for all $F \in \mathbb{R}^{3\times3}$.  
\end{proposition}
\begin{proof}
We may assume $\det F = 1$ as the bound holds trivially otherwise.  Consequently and by the polar decomposition theorem, $F = R U$ for $R \in SO(3)$ and $U$ positive definite.  Hence, we find $\dist(F,SO(3)) = |U- I_{3\times3}|$.  In addition, since $\det U = 1$ we conclude
\begin{align*}
\frac{\mu}{2}\dist^2(F,SO(3)) &= \frac{\mu}{2}|U- I_{3\times3}|^2 \\
&= \frac{\mu}{2}(|U|^2 - 2\Tr(U) + 3) \\
&\leq \frac{\mu}{2}(|U|^2 - \inf\{\Tr(G)\colon G \text{ pos. def.}, \det G =1\} + 3) \\
&= \frac{\mu}{2}(|U|^2 - 3) = W_{nH}(F).
\end{align*}
Here, we used that the infimum above is attained at $G = I$.  
\end{proof}
\begin{proposition}\label{EquivalenceProp}
The energy density $W^{e}$ in (\ref{eq:We}) satisfies $W^{e}(F,F\nu_0/|F\nu_0|,\nu_0) = 0$ if and only if $\det F = 1$ and $F^T F = \ell_{\nu_0}$ for $\ell_{\nu_0}$ defined in (\ref{eq:3DMetric}).  In addition, if these identities hold, then $F \nu_0/ |F \nu_0| = R \nu_0$ where $R \in SO(3)$ is the unique rotation associated with the polar decomposition of $F$.  
\end{proposition}
\begin{proof}
We first assume $W^{e}(F,F\nu_0/|F\nu_0|,\nu_0) = 0$.  Then $\det F = 1$, $\nu_0\in \mathbb{S}^2$ and $|F\nu_0|\neq 0$.  We set $\nu := F\nu_0/|F\nu_0| \in \mathbb{S}^2$ and observe from  (\ref{eq:WeWnH}),
\begin{align*}
0=W^e(F, F\nu_0/|F\nu_0|,\nu_0) = W_{nH}((\ell_{\nu}^f)^{-1/2} F (\ell_{\nu_0}^{0})^{1/2}).
\end{align*}
Thus, we deduce from Proposition \ref{LBProp} that $(\ell_{\nu}^f)^{-1/2} F (\ell_{\nu_0}^{0})^{1/2} = R$ for some $R \in SO(3)$.  Evidently then, 
\begin{align}\label{eq:StepLengthR}
F = (\ell_\nu^f)^{1/2} R (\ell_{\nu_0}^0)^{-1/2}.  
\end{align}
Further,
\begin{align*}
r_f^{1/6} \nu = (\ell_\nu^f)^{-1/2} \nu &= (\ell_\nu^f)^{-1/2} \left( \frac{F \nu_0}{|F\nu_0|}\right) \\
& = (\ell_\nu^f)^{-1/2} \left(\frac{(\ell_\nu^f)^{1/2} R (\ell_{\nu_0}^0)^{-1/2} \nu_0}{|F\nu_0|}\right) = r_0^{-1/6} \left(\frac{R \nu_0}{|F\nu_0|}\right).
\end{align*}
Here, we used the definition of $\nu$, the result in (\ref{eq:StepLengthR}) and properties of the step-length tensors (\ref{eq:StepLength}).  Since both $\nu$ and $R \nu_0 \in \mathbb{S}^2$, it follows from this equality chain that actually $\nu = R\nu_0$.  Substituting this into (\ref{eq:StepLengthR}) yields
\begin{align*}
F = R (\ell_{\nu_0}^f)^{1/2} R^T R (\ell_{\nu_0}^0)^{-1/2} = R (\ell_{\nu_0}^f)^{1/2} (\ell_{\nu_0}^0)^{-1/2} =R \ell_{\nu_0}^{1/2}
\end{align*}
noting that $(\ell^f_{R\nu_0})^{1/2} = R (\ell_{\nu_0}^f)^{1/2} R^T$ and $(\ell_{\nu_0}^f)^{1/2} (\ell_{\nu_0}^0)^{-1/2}=  \ell_{\nu_0}^{1/2}$.  Consequently, $F^T F = \ell_{\nu_0}$ as desired.  

For the other direction, we assume $\det F = 1$ and $F^T F = \ell_{\nu_0}$.  This implies $F = R \ell_{\nu_0}^{1/2}$ for some $R \in SO(3)$ and $\nu := F \nu_0/|F\nu_0| \in \mathbb{S}^2$.  Thus,
\begin{align*}
\nu = \frac{F \nu_0}{|F \nu_0|} = \frac{R \ell_{\nu_0}^{1/2} \nu_0}{|F\nu_0|} = r^{-1/6} \frac{R \nu_0}{|F \nu_0|},
\end{align*}
and since both $\nu$ and $R \nu_0 \in \mathbb{S}^2$, we deduce $\nu = R \nu_0$.  Then by definition (\ref{eq:We}), $W^e(F,F\nu_0/|F\nu_0|,\nu_0) = W^e(F,R\nu_0,\nu_0)$ and clearly this is finite.  Further given (\ref{eq:WeWnH}), we find
\begin{align*}
W^e(F,R\nu_0,\nu_0) &= W_{nH}((\ell_{R\nu_0}^f)^{-1/2} F (\ell_{\nu_0}^{0})^{1/2}) \\
&=W_{nH}(R (\ell_{\nu_0}^f)^{-1/2} R^T R \ell_{\nu_0}^{1/2} (\ell_{\nu_0}^0)^{1/2}) \\
&=W_{nH}(R(\ell_{\nu_0}^f)^{-1/2} (\ell_{\nu_0}^f)^{1/2} (\ell_{\nu_0}^0)^{-1/2} (\ell_{\nu_0}^0)^{1/2}) = W_{nH}(R)
\end{align*}
For this, we have exploited properties of the step-length tensors (see previous paragraph).  Hence since $R \in SO(3)$, by Proposition \ref{LBProp} we find $ W_{nH}(R) = 0$ as desired.  

Finally for the implication, we note that in the proof of both directions, we found $F = R \ell_{\nu_0}^{1/2}$ and $\nu = R\nu_0$ for $R \in SO(3)$.  Consequently, since $\ell_{\nu_0}^{1/2}$ is positive definite, $R$ is actually the unique rotation in the polar decomposition of $F$.  
\end{proof}

\begin{proposition}\label{WnonIdealProp}
Set $\widehat{W}(F,\nu,\nu_0) := (\mu/2)^{-1}(W^e(F,\nu,\nu_0) + W^{ni}(F,\nu,\nu_0))$ for $W^e$ in (\ref{eq:We}) and $W^{ni}$ in (\ref{eq:Wni}).  $\widehat{W}$ is minimized (and equal to zero) if and only if 
\begin{align}\label{eq:nonIdealIdents}
\det F = 1, \quad F^T F = \ell_{\nu_0} \quad \text{ and } \quad \nu = \sigma  \frac{F\nu_0}{|F\nu_0|} \;\; \text{ for } \sigma \in \{-1,1\}.
\end{align}
\end{proposition}
\begin{proof}
$(\Rightarrow.)$  Given $\widehat{W} = 0$, $W^e=0$ and $W^{ni} =0$ since both are non-negative.  The former equality implies $(\ell_{\nu}^f)^{-1/2} F (\ell_{\nu_0}^0)^{1/2} = R \in SO(3)$ give Proposition \ref{LBProp}.  Hence, we observe that 
\begin{align*}
W^{ni}(F,\nu,\nu_0) &= \frac{\mu \alpha}{2} r_f^{2/3}|(I_{3\times3}-\nu_0 \otimes  \nu_0) F^T (\ell_{\nu}^f)^{-1/2} \nu|^2 \\
&= \frac{\mu \alpha}{2} (r_0^{1/3} r_f^{2/3}) |(I_{3 \times 3}- \nu_0 \otimes \nu_0) (\ell_{\nu_0}^0)^{1/2} F^T (\ell_{\nu}^f)^{-1/2} \nu|^2 \\
&=\frac{\mu \alpha}{2} (r_0^{1/3} r_f^{2/3}) |(I_{3 \times 3}- \nu_0 \otimes \nu_0) R^T \nu|^2 ,
\end{align*} 
and this must vanish.  Consequently, $\nu = \sigma R\nu_0 = \sigma F\nu_0/|F\nu_0|$ for some $\sigma \in \{-1,1\}$ (the latter equality follows from $R=(\ell_{\nu}^f)^{-1/2} F (\ell_{\nu_0}^0)^{1/2}$).  Thus by Proposition \ref{EquivalenceProp}, $\det  F = 1$ and $F^T F = \ell_{\nu_0}$.  

$(\Leftarrow.)$  Given (\ref{eq:nonIdealIdents}), $W^e = 0$, $\nu = \sigma R\nu_0$ and $F = R(\ell_{\nu_0})^{1/2}$ by Proposition \ref{EquivalenceProp}.  Thus with $(I_{3 \times 3} - \nu_0 \otimes \nu_0) \ell_{\nu_0}^{1/2} = r^{1/6} (I_{3\times 3} - \nu_0 \otimes \nu_0)$, it is easy to see that $W^{ni}$ also vanishes.  This completes the proof.
\end{proof}

\begin{proposition}\label{bPropDef}
If $\tilde{F} \in \mathbb{R}^{3 \times2}$ and $\nu_0 \in \mathbb{S}^2$ such that $\tilde{F}^T \tilde{F} = \tilde{\ell}_{\nu_0}$, then there exists a $b \in \mathbb{R}^3$ such that 
\begin{align*}
(\tilde{F}|b)^T (\tilde{F}|b) = \ell_{\nu_0}, \quad \det(\tilde{F}|b) = 1.
\end{align*} 
In particular, 
\begin{equation}
\begin{aligned}\label{eq:rewritb}
&b = \bar{b}_1 \tilde{F} \tilde{e}_1 + \bar{b}_2 \tilde{F} \tilde{e}_2 + \bar{b}_3(\tilde{F} \tilde{e}_1 \times \tilde{F} \tilde{e}_2), \\
&\left(\begin{array}{c} \bar{b}_1 \\ \bar{b}_2 \end{array}\right) = (\tilde{\ell}_{\nu_0})^{-1} I_{2\times3} \ell_{\nu_0} e_3, \quad \bar{b}_3 =\frac{1}{|\tilde{F} \tilde{e}_1 \times \tilde{F} \tilde{e_2}|^2},  \\
&(\tilde{\ell}_{\nu_0})^{-1} = r^{1/3} \left(I_{2\times2} +\left(\frac{1- r}{1+|\tilde{\nu}_0|^2(r-1)} \right) \tilde{\nu}_0 \otimes \tilde{\nu}_0 \right). 
\end{aligned}
\end{equation}
for $\tilde{\nu}_0 = (\nu_0 \cdot e_1, \nu_0 \cdot e_2) \in B_1(0) \subset \mathbb{R}^2$.  
\end{proposition}
\begin{proof} 
We remark that $\det(\tilde{\ell}_{\nu_0})  = r^{-2/3} ( 1 + (r - 1) |\tilde{\nu}_0|^2) > 0$ for $r > 0$.  Thus $\rank \tilde{F} = 2$ since by hypothesis $\tilde{F}^T \tilde{F} = \tilde{\ell}_{\nu_0}$.   Therefore, $\spn\{ \tilde{F} e_1, \tilde{F} e_2, \tilde{F} e_1 \times \tilde{F} e_2 \} = \mathbb{R}^3$.  Hence, (\ref{eq:rewritb}) simply rewrites $b \in \mathbb{R}^3$ equivalently in terms of $(\bar{b}_1,\bar{b}_2, \bar{b}_3) \in \mathbb{R}^3$.  The proof follows by explicitly verifying the formula.  
\end{proof}

\section{On the incompressibility of thin elastomers}\label{sec:IncompAppendix}
Here we prove Lemma \ref{IncompressibleLemma}, which develops (and catalogues properties associated to) explicit constructions of incompressible deformations for thin elastomer sheets.  
\begin{proof}[Proof of Lemma \ref{IncompressibleLemma}.]
We set $\delta_h = mh$ for $m\geq 1$ to be determined in Proposition \ref{ContractionProp} below.  We consider the function 
\begin{align}\label{eq:VhDef}
V_{\alpha}^h(\tilde{x},x_3) := y_{\alpha}^{\delta_h}(\tilde{x}) + x_3 b_{\alpha}^{\delta_h}(\tilde{x})
\end{align}
and assume $x_3 \in (-h/2, h/2)$.  Since $\nabla V_{\alpha}^{h}= (\tilde{\nabla} y_{\alpha}^{\delta_h}|b_{\alpha}^{\delta_h}) + x_3 (\tilde{\nabla}b_{\alpha}^{\delta_h}|0)$ and $\det (\tilde{\nabla} y_{\alpha}^{\delta_h} |b_{\alpha}^{\delta_h}) = 1$, we let $S_{\alpha}^{h} := (\tilde{\nabla} y_{\alpha}^{\delta_h} |b_{\alpha}^{\delta_h})^{-1} (\tilde{\nabla} b_{\alpha}^{\delta_h}|0)$ and find 
\begin{align}\label{eq:detValphah}
\det \nabla V_{\alpha}^h &= \det( (\tilde{\nabla} y_{\alpha}^{\delta_h} |b_{\alpha}^{\delta_h})^{-1} \nabla V_{\alpha}^{h})  = \det(I + x_3 S_{\alpha}^h) \nonumber \\
&=1 + x_3 \Tr(S_{\alpha}^{h}) + x_3^2 \Tr(\cof S_{\alpha}^{h}) + x_3^3 \det(S_{\alpha}^{h}). 
\end{align}
For the estimates below, $C \equiv C(M)$.  We note that $\|(\tilde{\nabla} y_{\alpha}^{\delta_h}|b_{\alpha}^{\delta_h})^{-1}\|_{L^{\infty}(\omega)} \leq C$ since the determinant is unity, and therefore $|S_{\alpha}^h| \leq C \delta_h^{\min\{-\alpha,0\}}$ by hypothesis and 
\begin{align}\label{eq:DetEst1}
|\det \nabla V_{\alpha}^{h} - 1| &\leq C \left( |x_3| \delta_h^{\min\{-\alpha,0\}} + |x_3|^2 \delta_h^{2\min\{-\alpha,0\}} + |x_3|^3 \delta_h^{3\min\{-\alpha,0\}}\right) \nonumber \\
& \leq  C|x_3| \delta_h^{\min\{-\alpha,0\}} \;\; \text{ for } \;\ \alpha \in \{-1,0,1\}, \;\; m \geq 1.
\end{align}
In addition for $\beta = 1,2$, since $\|\partial_{\beta} S_{\alpha}^{h}\|_{L^{\infty}(\omega)} \leq C(\delta_h^{2\min\{-\alpha,0\}} + \delta_h^{-\alpha-1}) \leq C \delta_h^{-\alpha-1}$ for $\alpha \in \{-1,0,1\}$ and since $|\partial_{\beta} \Tr (S_{\alpha}^h)| \leq |\partial_\beta S_{\alpha}^h|$, $|\partial_{\beta} \Tr(\cof S_{\alpha}^h)| \leq 2|S_{\alpha}^h||\partial_\beta S_{\alpha}^h|$ and $|\partial_{\beta} \det (S_{\alpha}^h)| \leq |S_{\alpha}^h|^2 |\partial_\beta S_{\alpha}^h|$, we conclude that 
\begin{align}\label{eq:DetEst2}
|\partial_{\beta} \det \nabla V_{\alpha}^h | &\leq C(|x_3|\delta_h^{-\alpha-1} + |x_3|^2\delta_h^{\min\{-\alpha,0\}}\delta_h^{-\alpha-1} + |x_3|^3 \delta_h^{2\min\{-\alpha,0\}}\delta_h^{-\alpha-1}) \nonumber \\
&\leq C|x_3| \delta_h^{-\alpha-1} \;\; 
 \text{ for} \;\; \beta \in \{1,2\}, \;\; \alpha \in \{-1,0,1\}, \;\; m \geq 1.
\end{align} 

Now since $V_{\alpha}^h$ is not incompressible, we modify it through a non-linear change in coordinates.  We let $\Xi^h(\tilde{x}, x_3) = (\tilde{x}, \xi^h(\tilde{x},x_3))$ for $\xi^h \in C^{1}(\overline{\Omega}_h,\mathbb{R})$ to be determined, and we define $Y_{\alpha}^h := V_{\alpha}^h \circ \Xi^h$.  Hence, by the column linearity of the determinant, we find that 
\begin{align*}
\det \nabla Y_{\alpha}^h = \det(\nabla V_{\alpha}^h \circ \Xi^h) \partial_3 \xi^h.
\end{align*}
Thus, satisfying the determinant constraint on $\nabla Y^h$ amounts to satisfying the ordinary differential equation 
\begin{align}\label{eq:ODE}
\partial_3 \xi^h = \frac{1}{\det (\nabla V_{\alpha}^h \circ \Xi^h)} \;\;\; \text{ on } \;\;\; \Omega_h
\end{align}
for some $\xi^h$.  There is an $m = m(\alpha,M) \geq 1$ such that for $h > 0$ sufficiently small, there is a solution to (\ref{eq:ODE}), i.e., $\xi^h \equiv \xi^h_{\alpha}$ for a $\xi_{\alpha}^h \in C^1(\overline{\Omega}_h,\mathbb{R})$ with the initial condition $\xi_{\alpha}^h(\tilde{x},0) = 0$, see Proposition \ref{ContractionProp}.  

It remains to prove the estimates in (\ref{eq:xihEst}).  By Proposition \ref{ContractionProp}, the map $\xi_{\alpha}^h$ satisfies pointwise
\begin{align}\label{eq:xihEst1}
|\xi_{\alpha}^h| \leq 2|x_3|, \;\;\; |\partial_3 \xi_{\alpha}^h| \leq 2  
\end{align}
everywhere on $\Omega_h$.  Thus, given (\ref{eq:ODE}),(\ref{eq:DetEst1}) and the estimates above, we deduce 
\begin{align*}
|\partial_3 \xi^h_{\alpha} - 1| \leq |\partial_3 \xi_{\alpha}^h||\det( \nabla V_{\alpha}^h \circ \Xi_{\alpha}^h) -1| \leq Ch^{\min\{-\alpha,0\}}|\xi_{\alpha}^h| \leq Ch^{\min\{-\alpha,0\}}|x_3|
\end{align*}
everywhere on $\Omega_h$.  Similarly, 
\begin{align*}
|\xi_{\alpha}^h - x_3| \leq | \int_{0}^{x_3}( \partial_3 \xi^h_{\alpha} - 1) d \bar{x}_3 | \leq \int_0^{|x_3|} |\partial_3 \xi^h_{\alpha} -1| d\bar{x}_3 \leq C h^{\min\{-\alpha,0\}}|x_3|^{2}
\end{align*}
everywhere on $\Omega_h$.   Finally, to estimate the first and second derivative of $\xi_h$, we define $F_h(\tilde{x},t) := \int_0^s \det(\nabla V_{\alpha}^h(\tilde{x},s)) ds$, and notice that the ordinary differential equation in (\ref{eq:ODE}) is equivalent to the implicit equation $F_h(\tilde{x},\xi_{\alpha}^h(x)) = x_3$.   Differentiating this equation with respect to $x_{\beta}$, $\beta = 1,$ or $2$, we find 
\begin{align*}
\int_{0}^{\xi_{\alpha}^h} \partial_{\beta} \det(\nabla V_{\alpha}^h) ds + \det(\nabla V_{\alpha}^h \circ \Xi_{\alpha}^h) \partial_\beta \xi^h_{\alpha} = 0.
\end{align*}
Hence using (\ref{eq:ODE}), (\ref{eq:DetEst2}) and (\ref{eq:xihEst1}), 
\begin{align*}
|\partial_\beta \xi_{\alpha}^h| \leq |\partial_3 \xi_{\alpha}^h| \int_0^{|\xi_{\alpha}^h|} |\partial_{\beta} \det \nabla V_{\alpha}^h| ds \leq C h^{-\alpha-1} \int_0^{|\xi_{\alpha}^h|}|s| ds \leq C h^{-\alpha-1}|x_3|^2
\end{align*}
everywhere on $\Omega_h$ for $\beta = 1,2$.  These are the desired estimates.  
\end{proof}
\begin{proposition}\label{ContractionProp}
Let $\alpha \in \{-1,0,1\}$.  Let $V_{\alpha}^h$ defined in (\ref{eq:VhDef}) with $y^{\delta_h}_{\alpha}$ and $b^{\delta_h}_{\alpha}$ as in Lemma \ref{IncompressibleLemma} with $\delta_h = mh$.  There is an $m = m(\alpha, M) \geq 1$ such that for any $h > 0$ sufficiently small, there exists a $\xi^h_{\alpha} \in C^{1}(\overline{\Omega}_h,\mathbb{R})$ such that 
\begin{align}\label{eq:ODECont}
\partial_3 \xi_{\alpha}^h = \frac{1}{\det(\nabla V_{\alpha}^h \circ \Xi_{\alpha}^h)} \quad \text{ on } \Omega_{h}, \quad \text{ with } \quad \xi_{\alpha}^h(\tilde{x},0) = 0.
\end{align}
Moreover $\xi_{\alpha}^h$ satisfies pointwise the estimate 
\begin{align}\label{eq:xihalphaEst}
|\xi_{\alpha}^h| \leq 2|x_3|, \quad |\partial_3 \xi_{\alpha}^h| \leq 2 \quad \text{ on } \Omega_{h}.
\end{align}
\end{proposition}
\begin{proof}
For $\alpha \in \{ -1,0,1\}$ and $h > 0$, we consider the mapping $T_{\alpha}^h \colon \mathcal{M}^{h}_{\alpha} \rightarrow C(\bar{\Omega}_{h})$ given by 
\begin{align*}
T_{\alpha}^h(\phi)(\tilde{x},x_3) = \int_{0}^{x_3}  \frac{1}{\det(\nabla V_{\alpha}^h(\tilde{x},\phi(\tilde{x},s)))} ds \quad \text{ for each } \;\; (\tilde{x}, x_3) \in \Omega_{h},
\end{align*}
where $\mathcal{M}^{h}_{\alpha}$ is given by
\begin{align*}
\mathcal{M}^{h}_{\alpha} &:= \left\{ \phi \in C(\overline{\Omega}_{h}) \colon \phi(\tilde{x},0)= 0, \;\; |\phi(\tilde{x},x_3)| \leq 2|x_3|,\right.\\
&\;\;\;\;\;\;\;\;\;\;\;\;\;\;\;\;\;\;\;\;\;\;\;\;\;\;\;\; \left. \det(\nabla V_{\alpha}^h(\tilde{x},\phi(\tilde{x},x_3))) \geq 1/2 \text{ for each } (\tilde{x},x_3) \in \Omega_{h} \right\}.  
\end{align*}
This is a (non-empty) complete space under the infinity norm.  Thus, we aim to show that there is an appropriate choice of $m = m(\alpha,M)$ in $\delta_h$ such that for $h >0$ sufficiently small, the mapping $T_{\alpha}^h$ is, in fact, a contraction map in the space $\mathcal{M}^h_{\alpha}$ under the infinity norm.  The proposition will follow by the equivalence of the integral representation of (\ref{eq:ODECont}).

We first prove that $T_{\alpha}^h$ is an operator (i.e., $T_{\alpha}^h \colon \mathcal{M}^h_{\alpha} \rightarrow \mathcal{M}^h_{\alpha}$) for an appropriate choice of $m = m(\alpha,M)$ and small enough $h$.  For the estimates below, $C \equiv C(M)$.  Since $\phi \in \mathcal{M}^h_{\alpha}$, we have 
\begin{align*}
|T_{\alpha}^h(\phi)(\tilde{x},x_3)| \leq 2|x_3|, \quad \text{ for each } \;\; (\tilde{x},x_3) \in \Omega_{h}.
\end{align*}
In addition, using a similar estimate to (\ref{eq:DetEst1}), we obtain
\begin{align*}
|\det \nabla V_{\alpha}^h(\tilde{x},T_{\alpha}^h(\tilde{x},x_3))  -1| \leq C |T_{\alpha}^h(\tilde{x},x_3)|\delta_h^{\min\{-\alpha, 0\}} \leq C|x_3|h^{\min\{-\alpha,0\}}m^{\min\{-\alpha,0\}}
\end{align*}
for $\alpha \in \{ -1,0,1\}$.  Thus, for $\alpha \in \{-1,0\}$, we need only enforce $m \geq 1$ and for $\alpha = 1$ we enforce $m = m(\alpha, M) \geq \max\{ 2C,1\}$ to ensure $T_{\alpha}^h$ is an operator for small $h$.  
 
It remains to prove that $T^h_{\alpha}$ is a contraction under the $L^{\infty}$ norm.  Observe for $\phi, \psi \in \mathcal{M}_{\alpha}^h$,
\begin{align*}
|T^h_{\alpha}(\phi)(\tilde{x},x_3) - T^h_{\alpha}(\psi)(\tilde{x},x_3)| &\leq 4\int_{0}^{|x_3|}| \det ( \nabla V^h_{\alpha} (\tilde{x}, \psi(\tilde{x},s)) - \det(\nabla V^h_{\alpha}(\tilde{x}, \phi(\tilde{x},s))| ds\\
&\leq C \delta_h^{\min\{-\alpha,0\}} \int_{0}^{|x_3|} |\psi(\tilde{x},s) - \phi(\tilde{x},s)| ds  \\
&\leq C m^{\min\{-\alpha,0\}} h^{\min\{-\alpha,0\}} h \|\psi - \phi\|_{L^{\infty}(\Omega_h)}
\end{align*}
for any $(\tilde{x},x_3) \in \Omega_h$.  Here the first inequality uses the determinant constraint on $\mathcal{M}_{\alpha}^h$, the second uses the equation (\ref{eq:detValphah}), and the third uses that $\delta_h = mh$.  Finally, from this estimate, it is clear that we can choose $m = m(\alpha,M) \geq 1$ independent of $h$ (in fact $m = 1$ suffices for $\alpha = -1,0$ as in the remark), such that for $h$ sufficiently small
\begin{align*}
\|T_{\alpha}^h(\phi) - T_{\alpha}^h(\psi)\|_{L^{\infty}(\Omega_h)} < \| \psi -\phi\|_{L^{\infty}(\Omega_h)},
\end{align*}
i.e., it is a contraction map.  

We now fix this $m = m(\alpha,M)$ and an $h>0$ sufficiently small.  Since $T^h_{\alpha}$ is a contraction map, there exists a $\xi_{\alpha}^h$ such that 
\begin{align*}
\xi_{\alpha}^h = T_{\alpha}^h(\xi_{\alpha}^h) = \int_{0}^{x_3}  \frac{1}{\det(\nabla V_{\alpha}^h(\tilde{x},\xi_{\alpha}^h(\tilde{x},s)))} ds \quad \text{ for each } \;\; (\tilde{x}, x_3) \in \Omega_{h}.
\end{align*}
This is equivalent to the ordinary differential equation (\ref{eq:ODECont}). The regularity $\xi_{\alpha}^h \in C^{1}(\bar{\Omega}_h,\mathbb{R})$ follows from the regularity of $y_{\alpha}^{\delta_h}$ and $b_{\alpha}^{\delta_h}$.  The estimates (\ref{eq:xihalphaEst}) follow from the fact that $\xi_h^{\alpha} \in \mathcal{M}_{\alpha}^h$. This completes the proof.  
\end{proof}

\section{Geometric rigidity and nematic elastomers}\label{sec:CompactAppendix}

First, we derive the key estimate which relates geometric rigidity \cite{fjm_cpam_02} to the setting of nematic elastomers.  

\begin{proposition}\label{GeomRigidPropAppend}
Let $\omega \subset \mathbb{R}^3$ bounded and Lipschitz.  There exists a constant $C = C(r_0,r_f,\tau)$ with the following property: for all $h>0$, $Q_{\tilde{x}^{\ast},h} := (-h/2,h/2)^3 \subset \Omega_h$, $V^h \in W^{1,2}(\Omega_h, \mathbb{R}^3)$, $N^h \in W^{1,2}(\Omega_h, \mathbb{S}^2)$ and $N_0^h$ as in (\ref{eq:n0tomidn0}) with $n_0 \in W^{1,2}(\omega,\mathbb{R}^3)$, there exists an associated constant rotation $R_{\tilde{x}^{\ast}}^h \in SO(3)$ such that 
\begin{align*}
\int_{Q_{\tilde{x}^{\ast},h}} &|(\ell^f_{N^h})^{-1/2} \nabla Y^h (\ell^0_{N_0^h})^{1/2} - R_{\tilde{x}^{\ast}}^h|^2 dx\\
 &\leq C \int_{Q_{\tilde{x}^{\ast},h}} \left( \dist^2((\ell^f_{N^h})^{-1/2} \nabla Y^h (\ell^0_{N_0^h})^{1/2}, SO(3)) + h^2(|\nabla N^h|^2 + |\tilde{\nabla} n_0|^2 + 1) \right) dx
\end{align*}
\end{proposition}
\begin{proof}
Let $Y^h \in W^{1,2}(\Omega_h,\mathbb{R}^3)$, $N^h \in W^{1,2}(\Omega_h, \mathbb{S}^2)$ and  $n_0 \in W^{1,2}(\omega,\mathbb{S}^2)$ with $N_0^h$ as in (\ref{eq:n0tomidn0}).  we fix $\tilde{x}^{\ast}$ such that $Q_{\tilde{x}^{\ast}, h} \subset \Omega_h$ and set 
\begin{align}\label{eq:AveragesPoincare}
A_h^f := \frac{1}{|Q_{\tilde{x}^{\ast},h}|} \int_{Q_{\tilde{x}^{\ast},h}} (\ell_{N^h}^f)^{1/2} dx, \quad A^0_h := \frac{1}{|Q_{\tilde{x}^{\ast},h}|} \int_{Q_{\tilde{x}^{\ast},h}} (\ell_{n_0}^0)^{-1/2} dx.
\end{align}
Because of the structure of the step-length tensors, these averages are positive definite, and each of the eigenvalues lives in a compact set of the positive real numbers depending only on $r_f$ and $r_0$ (in particular, this set does not depend on $h$).  Hence, these linear maps belong to a family of $h$-indepdent Bilipschtiz maps with controlled Lipschitz constant, and so we write $A^f \equiv A^f_h$ and $A^0 \equiv A^0_h$ in sequel.

Now, we set 
\begin{align*}
V^h(s) =  (A^f)^{-1} Y^h( (A^0)^{-1} s ), \quad s \in (A^0) Q_{\tilde{x}^{\ast},h}.
\end{align*}
We observe that $V^h \in W^{1,2}((A^0) Q_{\tilde{x},h},\mathbb{R}^3)$ by the regularity of $Y^h$.  Therefore by geometric rigidity \cite{fjm_cpam_02}, there exists a constant rotation $R_{\tilde{x}}^h \in SO(3)$ such that 
\begin{align}\label{eq:GREst1}
\int_{Q_{\tilde{x}^{\ast},h}} &|(A^f)^{-1} \nabla Y^h(x) (A^0)^{-1} - R_{\tilde{x}^{\ast}}^h|^2 dx = |\det A^0|^{-1} \int_{(A^0) Q_{\tilde{x}^{\ast},h}} |\nabla V^h(s) - R_{\tilde{x}^{\ast}}^h|^2 ds  \nonumber \\
&\quad \quad \leq |\det A^0|^{-1} C((A^0)Q_{\tilde{x}^{\ast},h})\int_{(A_{\omega}^0)Q_{\tilde{x}^{\ast},h}} \dist^2(\nabla V^h(s), SO(3)) ds  \nonumber \\
&\quad \quad =  C((A^0)Q_{\tilde{x}^{\ast},h})\int_{Q_{\tilde{x}^{\ast},h}} \dist^2((A^f)^{-1} \nabla Y^h(x) (A^0)^{-1}, SO(3)) dx
\end{align}
The constant $C((A^0)Q_{\tilde{x}^{\ast},h})$ can be chosen uniformly for a family of domains which are Bilipschitz equivalent with controlled Lipschitz constant.  Hence, actually we can choose $C(r_0,Q_{\tilde{x}^{\ast},h}) \geq  C((A^0)Q_{\tilde{x}^{\ast},h})$.  Moreover, the constant is invariant under translation and dilatation.  Hence, actually we have $C(r_0,Q_{\tilde{x}^{\ast},h}) = C(r_0)$ for any $Q_{\tilde{x}^{\ast},h} \subset \Omega_h$.   These properties are given in Friesecke, James and M\"{u}ller, Theorem 9 \cite{fjm_arma_06}.  Since $r_0$ is fixed in this calculation, we write $C(r_0) \equiv C$, and thus
\begin{align}\label{eq:GeomKeyEst}
\int_{Q_{\tilde{x}^{\ast},h}} |(A^f)^{-1} \nabla Y^h(x) (A^0)^{-1} - R_{\tilde{x}^{\ast}}^h|^2 dx \leq C\int_{Q_{\tilde{x}^{\ast},h}} \dist^2((A^f)^{-1} \nabla Y^h(x) (A^0)^{-1}, SO(3)) dx
\end{align}
from (\ref{eq:GREst1}).  

Since we will no longer be dealing with a change of variables in this proof, we now drop the explicit dependence on $x$ inside the integrals.  We observe by the key estimate (\ref{eq:GeomKeyEst}) that
\begin{align}\label{eq:GREst2}
\int_{Q_{\tilde{x}^{\ast},h}}& |\nabla Y^h - (A^f) R_{\tilde{x}^{\ast}}^h(A^0)|^2 dx \leq C \int_{Q_{\tilde{x}^{\ast},h}} |(A^f)^{-1} \nabla Y^h (A^0)^{-1} - R_{\tilde{x}^{\ast}}^h|^2 dx  \nonumber \\
&\leq  C\int_{Q_{\tilde{x}^{\ast},h}} \dist^2((A^f)^{-1} \nabla Y^h (A^0)^{-1}, SO(3)) dx \leq C \int_{Q_{\tilde{x}^{\ast},h}} \dist^2(\nabla Y^h, (A^f) SO(3) (A^0)) dx  \nonumber \\
&\leq C \int_{Q_{\tilde{x}^{\ast},h}} \left( \dist^2(\nabla Y^h, (\ell_{N^h}^f)^{1/2} SO(3) (\ell_{N^h_0}^0)^{-1/2}) + |(\ell_{N^h}^f)^{1/2} - A^f|^2 + |(\ell_{N_0^h}^0)^{-1/2} - A^0|^2\right)dx \nonumber  \\
&\leq C \int_{Q_{\tilde{x}^{\ast},h}} \Big( \dist^2((\ell_{N^h}^f)^{-1/2} \nabla Y^h (\ell_{N^h_0}^0)^{1/2},  SO(3)) +  h^2 |\nabla N^h|^2 \nonumber \\
&\quad \quad \quad \quad \quad \quad \quad \quad   + |(\ell_{n_0}^0)^{-1/2} - A^0|^2 + |(\ell_{N_0^h}^0)^{-1/2} -(\ell_{n_0}^0)^{-1/2}|^2  \Big) dx \nonumber \\
& \leq  C \int_{Q_{\tilde{x}^{\ast},h}} \left( \dist^2((\ell_{N^h}^f)^{-1/2} \nabla Y^h (\ell_{N^h_0}^0)^{1/2},  SO(3)) +  h^2( |\nabla N^h|^2 + |\tilde{\nabla} n_0|^2 + 1)  \right) dx
\end{align}
Here, the constant $C = C(r_0,r_f,\tau)$ is due to several applications of the triangle inequality and the fact that the norm of the step-length tensors, inverses and averages are compact and this depends only on $r_f$,$r_0$.  We have also applied the standard Poincar\'{e} inequality given the averages (\ref{eq:AveragesPoincare}), and used that the diameter of $Q_{\tilde{x}^{\ast},h}$ is $h$ and that the gradients of the step-length tensors are controlled by the gradients of the directors.  Finally, from the assumed control of non-idealities for $N_0^h$ in (\ref{eq:n0tomidn0}), we have the estimate $\| (\ell_{N_0^h}^0)^{-1/2} - (\ell_{n_0}^0)^{-1/2}\|_{L^{\infty}} \leq c(r_0)\tau h$.  This gives the dependence on $\tau$ in the constant.   

Now using (\ref{eq:GREst2}), we find that
\begin{align*}
\int_{Q_{\tilde{x}^{\ast},h}} & |(\ell_{N^h}^f)^{-1/2} \nabla Y^h (\ell_{N_0^h}^0)^{1/2} - R_{\tilde{x}^{\ast}}^h|^2 dx \leq C \int_{Q_{\tilde{x}^{\ast},h}} |\nabla Y^h - (\ell_{N^h}^f)^{1/2} R_{\tilde{x}^{\ast}}^h (\ell_{N_0^h}^0)^{-1/2}|^2 dx \\
&\leq C \int_{Q_{\tilde{x}^{\ast},h} } \left(|\nabla Y^h - (\ell_{N^h}^f)^{1/2} R_{\tilde{x}^{\ast}}^h (\ell_{n_0}^0)^{-1/2}|^2 + |(\ell_{N_0^h}^0)^{-1/2} - (\ell_{n_0}^0)^{-1/2}|^2\right) dx \\
&\leq C \int_{Q_{\tilde{x}^{\ast},h} }\left(|\nabla Y^h - (A^f) R_{\tilde{x}^{\ast}}^h(A^0)|^2 + h^2 + |(\ell_{N^h}^f)^{1/2} - A^f|^2 + |(\ell_{n_0}^0)^{-1/2} - A^0|^2\right) dx \\
&\leq C \int_{Q_{\tilde{x}^{\ast},h}} \left( \dist^2((\ell_{N^h}^f)^{-1/2} \nabla Y^h (\ell_{n_0}^0)^{1/2},  SO(3)) +  h^2( |\nabla N^h|^2 + |\tilde{\nabla} n_0|^2 +1)  \right) dx
\end{align*}
as desired.
\end{proof}

Now we note that the approximations in Lemma \ref{GeomRigidityLemma} are not new.  They essentially follow from the same argument as that of Theorem 10 in Friesecke et al.\ \cite{fjm_arma_06}, modified appropriately for nematic elastomers using the estimate in Proposition \ref{GeomRigidPropAppend}.  In the general context of non-Euclidean plates, there is a recent body of literature on such estimates (e.g., Lewicka and Pakzad  (Lemma 4.1) \cite{lp_esaim_11} and Lewicka et al.\ (Theorem 1.6) \cite{lmp_prsa_11},  (Lemma 2.3) \cite{bls_arma_15}). Thus briefly:
\begin{proof}[Proof of Lemma \ref{GeomRigidityLemma}]
We repeat steps 1-3 in the proof of Theorem 10 in \cite{fjm_arma_06} with some modification due to our nematic elastomer setting.  The lemma follows by the estimate in Proposition \ref{GeomRigidPropAppend}.  
\end{proof}  

\bibliographystyle{abbrv}
\bibliography{plbMath}

\end{document}